%% file: Main.tex
\documentclass[12pt,a4paper]{article}

\usepackage{amsthm}   % added,must be in first ones o.w. error: openbox is defined.
\usepackage{amsmath}
\usepackage{graphicx,psfrag,epsf}
\usepackage{enumerate}
\usepackage{natbib}
\usepackage{url} % not crucial - just used below for the URL 

\usepackage{mathtools}
\usepackage{mathrsfs}
\usepackage{txfonts}
\usepackage{verbatim}
\usepackage{extarrows}
\usepackage{slashbox}
\usepackage{bbm}      % for indicator function

\usepackage[noend]{algpseudocode}
\usepackage{algorithm}

\algnewcommand{\IfThenElse}[3]{% \IfThenElse{<if>}{<then>}{<else>}
	\State \algorithmicif\ #1\ \algorithmicthen\ #2\ \algorithmicelse\ #3}

\algnewcommand{\StateFor}[2]{% \IfThenElse{<if>}{<then>}{<else>}
	\State  #1\ \algorithmicfor\ #2}

\usepackage{thmtools, thm-restate}
\usepackage{hyperref}

\newcommand{\rgiven}{\, \right| \,}            % must use \left in front
\newcommand{\given}{\, | \,}                     % | with space 

\newcommand{\bbone}{\mathbbm{1}}

\newcommand{\bbE}{\mathbb{E}}

\newcommand{\bbR}{\mathbb{R}}

\newcommand{\bA}{\boldsymbol{A}}

\newcommand{\bI}{\boldsymbol{I}}  % i
\newcommand{\bX}{\boldsymbol{X}}
\newcommand{\by}{\boldsymbol{y}}

\newcommand{\bz}{\boldsymbol{z}}

\newcommand{\mN}{\mathcal{N}}

\newcommand{\mO}{\mathcal{O}}
\newcommand{\bmO}{\boldsymbol{\mathcal{O}}}
\newcommand{\mP}{\mathcal{P}}
\newcommand{\bmP}{\boldsymbol{\mathcal{P}}}
\newcommand{\mR}{\mathcal{R}}
\newcommand{\bmR}{\boldsymbol{\mathcal{R}}}

\DeclareMathOperator{\diag}{diag}
\DeclareMathOperator{\argmin}{argmin}
\DeclareMathOperator{\argmax}{argmax}
\DeclareMathOperator{\sign}{sign}
\DeclareMathOperator{\AND}{AND}
\DeclareMathOperator{\OR}{OR}
\DeclareMathOperator{\equi}{equi}
\DeclareMathOperator{\sdp}{sdp}
\DeclareMathOperator{\re}{re}
\DeclareMathOperator{\swap}{swap}
\DeclareMathOperator{\FDR}{FDR}
\DeclareMathOperator{\mFDR}{mFDR}
\DeclareMathOperator{\FDP}{FDP}
\DeclareMathOperator{\TPP}{TPP}

\newtheorem{proposition}{Proposition}

\newtheorem{lemma}{Lemma}[section]

\newtheorem{definition}{Definition}[section]
\newtheorem{remark}{Remark}[section]

\title{GGM knockoff filter: False Discovery Rate Control for Gaussian Graphical Models}
\author{Jinzhou Li \\
	Seminar f{\"u}r Statistik, ETH Z{\"u}rich, Switzerland\\
	and \\
	Marloes H. Maathuis \\
	Seminar f{\"u}r Statistik, ETH Z{\"u}rich, Switzerland}
\date{\today}
\begin{document}

\maketitle

%\keywords{false discovery rate (FDR), Gaussian graphical model (GGM), knockoff framework, sample splitting and recycling, structure learning}
%\tableofcontents

\begin{abstract}
We propose a new method to learn the structure of a Gaussian graphical model with finite sample false discovery rate control. Our method builds on the knockoff framework of Barber and Cand\`{e}s for linear models. We extend their approach to the graphical model setting by using a local (node-based) and a global (graph-based) step: we construct knockoffs and feature statistics for each node locally, and then solve a global optimization problem to determine a threshold for each node. We then estimate the neighborhood of each node, by comparing its feature statistics to its threshold, resulting in our graph estimate. 
Our proposed method is very flexible, in the sense that there is freedom in the choice of knockoffs, feature statistics, and the way in which the final graph estimate is obtained. For any given data set, it is not clear a priori what choices of these hyperparameters are optimal. We therefore use a sample-splitting-recycling procedure that first uses half of the samples to select the hyperparameters, and then learns the graph using all samples, in such a way that the finite sample FDR control still holds. We compare our method to several competitors in simulations and on a real data set.
\end{abstract}

%\tableofcontents

%%%%%%%%%%%%%%%%%%%%%%%%%%%%%%%%%%%%%%%%%%%%%%%% Main Content %%%%%%%%%%%%%%%%%%%%%%%%%%%%%%%%%%%%%%%%%%%%%%%%%%%%
\input{Sections/Sec1-Introduction}

\input{Sections/Sec2-Preliminaries}
\input{Sections/Sec3-GGMKnockoffFilter}
\input{Sections/Sec4-SampleSplittingRecycling}
\input{Sections/Sec5-Simulations}
\input{Sections/Sec6-Discussion}
\input{Sections/LastParts}

% the bibliography
\nocite{*}         % otherwise will not show the uncited ones
\bibliographystyle{chicago}
\bibliography{Sections/Bibliography}

%%% Supplement
\addtocontents{toc}{\vspace{.5\baselineskip}}
\addtocontents{toc}{\protect\setcounter{tocdepth}{1}}
\appendix
\newpage
\input{Sections/Supplements}

\end{document}

%% file: Sections/Sec1-Introduction.tex
\section{Introduction}\label{sec:intro}

Gaussian graphical models are used to model multivariate distributions in many scientific fields, including biology, economics, health science, social science
\citep[e.g.,][]{lafit2019partial, shin2014atlas, giudici2016graphical, ahmed2009recovering, kalisch2010understanding}.
Formally, the model is defined as follows. To a multivariate Gaussian random vector $X= (X_1, \dots, X_p)^T$ with mean $\mu$ and covariance matrix $\Omega^{-1}$, we assign an undirected graph $G=(V,E)$, with node set $V = [p] = \{ 1,\dots,p \}$ and edge set $E = \left\{ (i,j) \in [p]^2: \left. X_i \not\!\perp\!\!\!\perp X_j \rgiven X_{ [p] \backslash \{i,j\} } \text{ and } i \neq j  \right\} $, where for any set $H \subseteq [p]$, $X_{ H } = (X_i : i\in H)$ and $ \left. X_i \not\!\perp\!\!\!\perp X_j \rgiven X_{ H }$ indicates that $X_i$ and $X_j$ are dependent given  $X_{ H }$. We refer to \cite{HandbookGraphicalModels19} for an overview of recent developments in graphical models.  

Structure learning of Gaussian graphical models aims to obtain an estimator $\widehat{E}$ of the true edge set $E$, based on $n$ independent and identically distributed (i.i.d.) observations of $X$ \citep[see, e.g.,][]{drton2017structure}.
To obtain a reliable outcome and alleviate reproducibility issues \citep[see, e.g.,][]{ioannidis2005most, begley2012drug, baker20161}, it is important to control the number of false positive edges.
In this paper, we propose a structure learning approach that controls the finite sample false discovery rate (FDR) as defined by \cite{benjamini1995controlling}. That is, our estimated edge set $\widehat{E}$ satisfies 
\begin{align}\label{FDRgraphGoal}
\FDR = \bbE \left( \frac{ | \widehat{F} | }{ | \widehat{E} | \vee 1} \right) \leq q,
\end{align}
where $q \in [0,1]$ is a pre-set nominal FDR level, $\widehat{F} = \widehat{E} \backslash E$ is the set of the falsely discovered edges,  $|\widehat{E}| \vee 1 = \max \left( |\widehat{E}|, 1 \right) $, and $|\cdot|$ is the cardinality of a set (where in the case of an edge set, we count $(i,j)$ and $(j,i)$ as a single edge). The quantity inside the expectation in \eqref{FDRgraphGoal} is also known as the false discovery proportion (FDP).

%{ \color{red} (Will delete: Previous research on structure learning of Gaussian graphical models has mainly focused on handling high-dimensional settings \citep[e.g.,][]{meinshausen2006high, yuan2007model, friedman2008sparse, d2008first, rothman2008sparse, raskutti2009model,  fan2009network, cai2011constrained, drton2017structure}. )}

Research towards FDR control in this context includes works based on multiple testing approaches. In low-dimensional cases, \cite{drton2007multiple} suggested testing pairwise partial correlations. After obtaining the corresponding p-values, the Benjamini-Yekutieli procedure \citep{benjamini2001control} can be applied to recover the graph structure with finite sample FDR control without any additional assumptions. The Benjamini-Hochberg procedure \citep{benjamini1995controlling} is also commonly used to control the finite sample FDR based on p-values, but it requires the so-called ``Positive Regression Dependency on each one from a Subset" (PRDS) assumption \citep{benjamini2001control}. In high-dimensional settings, \cite{liu2013gaussian} proposed a structure learning algorithm based on a certain test statistic and its asymptotic distribution, providing \emph{asymptotic} FDR control under some regularity conditions. \cite{lee2019structure} used sorted $l_1$-norm regularization in a node-wise manner to obtain an estimated graph with finite sample \emph{node-wise} FDR control under some assumptions. In Appendix \ref{appendix: NodeGraphWiseFDRexample}, we give a simple example illustrating that node-wise and graph-wise FDR control can be very different.

Recently, \cite{barber2015controlling} developed an interesting and innovative framework, called fixed-X knockoffs. This framework was originally designed for low-dimensional Gaussian linear models with fixed design, and it achieves finite sample FDR control for variable selection without resorting to p-values. The underlying idea is to construct artificial variables (knockoffs) that act as negative controls and mimic the correlation structure of the original predictors. The fixed-X knockoff framework has been extended in many ways. For example, \cite{dai2016knockoff} generalized it to group-variable selection and multi-task learning, and  \cite{janson2016familywise} adapted it to control the $k$-familywise error rate. \cite{barber2016knockoff} extended it to high-dimensional linear models through a sample-splitting-recycling procedure and developed novel theory showing that the directional FDR can also be controlled. \cite{katsevich2017multilayer} proposed a modification for multilayer controlled variable selection tasks and showed that this procedure controls the FDR at both the individual-variable and the group-variable level. 
%{ \color{red} \cite{weinstein2017power} studied the power of the knockoff approach under an i.i.d.\ Gaussian design with Lasso statistics. (Delete this? otherwise in the end we cite again) }

Another closely related line of work is based on the so-called model-X knockoff framework proposed by \cite{candes2016panning}. Here the distribution of the dependent variable given the predictors is fully flexible, but the joint distribution of the predictors is assumed to be known. 
\cite{huang2020relaxing} relaxed this assumption by conditioning on a sufficient statistic. In particular, they proved that one can construct valid model-X knockoffs without knowing the joint distribution of the covariates in low-dimensional multivariate Gaussian settings. Their proposed model-X knockoffs can be used in our method, but in this paper we will restrict ourselves to the fixed-X knockoff framework of \cite{barber2015controlling} for simplicity. 

%As we will see later, the fixed-X framework is most suitable for Gaussian graphical model learning.  

%Unlike fixed-X knockoffs which are constructed algebraically, model-X knockoffs are generated in a probabilistic manner, when the joint distribution of the predictors is assumed to be fully known. Since 
%
%The model-X knockoff procedure ensures finite sample FDR control for low-dimensional, high-dimensional and non-linear settings. 
%
%One of the most important issues in this framework is to find an efficient way to generate knockoffs. To achieve such efficiency, \cite{sesia2018gene}, \cite{gimenez2018knockoffs} and \cite{bates2019metropolized} proposed different knockoff generation methods when the predictors follow hidden Markov models, Bayesian networks and graphical models, respectively. \cite{jordon2018knockoffgan}, \cite{romano2018deep} and \cite{liu2018auto} considered the problem of generating approximate knockoffs when the distribution of the predictors is unknown. On the theory side, \cite{barber2018robust} and \cite{fan2019rank} investigated the robustness of the model-X knockoff procedure. Moreover, \cite{fan2019rank} also discussed the power of this method when the distribution of the predictors is characterized by a Gaussian graphical model. 

The knockoff idea is starting to be used in the context of Gaussian graphical models. \cite{zheng2018recovering} used it for Gaussian graphical model learning with \emph{node-wise} FDR control. \cite{YuKaufmannLederer19} claim graph-wise FDR control for Gaussian graphical models using the knockoff idea, but there is an issue with their proof. Therefore, to the best of our knowledge, the problem of using the knockoff framework to obtain finite sample graph-wise FDR control for Gaussian graphical models (see \eqref{FDRgraphGoal}) is still open. 

In this paper, we contribute to this problem in two ways. First, we propose the GGM knockoff filter. This method is based on the knockoff framework and consists of a local (node-based) and a global (graph-based) step: we construct knockoffs and feature statistics for each node locally (see \eqref{relationOfGGMandLM}), and then solve a global optimization problem to determine a threshold for each node (see \eqref{ANDopt} and \eqref{ORopt}). We then estimate the neighborhood of each node by comparing its feature statistics to its threshold, resulting in a graph estimate. We prove that this procedure provides finite sample graph-wise FDR control.

As discussed in \cite{barber2015controlling} and \cite{candes2016panning}, the knockoff framework is very flexible, since there are many valid ways to construct knockoffs and feature statistics, and in our case, the parameters in the optimization problem and two different ways to obtain the final graph estimate. We refer to these choices as hyperparameters of the method. Each choice of the hyperparameters leads to a valid FDR control procedure. The power of these procedures can vary widely depending on the underlying distribution of the data. 
Therefore, the following practical question arises: 
how should we choose the hyperparameters for a given data set to ensure good power while maintaining FDR control? A simple and tempting approach could be to implement many different procedures and report the procedure that yields the most discoveries. It should be clear, however, that such an approach loses FDR control. 

%Hence, we need a data driven method to choose the hyperparameters suitably, while maintaining FDR control. We emphasize that this problem of hyperparameter selection is not restricted to the structure learning task we consider here, and occurs more generally in knockoff based approaches. 

%{ \color{red} One can even ask a more broader question about how to choose FDR control procedures (including knockoff based and p-value based methods) to have a good power while still maintaining the FDR control property for a given data set. We leave this question as future research.}

%Our second contribution is a solution to this problem, based on the sample-splitting-recycling procedure of \cite{barber2016knockoff}. This procedure was introduced to handle high-dimensional settings. They applied sample splitting and used the first part of the sample to reduce the dimensionality. In a plain sample-splitting procedure, one would select variables based on the second half of the sample. To improve power, the sample-splitting-recycling procedure uses the entire sample, but in such a way that the FDR control is still ensured. We will use this approach for our problem: we use the first part of the sample for hyperparameter selection, and then apply the nodewise knockoff method to the entire sample, in such a way that finite sample FDR control is maintained. The issue of randomness has not been resolved. We therefore
%recommend applying the approach several times with different random splits, to investigate the variation in output due to the different splits. 

Our second contribution is a partial solution to this problem. First, we note that the issue can be solved by sample-splitting, where one part of the sample is used to choose the hyperparameters, and the procedure with the chosen hyperparameters is then applied to the other part of the sample. This approach has two disadvantages: (i) the outcome is random, as it depends on the split, and (ii) we suffer a loss of power, since we only apply the chosen procedure to the second part of the sample. 

We can alleviate the latter issue by adapting the sample-splitting-recycling procedure of \cite{barber2016knockoff}. This approach was originally introduced for a somewhat different problem, namely for handling high-dimensional settings. They applied sample splitting and used the first part of the sample to reduce the dimensionality. Then, instead of selecting variables in the low-dimensional problem based only on the second half of the sample, they used the entire sample, in such a way that they still ensured FDR control. This method tends to improve the power when compared to a simple splitting method, and we adapt it to the problem of hyperparameter selection. 

Coming back to the first issue of randomness, we note that the paper by \cite{gimenez2018improving} considers a similar problem: they aim to alleviate randomness induced by sampling model-X knockoffs, in order to obtain a more stable and powerful result. Their approach, however, is tailored to the model-X knockoff setting and cannot directly be applied to our case where the randomness is induced by sample splitting. We were not able to solve this issue and recommend applying the sample-splitting-recycling procedure several times with different random splits, to assess the variation in the resulting output.

The remainder of the paper is organized as follows. 
In Section \ref{sec:prelim}, we review the fixed-X knockoff framework for linear models. 
In Section \ref{sec: GGM knockoff filter}, we introduce the GGM knockoff filter and prove its finite sample FDR control property.  
In Section \ref{sec: sample-splitting-recycling}, we present the sample-splitting-recycling approach for GGM knockoff filter.
The numerical performance of the proposed procedure is evaluated in simulation studies and on a real data set in Section \ref{sec:simu}. 
We close the paper with potential directions for future research in Section \ref{sec:disc}.

%% file: Sections/Sec2-Preliminaries.tex
\section{Recap: The fixed-X knockoff framework for linear models} \label{sec:prelim}

We consider the linear model
\begin{equation}\label{eq: linear model}
\by = \bX \beta + \bz,
\end{equation}
where $\by \in \bbR^n$, $\bX \in \bbR^{n \times p}$, $\beta \in \bbR^p$ and $\bz \in \bbR^n$, with $n\ge 2p$.\footnote{The method can also be used when $p \leq n<2p$ with some small adaptations, see \cite{barber2015controlling} for details.} Here $\bX$ is a fixed and known design matrix. We assume that $ (\bX^T \bX )^{-1}$ exists and that $\bX$ is standardized, in the sense that all columns of $\bX$ have length $1$. The vector $\bz$ is a noise vector with multivariate Gaussian distribution $\mN_n(0, \sigma^2 I)$. 
We denote the index set of the non-null variables by $S = \{ i\in[p]: \beta_i \neq 0 \}$ and the index set of the null variables by $S^c = \{ i\in[p]: \beta_i = 0 \}$. The goal is to estimate $S$ with FDR control.
 
\begin{algorithm}[t]
	\caption[] {\textbf{: Knockoff procedure for linear models \citep[][]{barber2015controlling}} } \label{oriKnock}
		
	\textbf{Input}: $(\bX, \by, q, \delta, \mO, \mP)$, where $\bX \in \bbR^{n \times p}$ is the design matrix, $\by \in \bbR^n$ is the response vector, $q \in [0,1]$ is the nominal FDR level, $\delta \in \{0,1\}$ indicates FDR control $(\delta=1)$ or mFDR control $(\delta=0)$, $\mO \in \{ \mO_{\equi}, \mO_{\sdp} \}$ is used to construct knockoffs (specifically, it is the optimization strategy used to compute $s$), and $\mP$ is the procedure used to construct feature statistics.
	
	\textbf{Output}: Estimated set of the non-null variables $\widehat{S}$.
	\begin{algorithmic}[1]
		\State \textbf{Step 1.} Construct the knockoff matrix $\widetilde{\bX} \in \bbR^{n \times p}$ by
		\begin{equation} \label{ConstructKnockoffs}
			\widetilde{\bX}(\bX, \mO) = \bX \left( I - (\bX^T\bX)^{-1} \diag\{ s \} \right) + \widetilde{U} C.
		\end{equation} 
		\State \textbf{Step 2.} Create feature statistics $W(\bX, \widetilde{\bX}, \by, \mP ) = \mP(\bX, \widetilde{\bX}, \by) = (W_1, \dots, W_p) $.
		\State \textbf{Step 3.} Calculate the positive threshold $\widehat{T}$ by
		\begin{equation}\label{OriKnockThreEquation}
			\widehat{T}(W, q, \delta) = \min \left\{ T \in \{|W_i|: i \in [p] \} \backslash \{ 0 \} : \frac{ \delta + \left| \{ i \in [p]: W_i \leq -T \} \right| } { \left| \{  i \in [p]: W_i \geq T \} \right| \vee 1 }  \leq q \right\}  %> 0
		\end{equation}
		and set $\widehat{T}=+\infty$ if the above set is empty. Finally, estimate the set of the non-null variables by $\widehat{S} = \left\{  i \in [p]: W_i \geq \widehat{T} \right\} $.
	\end{algorithmic}
\end{algorithm}

The knockoff procedure is summarized in Algorithm \ref{oriKnock}.
Step 1 of the algorithm constructs a knockoff matrix $\widetilde{\bX}\in \bbR^{n\times p}$. The knockoffs are constructed so that they act as negative controls, in the sense that all $\widetilde{X}_i$, $i\in [p]$, are null variables in the following model:
\begin{align}\label{eq: linear model extended}
\by = [\bX \widetilde{\bX}] \left( \begin{array}{c} \beta\\ 0 \end{array}\right) + \bz,  
\end{align}
with $0 \in \bbR^p$. Moreover, the knockoffs must mimic the correlation structure of the original variables. Formally, 
\begin{equation}\label{knockConstruct}
\left[ \bX \ \widetilde{\bX} \right]^T \left[ \bX \ \widetilde{\bX} \right]  = 
\begin{pmatrix}
\bX^T \bX & \bX^T \widetilde{\bX}  \\
\widetilde{\bX}^T \bX  & \widetilde{\bX}^T \widetilde{\bX}
\end{pmatrix}
=
\begin{pmatrix}
\bX^T \bX & \bX^T \bX - \diag\{ s \} \\
\bX^T \bX - \diag\{ s \} & \bX^T \bX
\end{pmatrix},
\end{equation}
where $s\in \bbR_{\geq 0}^{p}$ is a non-negative vector, and $\diag\{ s \}$ denotes a diagonal matrix with $s$ on the diagonals. 
The correlations between the original non-null variables and their knockoff counterparts should be small, so that the selection procedure can select the non-null variables. This means that the elements in $s$ corresponding to the non-null variables should be close to $1$. Such $s$, however, cannot be obtained as the non-null variables are unknown. \cite{barber2015controlling} have suggested two approaches to obtain $s$, which we denote by $\mO_{\sdp}$ and $\mO_{\equi}$. Here $\mO_{\sdp}$ calculates $s$ by solving the following convex  semidefinite program (SDP):
\begin{equation}\label{SDP}
\underset{s=(s_1,\dots,s_p)}{\text{minimize}}
\sum_{i=1}^{p} (1-s_i) 
\quad 
\text{subject to}
\quad
\diag \{ s \} \preceq 2\bX^T\bX, \
0 \leq s_i \leq 1, \ i \in [p],
\end{equation}
where the constraints are necessary conditions for the existence of $\widetilde{\bX}$ in equation \eqref{knockConstruct}. This optimization problem can be solved efficiently. The second method $\mO_{\equi}$ adds the additional restriction that all $s_i$ are identical, leading to the following analytical solution:
\begin{equation}\label{Equi}
s_i = \min(2 \lambda_{\min} (\bX^T\bX), 1), \quad i \in [p],
\end{equation}
where $\lambda_{\min} (\bX^T\bX)$ denotes the minimal eigenvalue of $\bX^T\bX$. There is no dominance relationship between $\mO_{\sdp}$ and $\mO_{\equi}$ with respect to statistical power. From our experience, $\mO_{\equi}$ is often less powerful than $\mO_{\sdp}$. 

% Even, want s related to nulls close to 0, so W=0. Does it matter a lot for power? when select hyperpaprameters can also selection this thing? recursive way after obtaining a possible non-null set? Can also first screening and obtain a poss-non-null set, then set s that not included to be 0, and use equi for others.

Regarding the two matrices $\widetilde{U}$ and $C$ in $\eqref{ConstructKnockoffs}$,  $\widetilde{U} \in \bbR^{ n \times p}$ must have column space orthogonal to that of $\bX$, and $C \in \bbR^{ p \times p}$ must satisfy $ C^T C =  2 \diag\{ s \} - \diag\{ s \} (\bX^T\bX)^{-1} \\ \diag\{ s \}$ and can be obtained by a Cholesky decomposition.

In Step 2 of Algorithm \ref{oriKnock}, we construct for each predictor $X_i$ a feature statistic $W_i$ which measures the importance of $X_i$ to the response in model \eqref{eq: linear model}. A large and positive $W_i$ indicates that $X_i$ is likely to be a non-null variable. The construction of $W \in \bbR^p$ is very flexible. The only two conditions that $W$ should satisfy are the so-called sufficiency and antisymmetry properties. The sufficiency property requires $W$ to depend only on the Gram matrix $[ \bX \ \widetilde{\bX} ]^T [ \bX \ \widetilde{\bX} ]$ and the feature-response inner product $[\bX \ \widetilde{\bX} ]^T \by$. The antisymmetry property requires that swapping the columns of an original variable $X_i$ and its knockoff counterpart in $[ \bX \ \widetilde{\bX} ]$ results in a sign switch of $W_i$. 
Please see Appendix \ref{appendix:ExamplesOfW} for some concrete examples of $W$.

The sufficiency and antisymmetry properties of the feature statistics vector $W$ imply the so-called sign-flip property:
\begin{definition}\citep[\textbf{Sign-flip property on index set $H$},][]{barber2015controlling}
	~\\
	Let $A = (A_1, \dots, A_p)$ be any random vector and $H \subseteq [p]$ be any index set. Let $\epsilon_1, \dots, \epsilon_p$ be a sign sequence independent of $A$ with $\epsilon_i = +1$ for $i \not\in H$ and $\epsilon_i$ i.i.d.\ from a Rademacher distribution for $i \in H$. Then, we say that $A$ possesses the sign-flip property on $H$ if
	\begin{equation*}
	(A_1, \dots, A_p) \stackrel{d}{=} (A_1 \cdot \epsilon_1, \dots, A_p\cdot \epsilon_p),
	\end{equation*}
	where $\stackrel{d}{=}$ denotes equality in distribution.
\end{definition}

\begin{lemma} \citep[\textbf{Sign-flip property on $S^c$ with fixed design matrix},][]{barber2015controlling} \label{SignFlipFixDesign}
	~\\
	Let $\bX \in \bbR^{n \times p}$ be a fixed design matrix in a linear model, $S^c$ be the index set of null variables and $W$ be a feature statistic vector satisfying both the sufficiency and the antisymmetry properties. Then $W$ possesses the sign-flip property on $S^c$.
\end{lemma}

The sign-flip property is used to choose a threshold $\widehat{T}$ in Step 3 of Algorithm \ref{oriKnock} (see \eqref{OriKnockThreEquation}, and please see Appendix \ref{appendix:intuitionLinearKnockoff} for some intuition behind this equation) that leads to finite sample FDR control. 
Specifically, based on the sign-flip property of $W$, \cite{barber2015controlling} showed that $\widehat{T}$ can be viewed as a stopping time with respect to some super-martingale. Then, by using the optional stopping time theorem, one has
\begin{align}\label{OriknockoffLessThan1Inequality}
	\bbE \left[ \frac{ \left| \left\{ j \in S^c: W_j\ge \widehat{T} \right\} \right|  } {1 + \left| \left\{ j \in S^c: W_j\le - \widehat{T} \right\}   \right| } \right] \leq 1,
\end{align}
which can be used to prove the finite sample FDR control guarantee of the knockoff procedure (Algorithm \ref{oriKnock}) with $\delta=1$:
\begin{equation*}
\FDR = \bbE[\FDP] = \bbE \left[ \frac{ \left| \left\{ i \in S^c: W_i \geq \widehat{T}  \right\}  \right| } { \left| \left\{i \in[p]: W_i \geq \widehat{T} \right\} \right| \vee 1} \right] \leq q.
\end{equation*}

Without adding the $1$ in the denominator (which corresponds to the case that $\delta=0$), inequality \eqref{OriknockoffLessThan1Inequality} does not hold anymore, nor does the FDR control guarantee. In this case, \cite{barber2015controlling} proved that the knockoff procedure  (Algorithm \ref{oriKnock}) with $\delta=0$ controls a modified version of the FDR:
\begin{equation*}
\mFDR = \bbE \left[ \frac{ \left| \left\{ i \in S^c: W_i \geq \widehat{T} \right\} \right|  } { \left| \left\{ i \in[p]: W_i \geq \widehat{T} \right\}  \right| + q^{-1}} \right] \leq q.
\end{equation*}
The mFDR is smaller than the FDR, so controlling the mFDR may not control the FDR. The two quantities are close if $\left| \left\{i \in[p]: W_i \geq \widehat{T} \right\} \right| \gg q^{-1}$.

%% file: Sections/Sec3-GGMKnockoffFilter.tex
\section{GGM knockoff filter} \label{sec: GGM knockoff filter}

We now turn to the Gaussian graphical model setting. Let $\bX \in \bbR^{n \times p}$ be the data matrix with $n$ i.i.d.\ observations from a $p$-dimensional Gaussian distribution $\mN_p(0, \Omega^{-1})$ (we assume that the mean $\mu=0$ for simplicity
\footnote{When $\mu \neq 0$, the node-wise linear model (see \ref{relationOfGGMandLM}) has an intercept term. Please see footnote $6$ in \cite{barber2015controlling} for details on handling this. }),
and let $G=(V,E)$ be the corresponding undirected graph. 
Throughout this paper, we assume that $n \geq 2p$
\footnote{As the fixed-X knockoff method, the GGM knockoff filter can still be used when $p\leq n<2p$ with some small adaptations. We refer to \cite{barber2015controlling} for details.}.
It is well-known \citep[see, e.g.,][]{lauritzen1996graphical} that $(i,j) \not\in E$ if and only if $\beta^{(i)}_j = \beta^{(j)}_i = 0$, where $\beta^{(i)}_j$ is the regression coefficient of $X_j$ in the regression of $X_i$ on $X_{[p] \backslash \{i\}}$, i.e.,
\begin{equation}\label{relationOfGGMandLM}
\bX^{(i)} = \sum_{j \neq i} \beta^{(i)}_j \bX^{(j)} + \bz^{(i)} = \bX^{(-i)} \ \beta^{(i)} + \bz^{(i)}, \quad i \in [p],
\end{equation}
where $\beta^{(i)} = \left( \beta^{(i)}_j: j\in [p]\backslash \{i\} \right) \in \bbR^{p-1}$, $\bz^{(i)} \in \bbR^{n} $, $\bX^{(i)} \in \bbR^{n}$ is the $i$th column of $\bX$, and $\bX^{(-i)} \in \bbR^{n \times (p-1)}$ denotes the matrix obtained after deleting the $i$th column of $\bX$. Equation \eqref{relationOfGGMandLM} builds the relationship between a Gaussian graphical model and linear models. It is therefore natural to consider whether we can make use of the knockoff framework for FDR controlled graph estimation in Gaussian graphical models. 
For simplicity, we will use the fixed-X rather than the model-X knockoff framework.
We note, however, that the design matrix in equation \eqref{relationOfGGMandLM} is random. This is not a problem, since it is easy to show that the sign-flip property also holds for random design matrices. For completeness, we state this as Lemma \ref{SignFlipRandomDesign} in Appendix \ref{appendix: lemmaRandomDesign}.

We already mentioned that node-wise FDR control at level $q$ does not imply graph-wise FDR control at level $q$. We may get graph-wise FDR control at level $q$, however, if we obtain node-wise FDR control at level $q/p$. We will briefly pursue this idea. We denote by $NE_i$ and $NE_i^c$ the true neighborhood and non-neighborhood of node $i, i\in [p]$:
\begin{align*}
NE_i & = \{ j\in [p]\backslash\{i\}: (i,j) \in E \} = \{ j\in [p]\backslash\{i\}: \beta^{(i)}_j \neq 0 \}, \\
NE_{i}^c & = \{ j\in [p]\backslash\{i\}: (i,j) \not\in E \} = \{ j\in [p]\backslash\{i\}: \beta^{(i)}_j = 0 \}.
\end{align*} 
Now suppose that we use the fixed-X knockoff procedure (Algorithm \ref{oriKnock}) at nominal FDR level $q/p$ for each node $i \in [p]$, by treating variable $X_i$ as response and the remaining variables as predictors (see equation \eqref{relationOfGGMandLM}), to obtain an estimated neighborhood $\widehat{NE}_{i}$. Then, by implementing either the AND rule or the OR rule defined below, 
\begin{equation}\label{ANDandORrules}
\begin{aligned}
\text{AND rule:} &\quad \widehat{E}_{\AND} = \left\{ (i,j) \in [p]^2: i \in \widehat{NE}_j \ \text{ and } \ j \in \widehat{NE}_{i} \right\}, \\
\text{OR rule:} &\quad \widehat{E}_{\OR} =  \left\{ (i,j) \in [p]^2: i \in \widehat{NE}_j \ \text{ or } \ j \in \widehat{NE}_{i} \right\}, \\
\end{aligned}
\end{equation}
we obtain estimated graphs $\widehat E_{\AND}$ or $\widehat E_{\OR}$. 
We denote these two rules by $\mR_{\AND}$ and $\mR_{\OR}$ respectively, and let $\widehat{F}_{\AND} = \widehat{E}_{\AND} \setminus E$ and $\widehat{F}_{\OR} = \widehat{E}_{\OR} \setminus E$ be the corresponding falsely discovered edge sets. Using that $|\widehat{F}_{\OR}| \leq \sum_{i=1}^p |\widehat{NE}_{i} \backslash NE_i|$ and $ |\widehat{E}_{\OR}| \geq |\widehat{NE}_{i} | $, it follows that we have graph-wise FDR control for $\widehat{E}_{\OR}$ at level $q$: 
\begin{equation*}
\FDR 
= \bbE \left[ \frac{ |\widehat{F}_{\OR}| }{ |\widehat{E}_{\OR}| \vee 1} \right]
\leq \sum_{i=1}^p \bbE \left[ \frac{ |\widehat{NE}_{i} \backslash NE_i|}{ |\widehat{NE}_{i} | \vee 1} \right]
\leq \sum_{i=1}^p \frac{q}{p}
= q,
\end{equation*}
where the last inequality holds because the FDR of $\widehat{NE}_{i} $ is controlled at level $q/p$. This approach, however, has zero power, since for any $i \in [p]$, $T > 0$ and $q\in[0,1]$, we have
\begin{equation*}
\frac{ 1+ | \{ j \in[p]\backslash\{i\}: W^{(i)}_j \leq -T \} | } { | \{ j \in[p]\backslash\{i\}: W^{(i)}_j \geq T \} | \vee 1 }
> \frac{1}{p}
\geq \frac{q}{p},
\end{equation*}
where $W^{(i)}_1, \dots, W^{(i)}_p $ denote the feature statistics corresponding to linear model (\ref{relationOfGGMandLM}). Hence, the threshold calculated by formula \eqref{OriKnockThreEquation} (with $\delta=1$) of Algorithm \ref{oriKnock} will always be $+\infty$, which implies that $\widehat{NE}_{i}  = \emptyset$ for all $i\in [p]$ and hence $\widehat{E}_{\OR} = \emptyset$.

%%%%%%%%%%%%%%%%%%%%%%%%%%%%%%%%%%%%%%%%%%%%%%%%%%%%%%%%%%%%%%%%%%%%%%%%%%%%%%%%%%%%%%%%%%%%%%%%%%%%%%%%%%%%%%%%%%%%%%%%%%%%%%%%%%%%%%%%%%%%%%%
%\subsection{GGM knockoff filter}
%The above discussion shows that the relationship between node-wise and graph-wise FDR control is rather unclear. Although there is one procedure that controls the graph-wise FDR, it requires applying an excessively conservative nominal FDR level to each node, making it powerless. Thus, it is not immediately clear about how to get a non-powerless structure learning procedure that controls the graph-wise FDR.

Having seen the problems with simple node-wise approaches, we now propose a different route. Our approach involves two main steps: a node-wise construction of knockoffs and feature statistics, and a global procedure to obtain thresholds for different nodes. 

% method that does not directly set the node-wise nominal FDR levels. Instead, we set one graph-wise level for the entire feature statistic matrix $W$ obtained by treating each variable as response. 

\subsection{Node-wise construction of knockoffs and feature statistic matrix}

For each $i\in [p]$, we construct knockoffs by treating variable $X_i$ as response and the remaining variables as predictors. Specifically, based on the linear relationship (\ref{relationOfGGMandLM}) and equation (\ref{ConstructKnockoffs}), we construct the knockoffs by
\begin{equation*}
\widetilde{\bX}^{(-i)}= \bX^{(-i)} \left( \bI - ({\bX^{(-i)}}^T \bX^{(-i)})^{-1} \diag\{ s^{(i)} \} \right) + \widetilde{U}^{(i)} C^{(i)} , \quad i \in [p],
\end{equation*}
where $s^{(i)}$ is calculated using the pre-decided optimization strategy $\mO_{\sdp}$ or $\mO_{\equi}$ (see \eqref{SDP} and \eqref{Equi}). 

Then, we construct the feature statistics based on $\bX^{(i)}$ and $[\bX^{(-i)} \, \widetilde{\bX}^{(-i)}]$ for each $i \in [p]$.
As we will see in Figure \ref{OppoPower} of Section \ref{section4.1}, feature statistics with good power in one setting can be powerless in another setting. Hence, in order to allow for good power in a wide range of settings, we consider the flexible elastic net regularization \citep{zou2005regularization}, which is a generalization of the Lasso \citep{tibshirani1996regression} that was used in \cite{barber2015controlling}:
\begin{equation}\label{ElasticNet}
\left( \hat{\beta}^{(i)}, \tilde{\beta}^{(i)} \right) = \underset{b \in \bbR^{2 (p-1)}}{\operatorname{\argmin}} \ \left( \frac{1}{2} \left \lVert \bX^{(i)} - [\bX^{(-i)} \, \widetilde{\bX}^{(-i)}] b \right \rVert^2_2
	+ \lambda \left[ (1-\alpha) \lVert b \rVert_2^2/2 + \alpha \lVert b \rVert_1 \right] \right), 
\end{equation}
for some $\alpha \in [0,1]$ and $\lambda \geq 0$. 
For each $j\in[p]\backslash\{i\}$, we denote the estimated coefficient of $X_j$ and $\widetilde{X}_j$ by $\hat{\beta}^{(i)}_j$ and $\tilde{\beta}^{(i)}_{j}$, respectively. As in the previous section,  we can use
\begin{equation}\label{LambdaMax}
Z^{(i)}_j = \sup \{ \lambda \geq 0: \hat{\beta}^{(i)}_j(\lambda, \alpha) \neq 0 \}
\quad \text{and} \quad
\widetilde{Z}^{(i)}_j = \sup \{ \lambda \geq 0: \tilde{\beta}^{(i)}_{j}(\lambda, \alpha) \neq 0\}, 
\end{equation}
for a pre-set $\alpha$, or 
\begin{equation}\label{Coef}
Z^{(i)}_j  = | \hat{\beta}^{(i)}_j(\lambda, \alpha) |
\quad \text{and} \quad
\widetilde{Z}^{(i)}_j = |  \tilde{\beta}^{(i)}_{j}(\lambda, \alpha) |,
\end{equation}
for pre-set $\alpha$ and $\lambda$. 
Then, setting $Z^{(i)}_i = \widetilde{Z}^{(i)}_i = 0$ for notational convenience, we can construct the feature statistics by
\begin{equation}\label{WfromZ1}
W^{(i)}_j = Z^{(i)}_j - \widetilde{Z}^{(i)}_j
\end{equation}
or
\begin{equation}\label{WfromZ2}
W^{(i)}_j = (Z^{(i)}_j \vee \widetilde{Z}^{(i)}_j) \cdot \sign(Z^{(i)}_j - \widetilde{Z}^{(i)}_j)
\end{equation}
for all $ j \in [p]$. Ultimately, we obtain a $p\times p$ feature statistic matrix $W$:
\begin{equation*}
W
= \left( W^{(1)}, \dots,  W^{(p)} \right)
=
\begin{pmatrix}
0 & W^{(2)}_{1}  & \dots & W^{(p-1)}_{1} & W^{(p)}_{1} \\
W^{(1)}_{2} & 0 & \dots & W^{(p-1)}_{2} & W^{(p)}_{2} \\
\vdots & \vdots  & \ddots & \vdots & \vdots \\
W^{(1)}_{p-1} & W^{(2)}_{p-1}  & \dots & 0 & W^{(p)}_{p-1} \\
W^{(1)}_{p} & W^{(2)}_{p}  & \dots & W^{(p-1)}_{p} & 0
\end{pmatrix},
\end{equation*}
where a large and positive $W^{(i)}_{j}$ indicates that there is likely an edge between node $i$ and $j$.

%%%%%%%%%%%%%%%%%%%%%%%%%%%%%%%%%%%%%%%%%%%%%%%%%%%%%%%%%%%%%%%%%%%%%%%%%%%%%%%%
%%%%%%%%%%%%%%%%%%%%%%%%%%%%%%%%%%%%%%%%%%%%%%%%%%%%%%%%%%%%%%%%%%%%%%%%%%%%%%%%

\subsection{Obtaining the global threshold vector}

Given the feature statistic matrix $W$ and a positive threshold vector $T=( T_1, \dots, T_p ) $, we define
\begin{equation}\label{NotationV}
\begin{split}
\widehat{V}^{+}_{i} & = \{ j \in [p]: W^{(i)}_j \geq T_i \},\\ 
\widehat{V}^{-}_{i} & = \{ j \in [p]: W^{(i)}_j \leq -T_i \}, \\
\end{split}
\qquad \quad
\begin{split}
\widehat{V}^{+}_{N_i} &  = \{ j \in NE_{i}^c: W^{(i)}_j \geq T_i \},  \\
\widehat{V}^{-}_{N_i} &  = \{ j \in NE_{i}^c: W^{(i)}_j \leq -T_i \}. \\
\end{split}
\end{equation}
Here, $\widehat{V}^{+}_{i}$ will be the estimated neighborhood of node $i$, taking the role of the previously used $\widehat{NE}_{i}$. 
The set $\widehat{V}^{+}_{N_i} \subseteq \widehat{V}^{+}_{i}$ is the set of the falsely discovered neighbors of node $i$, where the subscript $N$ stands for ``null". The $\widehat{V}^{+}_{i}$, $i\in [p]$, can be combined using $\mR_{\AND}$ or $\mR_{\OR}$ to get an estimated edge set $\widehat{E}$. The roles of $\widehat{V}^{-}_{i}$ and $\widehat {V}^{-}_{N_i}$ will become clear later (see \eqref{IntuitionOfOurMehotd}). Both $\widehat{V}^{+}_{N_i}$ and $\widehat{V}^{-}_{N_i}$ are unobservable, while $\widehat{V}^{+}_{i}$ and $\widehat{V}^{-}_{i}$ are observable.

%The general steps to estimate the edge set $E$ are as follows:
%(1) Choose a threshold $\widehat{T}_i > 0$ and obtain estimated neighbors $\widehat{V}^{+}_{i}$ for each node $i$. (2) Combine the estimated neighbors $\widehat{V}^{+}_{i} $ by using either $\mR_{\AND}$ or $\mR_{\OR}$ to get $\widehat{E}$.

Obtaining a proper threshold vector $T=(T_1,\dots,T_p)$ is the key to control the FDR without losing all power. Since simple node-wise thresholds do not seem to work, our approach determines the threshold globally, by considering the entire feature statistic matrix $W$. On the matrix level, however, $W$ does not possess the sign-flip property, an essential ingredient in the proof of FDR control in the knockoff framework. Fortunately, $W$ does satisfy the sign-flip property at the column level. Therefore, we can split the FDR (or FDP) into column-wise terms, so that we can make use of the column-wise sign-flip property. For example, if we consider using $\mR_{\AND}$, then for any positive fixed threshold vector $T=(T_1,\dots,T_p)$ and $a>0$, we have
\begin{equation}\label{IntuitionOfOurMehotd}
	\begin{aligned}
	\FDP(T) &= \frac{ |\widehat{F}_{\AND}(T)| }{ |\widehat{E}_{\AND}(T)| \vee 1}
	\leq \frac{1}{2} \sum_{i=1}^p \frac{  |\widehat{V}^{+}_{N_i}(T_i)| }{ |\widehat{E}_{\AND}(T)| \vee 1} \\
	& \leq \frac{1}{2} \sum_{i=1}^p \frac{ |\widehat{V}^{+}_{N_i}(T_i)| }  { a + |\widehat{V}^{-}_{N_i}(T_i)| }
	\cdot \frac{ a + |\widehat{V}^{-}_{i}(T_i)| }{ |\widehat{E}_{\AND}(T)| \vee 1} \\
	& \approx \frac{1}{2} \sum_{i=1}^p \frac{ |\widehat{V}^{-}_{N_i}(T_i)| }  { a + |\widehat{V}^{-}_{N_i}(T_i)| }
	\cdot \frac{ a + |\widehat{V}^{-}_{i}(T_i)| }{ |\widehat{E}_{\AND}(T)| \vee 1} \\
	& \leq \frac{1}{2} \sum_{i=1}^p \frac{ a + |\widehat{V}^{-}_{i}(T_i)| }{ |\widehat{E}_{\AND}(T)| \vee 1},
	\end{aligned}
\end{equation}
where the approximation in line $3$ is due to the sign-flip property of $W^{(i)}$. The last line of \eqref{IntuitionOfOurMehotd} motivates the following optimization problem to calculate a data-dependent threshold vector $\widehat{T}=(\widehat{T}_1,\dots,\widehat{T}_p)$:
\begin{equation}\label{ANDopt}
\begin{aligned}
\widehat{T}=
&\underset{  T=( T_1, \dots, T_p ) } {\argmax}
& & |\widehat{E}_{\AND}| \\
& \text{subject to}
& &  \frac{ a +|\widehat{V}^{-}_{i}(T_i)|}{|\widehat{E}_{\AND}(T)| \vee 1} \leq q_i := \frac{2}{c_a p}q \quad \text{and} \\
&&&  T_i \in  \left\{ |W^{(i)}_j|, j \in [p]  \right\} \cup \{ +\infty \} \backslash \{ 0 \}, \,\text{for all}\, i\in [p]
\end{aligned}
\end{equation}
where $c_a>0$ is a constant depending on $a$ (we will discuss them in more detail later), and we set $\widehat{T} = (+\infty, \dots, +\infty)$ if there is no feasible point satisfying the above constraints. Here, $q_1,\dots,q_p$ can be set to any values as long as $\sum_{i=1}^{p} q_i = 2q/c_a$.
We make the choice to set $q_i=2q/(c_a p)$ for all $i\in [p]$ for two reasons: it seems a natural choice that treats each node equally, and it allows us to find the global optimizer of this combinatorial optimization problem via a simple and fast algorithm (Algorithm \ref{algo2}) which will be discussed later.
We use separate constraints for each $i$ rather than one summed constraint (i.e., $\sum_{i=1}^{p} \frac{ a +|\widehat{V}^{-}_{i}(T_i)|}{|\widehat{E}_{\AND}(T)| \vee 1} \leq 2q/c_a$ ) due to a technical reason in proving the FDR control, please see Remark \ref{ReasonForNotUsingSummationInOpt} in Appendix \ref{appendix: proofGGMfilter FDR} for details.

Comparing \eqref{ANDopt} to the threshold formula \eqref{OriKnockThreEquation} in the linear model case, one can see that their forms are quite similar, because \eqref{OriKnockThreEquation} can be viewed as an optimization problem with objective $|\widehat{S}|$ and a constraint that is similar to the one in \eqref{ANDopt}. 
We note, however, that the numerator $a+|\widehat{V}^{-}_{i}(T_i)|$ in the constraint of \eqref{ANDopt} is a local term that involves only the estimated neighborhood of node $i$ for the given threshold $T_i$, while the corresponding denominator $|\widehat{E}_{\AND}(T)| \vee 1$ is a global term involving the entire estimated graph. The constant $c_a$ is an upper bound such that
%Rather than bounding the resulting fraction by $q$, it must be bounded by $q_i := (2/1.93p)q$. Here the factor $1/p$ in $q_i$ can be viewed as the price to pay for using the global term in the denominator, and the factor 2 comes from the factor $1/2$ in the last line of \eqref{IntuitionOfOurMehotd}.
\begin{align}\label{c_a bound}
\bbE \left[ \frac{ | \widehat{V}^{+}_{N_i}(\widehat{T}_i) | } {a + | \widehat{V}^{-}_{N_i}(\widehat{T}_i)|} \right] \leq c_a.
\end{align}
Inequality \eqref{c_a bound} in the GGM setting is similar to the inequality \eqref{OriknockoffLessThan1Inequality} in the linear model setting. However, although $\widehat{V}^{+}_{N_i}(T_i)$ and $\widehat{V}^{-}_{N_i}(T_i)$ depend only on column $W^{(i)}$ for a given $T_i$, the actually used data-dependent threshold $\widehat{T}_i$ is calculated using the entire matrix $W$, for which the sign-flip property does not hold. This prevents us from using the martingale arguments as in \cite{barber2015controlling} to obtain $c_a$ for a given $a$. \cite{katsevich2017multilayer} encountered a similar issue in their multilayer controlled variable selection problem. They only considered the case of $a=1$ and proposed to bound this term by taking a supremum inside the expectation with respect to the threshold, which results in an upper bound of $c_a=1.93$. Their proof arguments, in fact, hold as long as $a>0$ (see Lemma \ref{Lemma: Bound of c_a} in Appendix \ref{appendix: Bound of c_a}). Lemma \ref{Lemma: Bound of c_a} allows us to calculate $c_a$ for any given $a>0$. The product $a c_a$ can be smaller than $1.93$ when $a<1$, and this can be useful to obtain a better result for the optimization problem \eqref{ANDopt}. To see this, the constraint of \eqref{ANDopt} can be rewritten as
\begin{align*}
\frac{ 1 + \frac{1}{a} |\widehat{V}^{-}_{i}(T_i)|}{|\widehat{E}_{\AND}(T)| \vee 1} \leq \frac{2}{a c_a p}q.
\end{align*}
A smaller $a$ would penalize $|\widehat{V}^{-}_{i}(T_i)|$ more heavily (thus this term would tend to $0$), but at the same time result in a larger term on the right hand side as $(a c_a)^{-1}$ would be larger. This implies that there is no optimal $a$ for the optimization problem \eqref{ANDopt} in general. In simulations, we found that  $a=0.01$ (with $c_a=102$, see Proposition \ref{Proposition: get c_a} in Appendix \ref{appendix: Bound of c_a}) tends to work well. In this case, the product $ac_a$ is $1.02$, which is smaller than $1.93$ (corresponds to $a=1$ and $c_a=1.93$ from \cite{katsevich2017multilayer}). This can be useful when there is no feasible point for $(a, c_a)=(1, 1.93)$ in the optimization problem \eqref{ANDopt}. 
Therefore, for the choice of $(a, c_a)$, we consider two options $(1, 1.93)$ and $(0.01, 102)$. 
The sign-flip property of $W^{(i)}$ is crucial in order to use Lemma \ref{Lemma: Bound of c_a} to obtain $c_a$.

% Hence, $ac_a^{-1}$ can be viewed as the price we must pay for the coupled system. 

Similarly, if we use $\mR_{\OR}$ to recover the graph from the estimated neighborhoods $\widehat{V}^+_i$, the threshold vector can be obtained by
\begin{equation}\label{ORopt}
\begin{aligned}
\widehat{T}=
&\underset{ T=( T_1, \dots, T_p ) } {\argmax}
& & |\widehat{E}_{\OR}| \\
& \text{subject to}
& &  \frac{ a + |\widehat{V}^{-}_{i}(T_i)|}{|\widehat{E}_{\OR}(T)| \vee 1} \leq \tilde q_i := \frac{1}{c_a p}q, \\
&&&  T_i \in  \left\{ |W^{(i)}_j|, j \in [p] \right\} \cup \{ +\infty \} \backslash \{ 0 \}, \ i\in[p].
\end{aligned}
\end{equation}
Again, we set $\widehat{T} = (+\infty, \dots, +\infty)$ if no feasible point is found. 

\iffalse
\begin{equation}\label{equ9}
|\widehat{F}_{\AND}| = \frac{1}{2} \Bigg( \sum_{i=1}^p |\widehat{V}^{+}_{N_i}| - |\widehat{F}_{\OR} \backslash \widehat{F}_{\AND}| \Bigg),
\quad
|\widehat{F}_{\OR}| = \sum_{i=1}^p |\widehat{V}^{+}_{N_i}| - | \widehat{F}_{\AND} |,
\end{equation}
\fi

As mentioned before, both optimization problems (\ref{ANDopt}) and (\ref{ORopt}) can be solved via a simple and fast algorithm (Algorithm \ref{algo2}). This algorithm is computationally very cheap as there are at most $m_{\max}+1$ points to be checked and $m_{\max}$ is typically small. Thus, the most computationally expensive part of the GGM knockoff filter (Algorithm \ref{Algo: GGMknock}) consists of constructing the feature statistic matrix in a node-wise manner (Step 1). This can be easily parallelized, and this seems a fair price to pay for dealing with the more involved structure learning problem when compared to variable selection. For the details about Algorithm \ref{algo2}, please see Appendix \ref{appendix: optimization problem}.

\begin{algorithm}
	\caption{\textbf{: Computing the threshold vector $\widehat{T}$ }} \label{algo2}
	\textbf{Input}: $(W, q, \delta, (a,c_a), \mR)$, where $W \in \bbR^{p \times p}$ is the feature statistics matrix, $q$ is the nominal FDR level, $\delta \in \{0,1\}$ indicates FDR control $(\delta=1)$ or mFDR control $(\delta=0)$, $(a,c_a) \in \{(1,1.93), (0.01, 102)\}$ are parameters in the optimization problem, and $\mR \in \{ \mR_{\AND}, \mR_{\OR} \}$ is the rule used to recover the estimated edge set from the estimated neighborhoods.
	
	\textbf{Output}: The threshold vector $\widehat{T}=( \widehat{T}_1, \dots, \widehat{T}_p )$.
	
	\begin{algorithmic}[1]
		\State Fix the optimization problem to be solved: \eqref{ANDopt} if $(\mR, \delta)=(\mR_{\AND},1)$, \eqref{ORopt} if $(\mR, \delta)=(\mR_{\OR},1)$, \eqref{ANDopt}  with replacement \eqref{mFDRoptAND} if $(\mR, \delta)=(\mR_{\AND},0)$, and \eqref{ORopt} with replacement \eqref{mFDRoptOR} if $(\mR, \delta)=(\mR_{\OR},0)$
		\State Let $m_{\max}$ be the largest integer that is smaller than or equal to: $q(p-1)/c_a-a \delta$ if $\mR=\mR_{\AND}$, and $q(p-1)/(2c_a)-a \delta$ if $\mR=\mR_{\OR}$.
		\If {$m_{\max}<0$} {return $\widehat{T}=(+\infty, \dots,+\infty)$.}
		\EndIf
		\For {$m=m_{\max},m_{\max}-1,\dots,0$,}
		\newline 
		\-\hspace{0.5cm} Let $\widehat{T}=( \widehat{T}_1, \dots, \widehat{T}_p )$ with 
		\newline 
		\-\hspace{0.5cm}
		$\widehat{T}_i = \min  \left\{ T_i \in  \left\{ |W^{(i)}_j|, j \in [p] \right\} \cup \{+\infty\}  \backslash \{0\}:  |\widehat{V}^{-}_{i}(T_i)| \leq m \right\}$.
		\newline
		\-\hspace{1cm} \textbf{if} $\widehat{T}$ satisfies the constraints of the given optimization problem, \textbf{then} return $\widehat{T}$.
		\EndFor
		\If { there is no feasible point, } {return $\widehat{T}=(+\infty, \dots,+\infty)$.}
		\EndIf
	\end{algorithmic}
\end{algorithm}

%%%%%%%%%%%%%%%%%%%%%%%%%%%%%%%%%%%%%%%%%%%%%%%%%%%%%%%%%%%%%%%%%%%%%%%%%%%%%%%%
%%%%%%%%%%%%%%%%%%%%%%%%%%%%%%%%%%%%%%%%%%%%%%%%%%%%%%%%%%%%%%%%%%%%%%%%%%%%%%%%

\subsection{Finite sample graph-wise FDR and mFDR control guarantees}

The GGM knockoff filter for FDR control is summarized in Algorithm \ref{Algo: GGMknock} (with $\delta=1$), and the theoretical FDR control guarantee is given in Theorem \ref{Thm: GGMfilter FDR}

\begin{algorithm}[t]
	\caption{\textbf{: GGM knockoff filter}} \label{Algo: GGMknock}
	\textbf{Input}: $(\bX, q, \delta, (a,c_a), \mO, \mP, \mR)$, where $\bX$ is a data matrix, $q \in [0,1]$ is the nominal FDR level, $\delta \in \{0,1\}$ indicates FDR control $(\delta=1)$ or mFDR control $(\delta=0)$, $(a,c_a) \in \{(1,1.93), (0.01, 102)\}$ are parameters in the optimization problem, $\mO \in \{ \mO_{\equi}, \mO_{\sdp} \}$ is the optimization strategy used for the construction of knockoffs, $\mP$ is the procedure used to construct feature statistics, and $\mR \in \{ \mR_{\AND}, \mR_{\OR} \}$ is the rule used to recover the estimated edge set from the estimated neighborhoods.
	
	\textbf{Output}: Estimated edge set $\widehat{E}$.
	
	\begin{algorithmic}[1]
		\State \textbf{Step 1.} Construct the feature statistics $W = (W^{(1)}, \dots, W^{(p)})$, where for each $i\in [p]$,  $W^{(i)}$ is created by treating $\bX^{(i)}$ as a response vector and $\bX^{(-i)}$ as a design matrix, 
		and applying the first two steps of Algorithm \ref{oriKnock} with input $(\bX^{(-i)}, \bX^{(i)}, q, \delta, \mO, \mP)$.		
			\State \textbf{Step 2.} Compute thresholds $\widehat{T}=( \widehat{T}_1, \dots, \widehat{T}_p ) $ based on $W$, $q$, $\delta$, $(a, c_a)$, and $\mR$ via Algorithm \ref{algo2}.
		\State \textbf{Step 3.} Obtain $\widehat{V}^{+}_{i}  = \{ j \in [p]: W^{(i)}_j \geq \widehat{T}_i \}$ for all $i\in[p]$, and recover the estimated edge set $\widehat{E}$ based on $\widehat{V}^{+}_1,\dots,\widehat{V}^{+}_p$ and $\mR$.
	\end{algorithmic}
\end{algorithm}

\begin{restatable}{theorem}{FDRcontrol}(\textbf{Finite sample graph-wise FDR control of the GGM knockoff filter}) \label{Thm: GGMfilter FDR}
	~\\
	Let $\bX_1, \dots, \bX_n$ be $n$ i.i.d.\ observations from a $p$-dimensional Gaussian distribution $\mN_p(0,\Omega^{-1})$ with undirected graph $G=(V,E)$, and assume that $n\geq 2p$. Let $\widehat{E}$ be the estimated edge set obtained using Algorithm \ref{Algo: GGMknock} with $\delta=1$ and let $ \widehat{F} = \widehat{E} \backslash E$ be the falsely discovered edge set. Then, for any $q \in [0,1]$, we have
	\begin{equation*}
	\FDR = \bbE \Bigg[ \frac{ |{\widehat{F}}| }{ |\widehat{E}| \vee 1} \Bigg] \leq q.
	\end{equation*}
\end{restatable}

%%%%%%%%%%%%% \subsection{Controlling the mFDR}

Similarly to the original fixed-X knockoff method in \cite{barber2015controlling}, the modified FDR (defined in \eqref{mFDRgraph}) of the estimated graph can also be controlled with a slight change of the constraints in optimization problems \eqref{ANDopt} and \eqref{ORopt}. 
\begin{equation} \label{mFDRgraph}
\mFDR_{\AND} = \bbE \Bigg[ \frac{ |\widehat{F}_{\AND}| }{ |\widehat{E}_{\AND}| + a c_a p/(2q) } \Bigg],
\quad
\mFDR_{\OR} = \bbE \Bigg[ \frac{ |\widehat{F}_{\OR}| }{ |\widehat{E}_{\OR}| + a c_a p/q } \Bigg].
\end{equation}
Specifically, by replacing
\begin{equation}\label{mFDRoptAND}
\frac{ a+|\widehat{V}^{-}_{i}(\widehat{T}_i)|}{|\widehat{E}_{\AND}(\widehat{T})| \vee 1} \leq \frac{2q}{c_a p}
\quad \text{by} \quad 
\frac{ |\widehat{V}^{-}_{i}(\widehat{T}_i)|}{|\widehat{E}_{\AND}(\widehat{T})| \vee 1} \leq \frac{2q}{c_ap}
\end{equation}
and 
\begin{equation}\label{mFDRoptOR}
\frac{ a+|\widehat{V}^{-}_{i}(\widehat{T}_i)|}{|\widehat{E}_{\OR}(\widehat{T})| \vee 1} \leq \frac{q}{c_a p}
\quad \text{by} \quad 
\frac{ |\widehat{V}^{-}_{i}(\widehat{T}_i)|}{|\widehat{E}_{\OR}(\widehat{T})| \vee 1} \leq \frac{q}{c_ap}
\end{equation}
in \eqref{ANDopt} and \eqref{ORopt}, respectively, the $\mFDR_{\AND}$ and $\mFDR_{\OR}$ can be controlled, respectively.
By controlling the mFDR instead of the FDR, more discoveries will be made. However, these two modified FDRs can be much smaller than the FDR, and are close to the FDR if $|\widehat{E}_{\AND}| \gg a c_a p/(2q) $ and $|\widehat{E}_{\OR}| \gg a c_a p/q $. There is a trade-off for different values of $(a,c_a)$: for a smaller $a$, the product $a c_a$ would be smaller, thus the mFDR would be closer to the FDR. The constraints \eqref{mFDRoptAND} and \eqref{mFDRoptOR} of the optimization problem, however, would have  smaller right hand sides as $c_a$ would be larger, and thus a smaller number of discoveries. 

The GGM knockoff filter for mFDR control is summarized in Algorithm \ref{Algo: GGMknock} (with $\delta=0$), and the theoretical mFDR control guarantee is given in Theorem \ref{Thm: GGMfilter mFDR}.

\begin{restatable}{theorem}{mFDRcontrol}(\textbf{Finite sample graph-wise mFDR control of GGM knockoff filter}) \label{Thm: GGMfilter mFDR}
	~\\
	Let $\bX_1, \dots, \bX_n$ be $n$ i.i.d.\ observations from a $p$-dimensional Gaussian distribution $\mN_p(0, \Omega^{-1})$ with undirected graph $G=(V,E)$, and assume that $n\geq 2p$. Let $\widehat{E}_{ \AND / \OR }$ be the estimated edge set obtained using Algorithm \ref{Algo: GGMknock} with $\delta=0$ and hyperparameters $(a,c_a)$. Let $ \widehat{F}_{ \AND / \OR } = \widehat{E}_{ \AND / \OR } \backslash E$ be the falsely discovered edge set.
	Then, for any $q \in [0,1]$, we have
	\begin{equation*}
	\mFDR_{\AND} = \bbE \Bigg[ \frac{ |\widehat{F}_{\AND}| }{ |\widehat{E}_{\AND}| + a c_a p/(2q) } \Bigg] \leq q 
	\quad \text{and} \quad
	\mFDR_{\OR} = \bbE \Bigg[ \frac{ |\widehat{F}_{\OR}| }{ |\widehat{E}_{\OR}| + a c_a p/q } \Bigg] \leq q.
	\end{equation*}
\end{restatable}

%% file: Sections/Sec4-SampleSplittingRecycling.tex
\section{GGM knockoff filter with sample-splitting-recycling} \label{sec: sample-splitting-recycling}

\subsection{Motivation} \label{section4.1}
For a given $(\bX, q, \delta)$, the GGM knockoff filter (Algorithm \ref{Algo: GGMknock}) requires hyperparameters $(a,c_a)$, $\mO$, $\mP$ and $\mR$ which can be chosen very flexibly. It is not clear a priori, however, what choice is optimal for a given data set. 
In this paper, we consider $880$ different choices of $ ((a,c_a), \mO, \mP, \mR)$ (see Section \ref{sec:simu5.1} for details).
To investigate their performances, we apply them (with Algorithm \ref{Algo: GGMknock}) to two different settings taken from the simulation settings in Section \ref{sec:simu} (see Appendix \ref{appendix: SimulationSettingSection4} for more details). Each point in Figure \ref{OppoPower} represents the power of one procedure in the two different settings. One can see that no procedure is uniformly the best in both settings, and some powerful procedures in setting 1 are even powerless in setting 2, and vice versa. The procedures with hyperparameter $\mP$ used as in the simulation part of \cite{barber2015controlling} (that is, $\mP \in \{ \eqref{LambdaMax}+\eqref{WfromZ1}, \eqref{LambdaMax}+\eqref{WfromZ2}\}$ with $\alpha=1$) are shown in red. They have good powers in setting 1, but not in setting 2. The main message conveyed by Figure \ref{OppoPower} is that different choices of $((a,c_a), \mO, \mP, \mR)$ can result in very different statistical powers in different settings, and no particular choice is dominant.
\begin{figure}[t!]
	\centering
	\includegraphics[width=8cm, height=8cm]{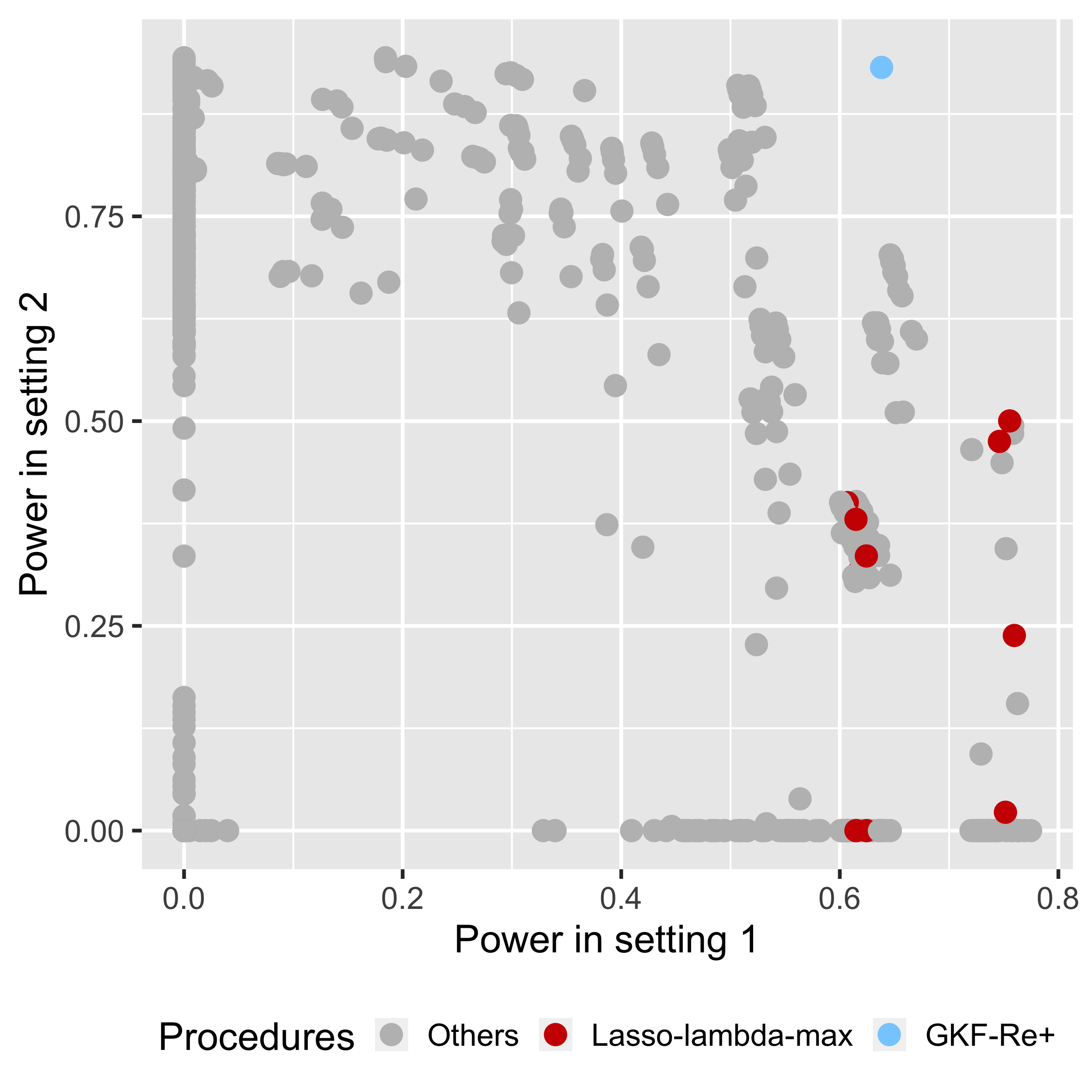}
	\caption{The power of $880$ different FDR control procedures resulting from different choices of $ ((a,c_a), \mO, \mP, \mR)$ and the power of the GGM knockoff filter with sample-splitting-recycling in settings 1 and 2. The red points denote the power of the procedures with $\mP$ as in \cite{barber2015controlling}. The blue point denotes the power of the GGM knockoff filter with sample-splitting-recycling.}
	\label{OppoPower}
\end{figure}

%[JL: More plot: as discussed, color according to equi, sdp ,..., put in appendix? TBD]

Hence, one practical and important question naturally arises: given many valid FDR control procedures, 
which one should we choose for a given problem to ensure good power while maintaining the FDR control guarantee?
This problem has not received much attention in the FDR control literature so far, but seems a key issue for practical applicability and reproducibility when there are many FDR control procedures available (including both knockoff based and p-value based methods).

\subsection{Sample-splitting-recycling with FDR and mFDR control guarantees} 

As mentioned in Section \ref{sec:intro}, the naive idea of choosing the procedure returning the maximal number of discoveries loses FDR control. On the other hand, sample-splitting does guarantee FDR control, but has two  issues: (i) the statistical power is lower because only half of the sample is used to obtain the final result, and (ii) the final result is random, as the splitting is random.
We are not able to solve the second issue, but the first issue can be alleviated to some extent by using the sample-splitting-recycling approach proposed by \cite{barber2016knockoff}. 
Their paper considers using knockoffs to select variables in a high-dimensional linear model, and sample-splitting is used to reduce the dimensionality to obtain a low-dimensional linear model based on the first part of sample. Then, instead of selecting variables based only on the second half of sample, they use all samples in a particular way that ensures FDR control. The key point of reusing all samples lies in the construction of knockoffs corresponding to the used sample. In our case, we will use the same sample-splitting-recycling approach to select the FDR control procedure. 

Formally, let $\bX \in \bbR^{n \times p}$ be the original sample matrix, and $\bX_{1} \in \bbR^{n_1 \times p}$ and $\bX_{2} \in \bbR^{n_2 \times p}$ be two subsample matrices obtained by randomly splitting $\bX$ with $n_1+n_2=n$. In this paper, we will use $n_1 = \lfloor n /2 \rfloor$ and $n_2 = n - n_1$, where $\lfloor n /2 \rfloor$ denotes the biggest integer that is smaller than $n/2$. We denote the collection of these two subsamples by $ \bX^{\re} = \begin{pmatrix} \bX_{1} \\ \bX_{2} \end{pmatrix} \in \bbR^{n \times p}$.

First, we use $\bX_{1}$ to select one procedure among $m$ candidate procedures. Concretely, we run Algorithm \ref{Algo: GGMknock} with all $m$ choices of $((a,c_a), \mO, \mP, \mR)$ to obtain $m$ estimated edge sets $\widehat{E}$. Then, we choose the combination $((a^*,c_a^*), {\mO}^*, \mP^*, \mR^*)$ that maximize $|\widehat{E}|$ (we randomly choose one procedure if there is a tie). For the vanilla sample-splitting approach, one would then apply the chosen $((a^*,c_a^*), {\mO}^*, \mP^*, \mR^*)$ to $\bX_{2}$ to get the final estimated graph. Sample-splitting-recycling, however, also makes use of the used data $\bX_{1}$, but in a particular way. Specifically, for each $i\in[p]$, we construct the knockoff matrix by
\begin{equation}\label{ReKnockoffs}
	\widetilde{\bX}^{\re(-i)} =  \begin{pmatrix} \bX_{1}^{(-i)} \\ \widetilde{\bX}_{2}^{(-i)}\end{pmatrix} \in \bbR^{n \times (p-1)},
\end{equation}
where $\widetilde{\bX}_{2}^{(-i)}$ is the knockoff matrix constructed using formula \eqref{ConstructKnockoffs} based on $ \bX_{2}^{(-i)}$. It is easy to verify that the matrix $\widetilde{\bX}^{\re(-i)}$ in \eqref{ReKnockoffs} satisfies \eqref{knockConstruct}, so that it is indeed a valid knockoff matrix for $\bX^{\re(-i)}$. Note that the knockoff matrix $\widetilde{\bX}^{\re(-i)}$ is not constructed directly based on $ \bX^{\re (-i)}$ using \eqref{ConstructKnockoffs}. This is the only key modification, and the remaining procedure is the same as in Algorithm \ref{Algo: GGMknock}. We show by simulations that the sample-splitting-recycling approach indeed helps to improve the power compared to sample-spitting, see Appendix \ref{appendix: RecyclingPowerGain} for details.

The GGM knockoff filter with sample-splitting-recycling is summarized in Algorithm \ref{Algo: GGMknockWithSplittingRecycling}.
As the definition of the modified FDR is different for $\mR_{\AND}$ and $\mR_{\OR}$, the controlled modified FDR of Algorithm \ref{Algo: GGMknockWithSplittingRecycling} depends on the chosen rule. Hence, we define
\begin{equation}\label{mFDR: DefineforRecycleAlgo}
\mFDR^{\re}
= \bbE \left[ \frac{ |\widehat{F}| }{ |\widehat{E}| + a c_a p/(q\bbone_{ \{ \mR_{\AND}\text{ is chosen based on } \bX_{1} \} } + q)} \right].
\end{equation}
The following theorem guarantees the FDR and the modified FDR control of Algorithm \ref{Algo: GGMknockWithSplittingRecycling}.

\begin{algorithm}[h]
	\caption{\textbf{: GGM knockoff filter with sample-splitting-recycling}}\label{Algo: GGMknockWithSplittingRecycling}
	\textbf{Input}: $(\bX, q, \delta,  \bA, \bmO, \bmP, \bmR)$, where $\bX$ is a data matrix, $q \in [0,1]$ is the nominal FDR level, $\delta \in \{0,1\}$ indicates FDR control $(\delta=1)$ or mFDR control $(\delta=0)$, 
	$\bA$ is the set of $(a,c_a)$ used in the optimization problem, 
	$\bmO$ is the set of optimization strategies used for the construction of knockoffs, 
	$\bmP$ is the set of candidate procedures to construct feature statistics,
	and $\bmR$ is the set of rules used to recover the estimated edge set from the estimated neighborhoods.
	
	\textbf{Output}: Estimated edge set $\widehat{E}$.
	
	\begin{algorithmic}[1]
		\State \textbf{Step 1.} Randomly split the data matrix $\bX$ into two parts $\bX_{1}$ and $\bX_{2}$.
		\State \textbf{Step 2.} Implement Algorithm \ref{Algo: GGMknock} with input $ (\bX_{1}, q, \delta, (a,c_a), \mO, \mP, \mR)$ for all combinations of $(a,c_a) \in \bA, \mO \in \bmO, \mP \in \bmP, \mR \in \bmR$ to obtain $\widehat{E}(\bX_{1}, q, \delta, (a,c_a), \mO, \mP, \mR)$.
		Let $((a^*,c_a^*), {\mO}^*, \mP^*, \mR^*)$ be the combination leading to the maximum number of discoveries.
		\State {\textbf{Step 3.} Generate the feature statistic matrix $W=(W^{(1)}, \dots,W^{(p)})$ using $\bX^{\re} = \begin{pmatrix} \bX_{1} \\ \bX_{2} \end{pmatrix}$, ${\mO}^*$ and $\mP^*$:}
		\newline
		\-\hspace{0.5cm}
		\textbf{for} $i=1, \dots, p$, \textbf{do}
		\newline
		\-\hspace{1cm}
		Construct the knockoff matrix $\widetilde{\bX}^{\re(-i)}$ as in (\ref{ReKnockoffs}) using $\bX^{\re(-i)}$ and ${\mO}^*$. Then
		\newline
		\-\hspace{1cm}
		apply Step 2 of Algorithm \ref{oriKnock} with input $(\bX^{\re(-i)}, \widetilde{\bX}^{\re(-i)}, \bX^{\re(i)}, \mP^*)$ to obtain 
		\newline
		\-\hspace{1cm}
		the feature statistics $W^{(i)}$.
		\State \textbf{Step 4.} Apply Step 2 and Step 3 of Algorithm \ref{Algo: GGMknock} with $W$, $q$, $\delta$, $(a^*,c_a^*)$ and $\mR^*$, to obtain $\widehat{E}$.
	\end{algorithmic}
\end{algorithm}

\begin{restatable}{theorem}{SampleSplittingRecyclinfFDRmFDRresuklts}(\textbf{Finite sample FDR/mFDR control of GGM knockoff filter with sample-splitting-recycling}) \label{Thm: FDRmFDRSplittingRecycling}
	~\\
	Let $\bX_1, \dots, \bX_n$ be $n$ i.i.d.\ observations from a $p$-dimensional Gaussian distribution $\mN_p(0,\Omega^{-1})$ with undirected graph $G=(V,E)$, and assume that $n\geq 4p$. For any $q \in [0,1]$,
	Algorithm \ref{Algo: GGMknockWithSplittingRecycling} with $\delta=1$ controls the $\FDR$ at level $q$, and Algorithm \ref{Algo: GGMknockWithSplittingRecycling} with $\delta=0$ controls the $\mFDR^{\re}$ defined as \eqref{mFDR: DefineforRecycleAlgo} at level $q$.
\end{restatable}

As a quick illustration of the performance of the sample-splitting-recycling approach, we implemented Algorithm \ref{Algo: GGMknockWithSplittingRecycling} in the previous setting 1 and setting 2. Its power is shown in blue in Figure \ref{OppoPower}. This shows a successful result of this approach, since the blue point lies near the top right corner with good powers in both settings.

%Finally,  we would like to point out that these $880$ combinations of hyperparameters are our own choices in this paper, one can flexibly use more (or less) choices of hyperparameters and constructed them in ways different from what we have done here. We recommend to use as many choices of hyperparameters procedures as possible within an affordable computational expense to achieve a good statistical power.

%%%%%%%%%%%%%%%%%%%%%%%%%%%%%%%%%%%%%%%%%%%%%%%%%%%%%%%%%%%%%%%%%%%%%%%%%%%%%%%%

%% file: Sections/Sec5-Simulations.tex
\section{Simulations and a real data example} \label{sec:simu}
We now examine the performance of our proposed GGM knockoff filter with sample-splitting-recycling in a range of settings and compare it to other methods. All analyses were carried out in R and the code is made available at \url{https://github.com/Jinzhou-Li/GGMKnockoffFilter-R}.

\subsection{Simulation set-up and compared methods} \label{sec:simu5.1}

For each replication, we generate $n$ independent samples from $\mN_p(0,\Omega^{-1})$, where the precision matrix $\Omega$ corresponds to one of four graph types that are commonly used in Gaussian graphical model selection \citep[see, e.g.,][]{liu2017tiger}: band graphs, block graphs, Erd\H{o}s–R\'{e}nyi graphs and cluster graphs.

In particular, we let $\Omega := \Omega^{o} +(|\lambda_{\min}(\Omega^{o} )| + 0.5)I $, where $\lambda_{\min} (\Omega^{o} )$ is the minimum eigenvalue of $\Omega^{o}$. This construction ensures that the precision matrix is positive definite. The respective   $\Omega^{o}$ matrices are generated as follows:
\begin{itemize}
	\item[(i)] Band graph: $\Omega^{o}_{i,i}=1$ for $i=1,\dots, p$ and $\Omega^{o}_{i,j} = sign(b)\cdot |b|^{|i-j|/10} \cdot \bbone_{|i-j|\leq 10}$ for all $i \neq j$.
	\item[(ii)] Block graph: $\Omega^{o}$ is a block diagonal matrix consisting of $10$ blocks. Each block represents a fully connected graph of size $20$ with diagonal entries $1$ and off-diagonal entries $b$. 
	\item[(iii)] Erd\H{o}s–R\'{e}nyi graph: $\Omega^{o}_{i,i}=1$ and $\Omega^{o}_{i,j} = \Omega^{o}_{j,i} = \omega_{i,j} \cdot \phi_{i,j} $, for $i=1,\dots, p$ and $j<i$, where $\omega_{i,j}$ is drawn independently and uniformly from $[-0.6, -0.2]\cup[0.2,0.6]$ and $\phi_{i,j}$ is drawn independently from Bernoulli$(1/10)$.
	\item[(iv)] Cluster graph: $\Omega^{o}$ is a block diagonal matrix consisting of $5$ blocks. Each block is of size $40$ and is generated as the above Erd\H{o}s–R\'{e}nyi graph but with Bernoulli$(1/2)$. 
\end{itemize}
In all cases, we consequently permute the ordering of the variables, to break the pattern of the matrix. All graphs are comparable in the sense that the proportion of edges is about $0.1$. Throughout, we use $p=200$ variables, while the sample size $n$ is varied between $1500$ and $4000$, the edge parameter $b$ is varied between $-0.9$ and $0.9$, and the nominal FDR level $q$ is varied between $0.1$ and $0.4$.

%We first set $a=-0.6$ for band and block graphs, and investigate effect of sample size $n$ and nominal FDR level $q$ on the empirical performance of each method. Specifically, we consider the following settings for all four graphs:
%\begin{itemize}
%	\item Fix $q=0.2$ and let $n = 1500, 2000, 2500, 3000, 3500, 4000$.
%	\item Fix $n=3000$ and let $q=0.1, 0.15, 0.2, 0.25, 0.3, 0.35, 0.4$.
%\end{itemize}
%
%For band and block graphs, we further investigate the effect of the edge parameters $a$ on the empirical performance of each method via:
%\begin{itemize}
%	\item Fix $n=3000$ and $q=0.2$, and let $a=\pm 0.1, \pm 0.2, \pm 0.3, \pm 0.4, \pm 0.5, \pm 0.6, \pm 0.7, \pm 0.8, \pm 0.9$.
%\end{itemize}

We compare the following six methods for structure learning while aiming for FDR control: 
\begin{itemize}
	\item[(a)] BH: Partial correlations and the Benjamini-Hochberg procedure \citep{benjamini1995controlling}. We first obtain  p-values from two-sided tests for zero partial correlations $\rho_{ij \given [p] \backslash \{i,j\}}$, for $i=1,\dots, p$ and $j<i$. We then apply the Benjamini-Hochberg procedure to these p-values. %We use the R-package ``ppcor" developed by \cite{kim2015ppcor} to obtain the p-values in our simulations. 
	\item[(b)] BY: Partial correlations and the Benjamini-Yekutieli procedure \citep{benjamini2001control}. The p-values are obtained as in (a). We then apply the Benjamini-Yekutieli procedure.
	\item[(c)] GFC-L and GFC-SL: two methods proposed by \cite{liu2013gaussian} for high-dimensional Gaussian graphical models. We use the R-package ``SILGGM" developed by \cite{zhang2018silggm} with default values of the tuning parameters.  
	\item[(d)] KO2: Knockoff-based method proposed by \cite{YuKaufmannLederer19}. We use the R code from ``https://github.com/LedererLab/GGM-FDR" for its implementation.
	\item[(e)] GKF-Re+: the GGM knockoff filter with sample-splitting-recycling. We run Algorithm \ref{Algo: GGMknockWithSplittingRecycling} with $\delta=1$. The Hyperparameter space for GKF-Re+ is chosen as follows:
	$\bA = \{(1,1.93), (0.01, 102)\}$,
	$\bmO = \{ \mO_{\equi}, \mO_{\sdp} \}$,
	$\bmR = \{ \mR_{\AND}, \mR_{\OR} \}$, 
	and
	$\bmP=\{ \mP(\alpha) \text{ with } \mP \in \{\eqref{LambdaMax}+\eqref{WfromZ1}, \eqref{LambdaMax}+\eqref{WfromZ2} \}, 
	\mP(\alpha,\lambda) \text{ with } \mP \in \{ \eqref{Coef}+\eqref{WfromZ1}, \eqref{Coef}+\eqref{WfromZ2} \} \}$.
	For $\alpha$, we use $\alpha \in \{ 0.2, 0.4, 0.6, 0.8, 1 \}$. For $\lambda$, it is not practical to directly pick a set of possible values for $\lambda$ in (\ref{Coef}), as the range of $\lambda$'s in which the regression coefficients are nonzero is generally unknown in advance. Therefore, for each value of $\alpha$, we take the $\{ 0.1, 0.2, 0.3, 0.4, 0.5, 0.6, 0.7, 0.8, 0.9, 1\}$  quantiles of the $\lambda$-vector returned by the R package ``glmnet" \citep{friedman2010regularization}. One can verify that the feature statistics constructed in this way satisfy the required antisymmetry and sufficiency properties. 
	We consider all combinations of the above choices, yielding $880$ different choices of $ ((a,c_a), \mO, \mP, \mR) $ in total. 
\end{itemize}

We use 100 replications for each setting. For each estimated graph $\widehat E$ of $E$, we store the resulting false discovery proportion $\FDP = |\widehat E  \setminus E| / ( |\widehat E| \vee 1)$ and the true positive proportion $\TPP = | \widehat E \cap E| / (|E| \vee 1)$ to compute the empirical FDR and power.

%%%%%%%%%%%%%%%%%%%%%%%%%%%%%%%%%%%%%%%%%%%%%%%%%%%%%%%%%%%%%%%%%%%%%%%%%%%%%%%%

\subsection{Simulation results}

The simulation results with varying sample size $n$, nominal FDR level $q$ and edge parameter $b$ are displayed in Figures \ref{MainFig:VaryN}, \ref{MainFig:VaryFDR} and \ref{MainFig:VaryA}, respectively. Each plot shows the empirical FDR, obtained by averaging the 100 FDP realizations, as well as the empirical power, obtained by averaging the 100 TPP realizations, for the given settings and the six different methods. The plots also include dashed vertical bars, indicating plus/minus one empirical standard deviation of the FDP and TPP respectively, to visualize the concentration of the FDP and TPP around their means.

\begin{figure}[t!]
	\centering
	\includegraphics[width=13cm, height=13cm]{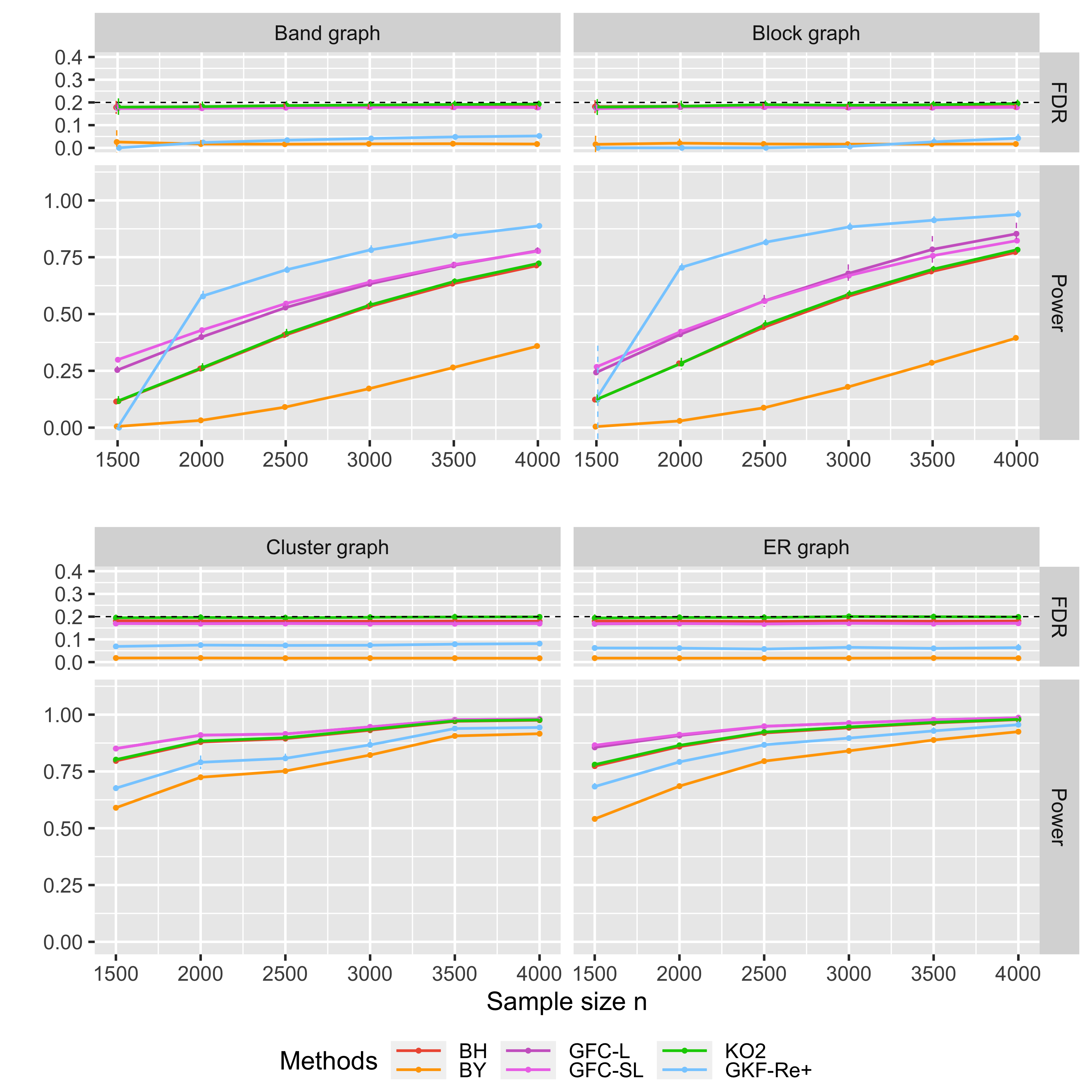}
	\caption{Simulation results: the empirical FDR and power of the six considered methods on the four graph types when varying the sample size $n$, while the nominal FDR level $q=0.2$ and the edge parameter $b=-0.6$ for band and block graphs. The dashed vertical bars indicate plus/minus one empirical standard deviation of the FDP and TPP.}
	\label{MainFig:VaryN}
\end{figure}
\begin{figure}[t!]
	\centering
	\includegraphics[width=13cm, height=13cm]{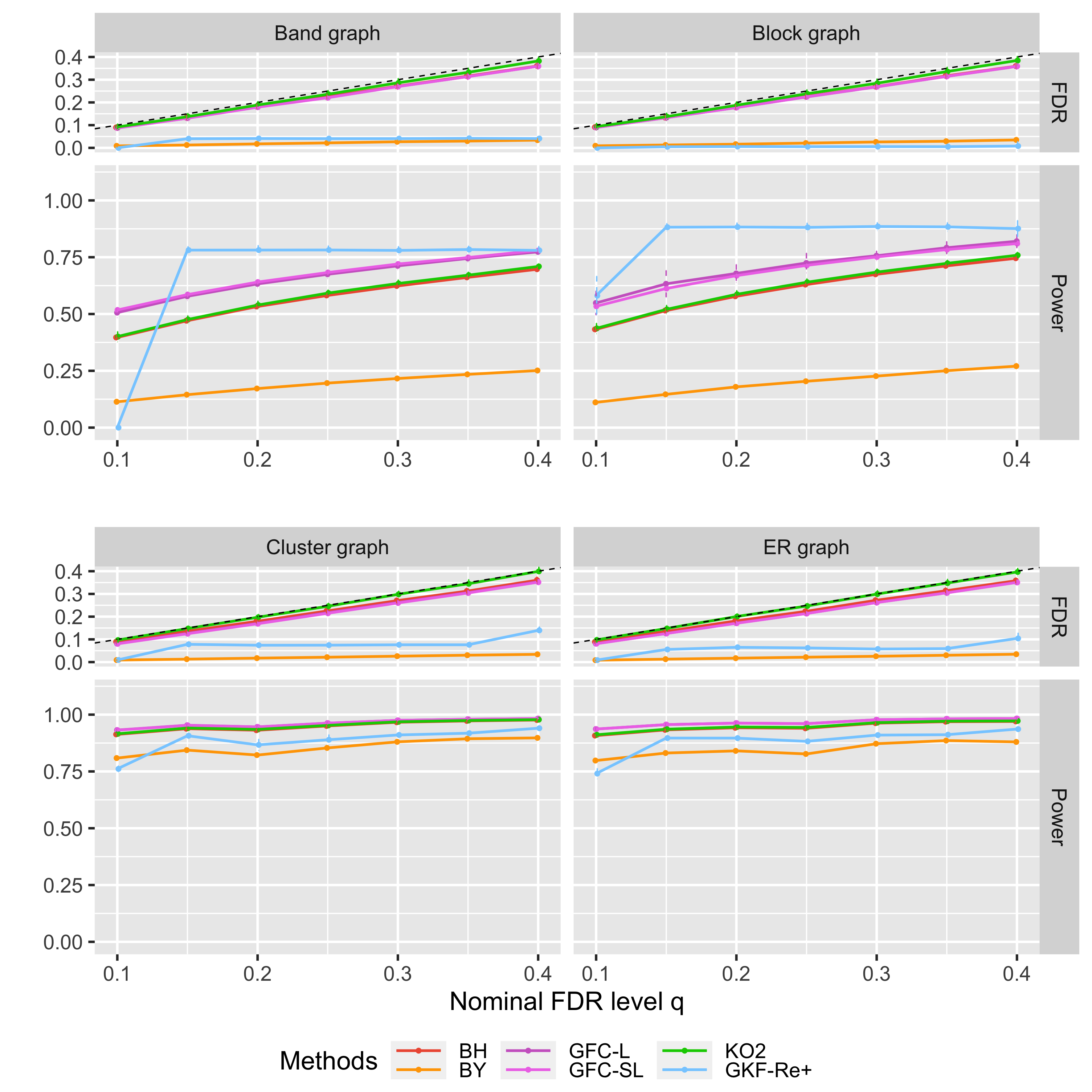}
	\caption{Simulation results: the empirical FDR and power of the six considered methods on four graph types when varying nominal FDR level $q$, while the sample size $n=3000$ and the edge parameter $b=-0.6$ for band and block graphs. The dashed vertical bars indicate plus/minus one empirical standard deviation of the FDP and TPP.}
	\label{MainFig:VaryFDR}
\end{figure}
\begin{figure}[t!]
	\centering
	\includegraphics[width=14cm, height=8cm]{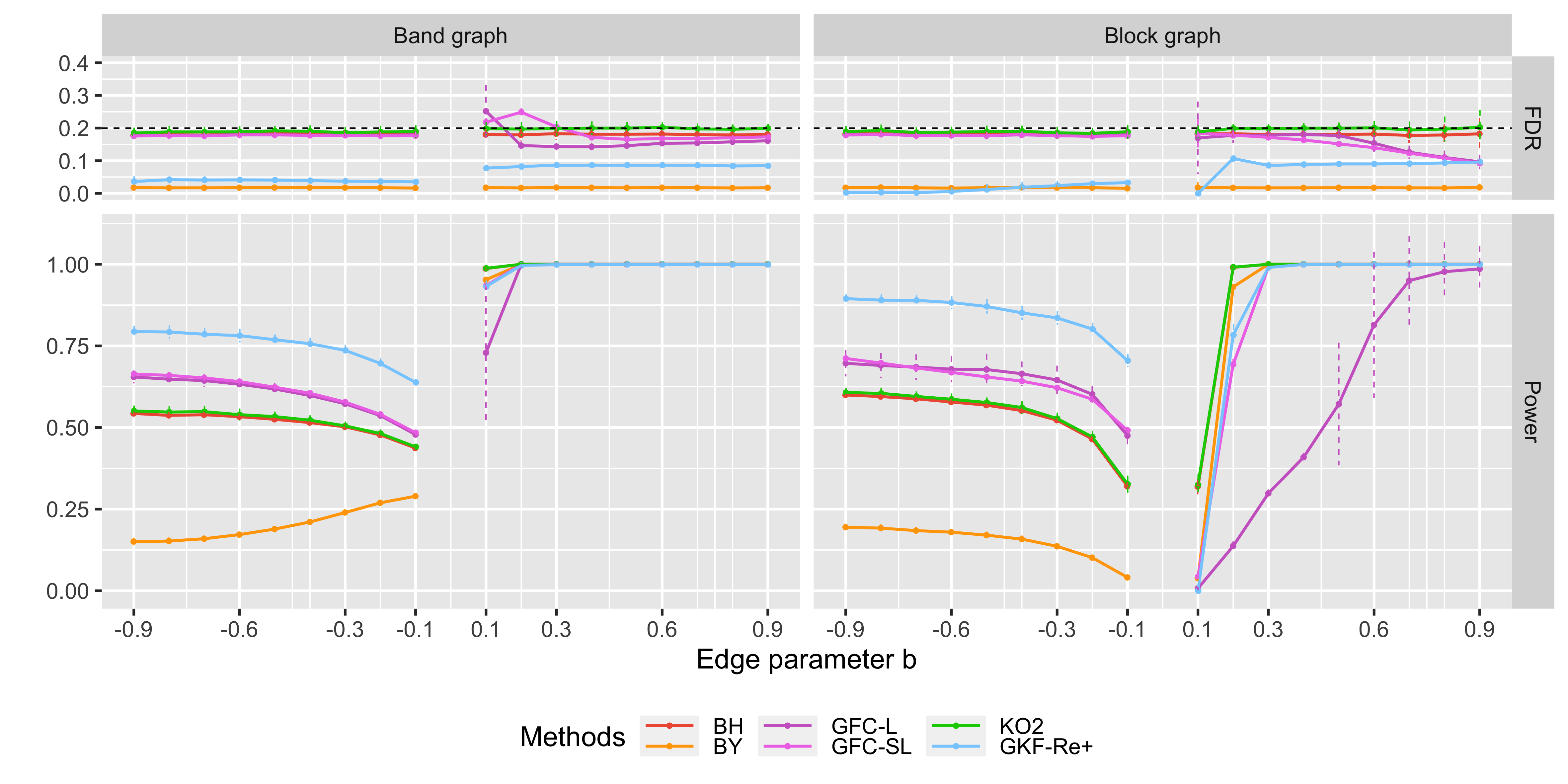}
	\caption{Simulation results: the empirical FDR and power of the six considered methods on band and block graphs when varying the edge parameter $b$, while the nominal FDR level $q=0.2$ and the sample size $n=3000$. The dashed vertical bars indicate plus/minus one empirical standard deviation of the FDP and TPP.}
	\label{MainFig:VaryA}
\end{figure}

We first look at the empirical FDR results. As expected, the two methods with finite sample FDR control guarantee, BY and GKF-Re+, control the FDR in all settings. Both methods are conservative, in the sense that their empirical FDRs are significantly smaller than the nominal FDR. 
%This could be due to the theoretical FDR control guarantee for the worst possible case. 
GFC-L and GFC-SL control the FDR in most settings, but they lose control for the band graph with $b=0.1$ or $b=0.2$ (see Figure \ref{MainFig:VaryA}). This is not surprising, because they only provide asymptotic FDR control under some regularity conditions.
Although there is no theoretical finite sample FDR control guarantee for BH and KO2, these methods seem to control the FDR successfully in our simulation settings. This could be related to results of \cite{clarke2009robustness}, who showed that difficulties with dependence tend to be less severe in settings with a large number of tests and light-tailed test statistics. It would be interesting to investigate this further, to see if these methods lose control in some settings, or if theoretical guarantee can be shown under some assumptions. 

We now turn to the results about empirical power. When the sample size $n$ is small (see the band and block graphs with $n=1500$ in Figure \ref{MainFig:VaryN}), or the nominal FDR level $q$ is small (see the band and block graphs with $q=0.1$ in Figure \ref{MainFig:VaryFDR}), or the signal to noise ratio is small (see the block graph with $b=0.1, 0.2$ in Figure \ref{MainFig:VaryA} ), the TPPs of GKF-Re+ tend to vary widely and the empirical power of GKF-Re+ decreases. In particular,  the TPP may equal zero, which happens if there is no feasible point for the optimization problem in Algorithm \ref{algo2}. 
In the other settings, GKF-Re+ performs quite well and outperforms BY, while it is sometimes better and sometimes worse than the methods without proven finite sample FDR control (BH, K02, GFC-L and GFC-SL).
%There is no universal winner across all settings: the empirical power of  is sometimes better and sometimes worse than that of GKF-Re+.
%Other methods can have non-zero, although might be quite small, power. 
%KO2 has a similar empirical power as BH over all settings in our simulations. This might be because both methods are similar in the sense that they use sample partial correlations.  
%Instead, GKF-Re+ does not rely on the sample partial correlation. Moreover, due to its flexibility and the use of sample splitting to find appropriate hyperparameters, it can achieve good performance over a wide range of settings. 
% and is very flexible, which makes it possibly have better performance in various different settings.

GKF-Re+ is especially good at handling the band and block graphs with $b<0$, in which its empirical power can be more than $3$ times larger than that of BY, and it also greatly outperforms BH, GFC-L, GFC-SL and KO2 (see the results related to band and block graphs in Figure \ref{MainFig:VaryN}, \ref{MainFig:VaryFDR} and \ref{MainFig:VaryA}). Such settings are special cases of the so-called multivariate totally positive of order two ($\text{MTP}_2$) distribution. 
%which is defined as a multivariate Gaussian distribution with positive diagonal and non-positive off-diagonal entries in its precision matrix. 
Under this distribution, positive conditional dependence of two variables given all of the remaining variables implies a positive conditional dependence given any subset of the remaining variables. This property makes it interesting and important (see \cite{fallat2017total}).

Additional simulation results can be found in the appendix. Appendix \ref{appendix: Oracle} compares the performance of GKF+ (Algorithm \ref{Algo: GGMknock} with $\delta=1$) with oracle hyperparameters to that of GKF-Re+, and Appendix \ref{appendix: mFDR} investigates the performance of the mFDR controlled counterpart of GKF-Re+.

%%%%%%%%%%%%%%%%%%%%%%%%%%%%%%%%%%%%%%%%%%%%%%%%%%%%%%%%%%%%%%%%%%%%%%%%%%%%%%%%%%%%%%%%%%%%%%%%%%%%%%%%%
%%%%%%%%%%%%%%%%%%%%%%%%%%%%%%%%%%%%%%%%%%%%%%%%%%%%%%%%%%%%%%%%%%%%%%%%%%%%%%%%%%%%%%%%%%%%%%%%%%%%%%%%%
\subsection{Real data application} \label{SecRealData}

We now apply the same six methods with nominal FDR level $q=0.2$ to a public single-cell RNA-seq dataset \citep[S2 Appendix]{zheng2017massively}.
These data were analyzed by \cite{zhang2018silggm} using Gaussian graphical models to recover the gene network. 
% Based on the raw data, \cite{zhang2018silggm} first filtered out unexpressed genes and then took the $2000$ genes with the largest sample variances 
As done by \cite{zhang2018silggm}, we preprocess the data by performing a $\log_2$(counts+1) transformation and a nonparanormal transformation \citep{liu2009nonparanormal}.
Moreover, we reduce the number of considered genes to $50$ by keeping the $50$ genes with the largest sample variances. 
Although it is hard to verify multivariate Gaussianity, we verify marginal Gaussianity of each gene by looking at the corresponding normal QQ-plot (see Appendix \ref{appendix: QQplot}). 
%The QQ-plot (see Figure \ref{QQ-plot}) in Appendix \ref{appendix: QQplot} shows that the marginal distribution of each gene is close to Gaussian distribution, but note that this does not fully justify the multivariate Gaussian assumption. 

To investigate the effect of the randomness coming from sample splitting in GKF-Re+, we randomly split $20$ times and aggregate these results to obtain a final estimation for GKF-Re+ by taking the edges that are discovered in more than $50\%$ of the splits. We denote the aggregation procedure by GKF-Re+(agg) and note that there is no theoretical FDR control guarantee for the aggregated result. 
We also present the result of GKF-Re+ for the first random split.

For visualization purposes, we assign each edge an index via a score defined by
\begin{equation*}
	\text{score} = \text{score}_{\text{BH}} + \text{score}_{\text{BY}} + \text{score}_{\text{GFC-L}} + \text{score}_{\text{GFC-SL}} + \text{score}_{\text{KO2}} + \text{score}_{\text{GKF-Re+}},
\end{equation*}
where $\text{score}_{\text{GKF-Re+}}$ is the percentage of splits in which this edge is discovered, while the other five terms are set to be $1$ if this edge is discovered by the respective method and $0$ otherwise. The edge index, which is used in the following two plots, is then obtained by ordering these scores in a decreasing manner. Thus, a small edge index indicates that the edge is discovered by many methods. Since the ground truth of the underlying graph is unknown, we compare the edge sets discovered by the different methods.

The results are shown in Figure \ref{fig:realdata2}. The methods are ordered by the number of edges they discovered (BY: $252$, GKF-Re+: $425$, GKF-Re+(agg): $443$, BH: $513$, GFC-L: $518$, KO2: $572$, GFC-SL: $1225$). 
We see that BY discoveres the smallest number of edges, most of which are also discovered by all other methods. 
GKF-Re+ discoveres $173$ edges more than BY, and most of the edges with index roughly between $250$ to $400$ are also discovered by the remaining methods.
Both BY and  GKF-Re+ provide finite sample FDR control guarantee.
GKF-Re+(agg) returns $18$ edges more than GKF-Re+.
BH and GFC-L return very similar edge sets, that are somewhat larger than the one from GKF-Re+(agg). KO2 also returns most of these edges, as well as roughly $50$ additional ones. Finally, GFC-SL returns a fully connected graph. 
We see that edges with index around $250$ are not discovered by the two knockoff based methods, while the remaining methods do discover them. 
On the other hand, some edges with index around $600$ are discovered by the two knockoff based methods, but are not discovered by most of the other methods.
It would be interesting to further investigate if such edges are biologically meaningful. 

\begin{figure}[t!]
	\centering
	\includegraphics[width=12cm, height=6cm]{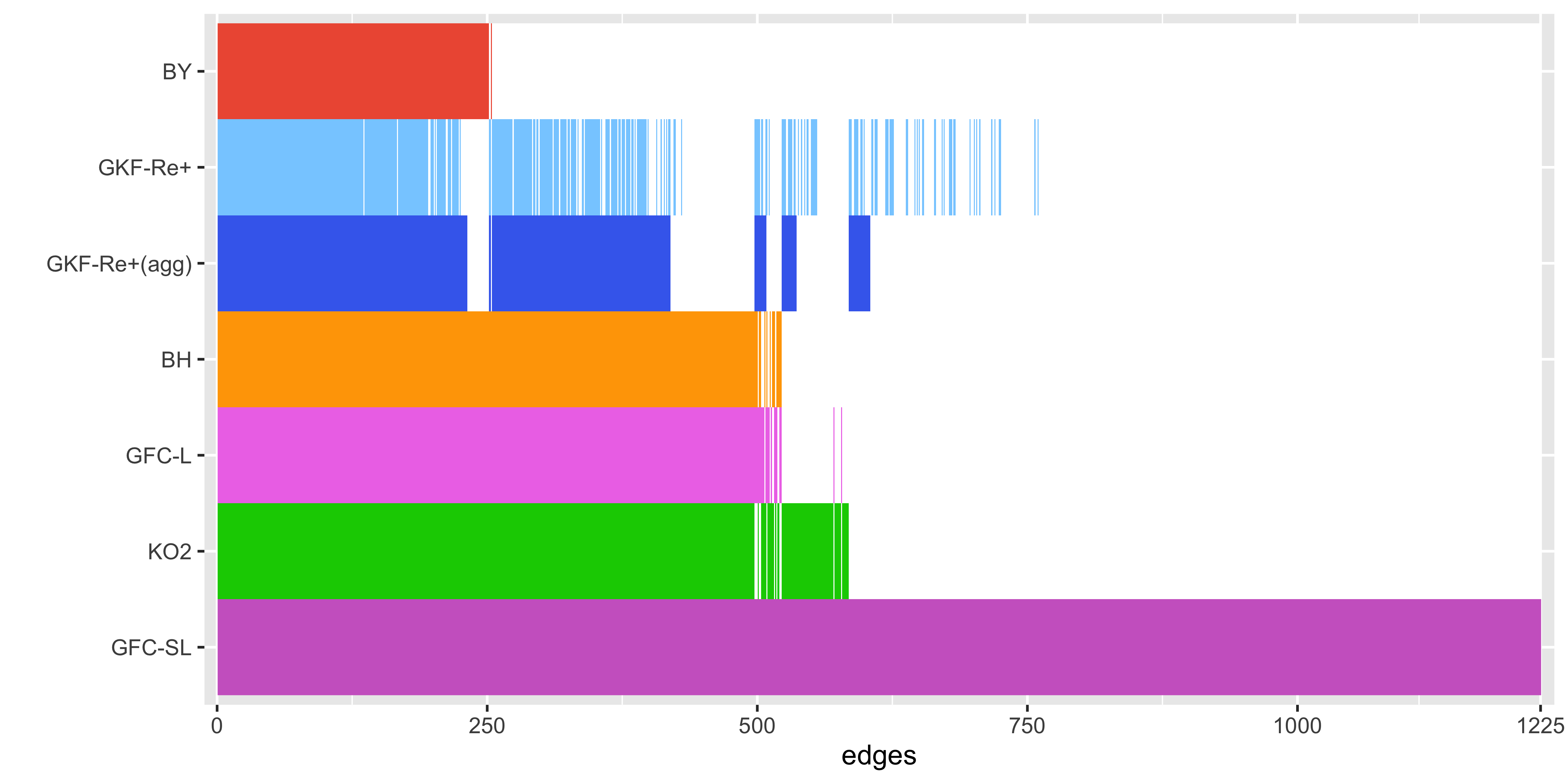}
	\caption{Results for real data: discovered edges of the seven methods. GKF-Re+ shows the results based on the first random split.}
	\label{fig:realdata2}
\end{figure}

The results of GKF-Re+ over the $20$ random splits are displayed in Figure \ref{fig:realdata1}. As expected, the discovered edges are different for different splits, but roughly have a similar pattern. 
\begin{figure}[t!]
	\centering
	\includegraphics[width=12cm, height=6cm]{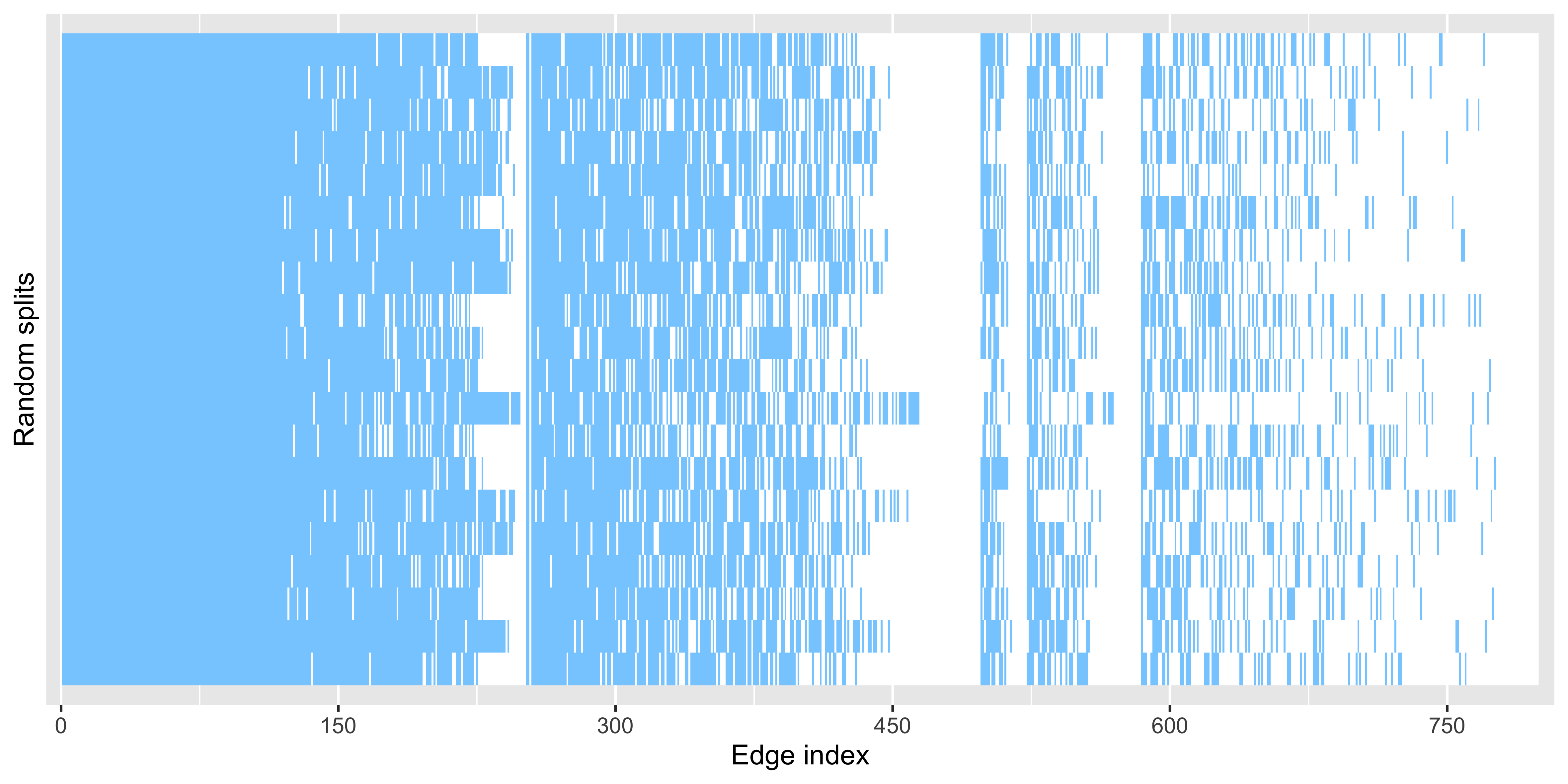}
	\caption{Results for real data: discovered edges of GKF-Re+ over $20$ random splits. Specifically, each line corresponds to the results of one random split, and the blue color indicates the discovered edges.}
	\label{fig:realdata1}
\end{figure}

%% file: Sections/Sec6-Discussion.tex
\section{Discussion} \label{sec:disc}

In this paper, we consider the problem of controlling the finite sample FDR in Gaussian graphical model structure learning. 
To the best of our knowledge, the only existing method achieving this goal up to now relies on p-values obtained by testing for zero partial correlations, combined with the Benjamini-Yekutieli procedure. 
Our approach builds on the knockoff framework of \cite{barber2015controlling}. In particular, we introduce the GGM knockoff filter for fixed hyperparameter settings (Algorithm \ref{Algo: GGMknock}).
Knockoff based methods are very flexible. Although this flexibility is a desirable property, it also causes an important and practical problem: how should one set the hyperparameters to achieve the best statistical power while maintaining FDR control? 
In this paper, we use a sample-splitting-recycling approach, motivated by \cite{barber2016knockoff}, to address this issue and propose the GKF-Re+ procedure (Algorithm \ref{Algo: GGMknockWithSplittingRecycling} with $\delta=1$).

We would like to point out the limitation that the GGM knockoff filter is not suitable for very sparse graphs.
For example, consider optimization problem \eqref{ANDopt} (similar arguments hold for \eqref{ORopt} ) and the favorable case where there exists a threshold vector $T$ such that $|\widehat{E}_{\AND}(T)| = |E| = \gamma p(p-1) / 2$ and $|\widehat{V}^{-}_{i}(T_i)|=0$ for all $i\in[p]$, where $\gamma \in [0,1]$ is the sparsity level of the true graph. 
Then, a feasible point for the optimization problem can only exist if $\gamma \geq ac_a/((p-1)q)$. 
Hence, the true underlying graph can not be too sparse and the minimal sparsity level depends on the dimension $p$ and the nominal FDR level $q$.
Such issues exist generally for knockoff based approaches. For example, for a linear model with dimension $p$ and sparsity level $\gamma$ of the regression coefficient vector $\beta$, consider the favorable case where there exists a threshold $T$ such that $|\widehat{S}(T)| = |S| = \gamma p$ and $| \{ i \in [p]: W_i \leq -T\} | = 0$. 
Then, a feasible point for optimization problem \eqref{OriKnockThreEquation} in the original knockoff method (Algorithm \ref{oriKnock}) can only exist if $\gamma \geq 1/(pq)$.
One can see that the two lower bounds for the sparsity level $\gamma$ in GGMs and linear models are quite similar, especially when using $(a,c_a)=(0.01,102)$ (so $ac_a=1.02 \approx 1$) for the GGM knockoff filter. 

As we mentioned before, in addition to the fixed-X knockoffs used in this paper, one can also use the model-X knockoffs following the method in \cite{huang2020relaxing}. One extra benefit of using model-X knockoffs is that the sufficient property is no longer needed for constructing feature statistics, so there is more flexibility in the construction of feature statistics. In other words, we have even more options for the hyperparameters in our proposed method by incorporating model-X knockoffs. 

For the sample-splitting-recycling approach, one can freely use the information from the first part of sample, such as the estimated sign in \cite{barber2016knockoff}. This might improve the power if the used information is correct and might decrease power if the used information is incorrect.

We close by pointing out some possible directions for future work.

In addition to Gaussian settings, \cite{huang2020relaxing} also proposed a method to construct valid model-X knockoffs for discrete graphical models. Therefore, by combining their model-X knockoffs and the idea of our paper, it seems possible to extend the FDR controlled structure learning method to discrete graphical model settings. It would be interesting to investigate this.

Theoretical analysis of the power of the knockoff methods is an interesting problem. There is some work in this direction regarding variable selection \citep[e.g.,][]{weinstein2017power, liu2019power, fan2020rank, weinstein2020power, wang2020power}. It would be interesting to conduct such power analysis for our proposed method in the structure learning setting.

Similar to \cite{barber2016knockoff}, GKF-Re+ can be easily extended to high-dimensional settings by implementing a screening step in the first part of the sample, that is, the same part we use to choose the hyperparameters. The screened graph should satisfy that the number of neighbors of each node is less than a quarter of the sample size, so that knockoffs can be constructed in a node-wise manner.
If the true graph is a subgraph of the screened graph, then Lemma \ref{Lemma: SplitRecyclingColumnSignFlip} still holds, as well as the FDR control guarantee. 
One potential drawback of this approach is that there are still $p$ constraints in the corresponding optimization problem, and a feasible point may not exist when $p$ is large.
One possible idea to reduce the number of the constraints in the optimization problem is to adapt the constraints based on the screened graph. For example, if the screened graph is a star shaped graph in which one central node is connected to the remaining nodes, then one constraint obtained from treating this central node as response would be enough.

\cite{barber2016knockoff} showed that the knockoff procedure (Algorithm \ref{oriKnock}) can control the directional FDR for linear models, when the signs of $\beta$ are estimated as $\sign((\bX-\widetilde{\bX})^T \by)$. In the Gaussian graphical model setting, things are more tricky as one needs to deal with the case that $\widehat{\sign}(\beta^{(i)}_{j}) \neq \widehat{\sign}(\beta^{(j)}_{i})$. Moreover, Lemma \ref{Lemma: Bound of c_a}, a crucial element in our FDR control proof, does not apply in this setting. Hence, new results are needed to prove directional FDR control for Gaussian graphical models.

\cite{katsevich2020simultaneous} proposed a way to provide high-probability bounds on the false discovery proportion for many FDR control procedures such as BH and the knockoff based method. It would be interesting to investigate if such bounds can be applied to our proposed GGM knockoff filter for a fixed hyperparameter setting. 

Finally, the choice of hyperparameters is a fundamental problem for all knockoff based methods. 
One can even ask a sightly more general question: given many FDR control procedures (including knockoff based and p-values based approaches), how should one choose the best FDR control procedure for a given data set to ensure good power while maintaining FDR control? Sample-splitting-recycling is one possible answer, but there are two undesired drawbacks of this approach: randomness of splits and the loss of power. Better solutions to this problem are desired.

%% file: Sections/LastParts.tex
%\section*{Supplement material} \label{sec:supp}

\section*{Acknowledgments} \label{sec:ack}
We thank Yuansi Chen and Leonard Henckel for comments on earlier versions of this paper and helpful discussions. We especially thank Yuansi for suggesting the idea related to Lemma \ref{Lemma: Bound of c_a}.
We thank the AE and the reviewers for their valuable comments.

%%%%%%%%%%%%%%%%%%%%%%%%%%%%%%%%%%%%%%%%%%%%%%%%%%%%%%%%%%%%%%%%%%%%%%%%%%%%%%%%%%%%%%

%% file: Sections/Supplements.tex
\begin{center}
{\large\bf SUPPLEMENTARY MATERIAL}
\end{center}

The supplementary material consists of the following five appendices.
	\begin{description}
		\item[A] An example illustrating the difference between node-wise and graph-wise FDR control
		\item[B] Some concrete examples of $W$
		\item[C] Some intuition behind equation \eqref{OriKnockThreEquation}
		\item[D] Algorithm to compute the threshold vector $\widehat{T}$
		\item[E] Technical details
		\item[F] Details of the simulations in Section \ref{sec: sample-splitting-recycling}
		\item[G] Additional simulation results
	\end{description}

%%%%%%%%%%%%%%%%%%%%%%%%%%%%%%%%%%%%%%%%%%%%%%%%%%%%%%%%%%%%%%%%%%%%%%%%%%%%%%%%%%%%%%%%%%%%%%
\section{An example illustrating the difference between node-wise and graph-wise FDR control} \label{appendix: NodeGraphWiseFDRexample}
Let $G=(V,E)$ be an empty graph with node set $V = [p]$ and edge set $E=\emptyset$. For a given nominal FDR level $q \in [0,1]$, consider the following procedure: if $q \leq 2/p$, return $\widehat{E} = \emptyset$; if $q > 2/p$, return $\widehat{E}$ by randomly choosing a set from $\{(1,2), (2,1)\}, \{(2,3), (3,2)\}, \dots, \{(p-1,p), (p,p-1)\}, \{(p,1),(1,p)\}$. Then, in the first case, the node-wise FDR of each node and the graph-wise FDR are both $0\leq q$. In the second case, all of the node-wise FDRs are equal to $2/p < q$, whereas the graph-wise FDR equals $1\geq q$. Therefore, this procedure controls the node-wise FDR for all $q\in[0,1]$ but loses the graph-wise FDR for $q > 2/p$. This example shows that node-wise and graph-wise FDR control are two distinct tasks, and the method guarantees to control the node-wise FDR may not be able to control the graph-wise FDR.

%%%%%%%%%%%%%%%%%%%%%%%%%%%%%%%%%%%%%%%%%%%%%%%%%%%%%%%%%%%%%%%%%%%%%%%%%%%%%%%%%%%%%%%%%%%%%%
\section{Some concrete examples of $W$} \label{appendix:ExamplesOfW}
We now present some concrete examples of $W$. Consider the Lasso estimator for model \eqref{eq: linear model extended}:
\begin{equation*}
\hat{\beta}(\lambda) = \underset{b \in \bbR^{2p}}{\operatorname{\argmin}} \ \left( \frac{1}{2} \left \lVert \by - [\bX \ \widetilde{\bX}] b \right \rVert^2_2 
	+ \lambda \lVert b \rVert_1   \right) , \quad \ \lambda \geq 0.
\end{equation*}
As an intermediate step, we first construct $Z_i$ and $\widetilde{Z}_i$ to measure the importance of $X_i$ and $\widetilde{X}_i$ to the response in model \eqref{eq: linear model extended}. For a fixed $\lambda \geq 0$, it is expected that most of the non-null variables have larger coefficients (in absolute value) than their knockoff counterparts. Hence, we can use
% \footnote{In order for the feature statistics to satisfy the antisymmetry and sufficiency properties, cross-validation cannot be used to choose tunning parameters $\lambda$.}
\begin{equation*}
Z_i = |\hat{\beta}_i(\lambda)|
\quad \text{and} \quad
\widetilde{Z}_i = |\hat{\beta}_{i+p}(\lambda) |.
\end{equation*}
Alternatively, one could use 
\begin{equation*}
Z_i = \sup \{ \lambda \geq 0: \hat{\beta}_i(\lambda) \neq 0 \}
\quad \text{and} \quad
\widetilde{Z}_i = \sup \{ \lambda \geq 0: \hat{\beta}_{i+p}(\lambda) \neq 0 \}, 
\end{equation*}
which is based on the idea that non-null variables tend to enter the Lasso path before their knockoff counterparts. 
In either case, $W_i$ can be obtained by
\begin{equation*}
W_i = ( Z_i \vee \widetilde{Z}_i ) \cdot \sign(Z_i - \widetilde{Z}_i)
\quad \text{or} \quad
W_i = Z_i - \widetilde{Z}_i.
\end{equation*}
By the above construction, a positive $W_i$ indicates that $X_i$ is chosen over $\widetilde{X}_i$, while a negative $W_i$ indicates the opposite. Hence, since the $\widetilde{X}_i$'s are null variables in model \eqref{eq: linear model extended}, the knockoff method selects variables $X_i$ with a large and positive $W_i$, i.e., $\widehat{S} = \{ i \in[p]: W_i \geq T \}$ for some positive threshold $T$. For more examples on the construction of $W$, we refer to \cite{barber2015controlling}.

%%%%%%%%%%%%%%%%%%%%%%%%%%%%%%%%%%%%%%%%%%%%%%%%%%%%%%%%%%%%%%%%%%%%%%%%%%%%%%%%%%%%%%%%%%%%%%
\section{Some intuition behind equation \eqref{OriKnockThreEquation}} \label{appendix:intuitionLinearKnockoff}
We now present some intuition behind equation \eqref{OriKnockThreEquation}. Since $W$ possesses the sign-flip property on $S^c$, we have for any fixed threshold $T>0$:
\begin{equation*}
| \{j \in S^c: W_j\ge T\} |
\approx
| \{j \in S^c: W_j\le -T\} |.
\end{equation*}
Hence, the false discovery proportion satisfies
\begin{equation*}
\begin{aligned}
\FDP & = \frac{|\widehat S \cap S^c|}{|\widehat S|\vee 1} = \frac{ | \{j \in S^c: W_j\ge T\} | } {|\widehat S| \vee 1} \\
& \leq
\frac{ |\{j \in S^c: W_j\ge T\}} | {1 + | \{j \in S^c: W_j\le - T\} | } \cdot
\frac{1 + |\{j \in [p]: W_j\le - T\} | } {|\widehat S| \vee 1} \\
& \approx 
\frac{ | \{j \in S^c: W_j\le -T\} | } {1 + | \{j \in S^c: W_j\le - T\} | } \cdot
\frac{1 + | \{j \in [p]: W_j\le - T\} | } {|\widehat S| \vee 1} \\
& <
\frac{1 +| \{j \in [p]: W_j\le - T\} | } {|\widehat S| \vee 1} = \frac{ 1 + | \{ j \in [p]: W_j \leq -T \} | } { | \{  j \in [p]: W_j \geq T \} | \vee 1 }.
\end{aligned}
\end{equation*}
The last line is exactly what we used in formula \eqref{OriKnockThreEquation}.

%%%%%%%%%%%%%%%%%%%%%%%%%%%%%%%%%%%%%%%%%%%%%%%%%%%%%%%%%%%%%%%%%%%%%%%%%%%%%%%%%%%%%%%%%%%%%%
\section{Algorithm to compute the threshold vector $\widehat{T}$} \label{appendix: optimization problem}
At first glance, \eqref{ANDopt} and \eqref{ORopt} (as well as \eqref{ANDopt} with replacement \eqref{mFDRoptAND} and \eqref{ORopt} with replacement \eqref{mFDRoptOR}) are both combinatorial optimization problems, and it is seemingly infeasible to solve them through a brute-force search when $p$ is large. Due to the structure of the constraints, however, the search spaces can be significantly restricted and there is in fact a simple algorithm that can efficiently find the optimal solutions. 

In following, we focus on \eqref{ANDopt}, but similar arguments hold for other cases.
% and give two observations that allow for a significant restriction of the search space. Similar arguments hold for \eqref{ORopt}.
First, since the maximum number of edges of an undirected graph is $|E_{\max}| = p(p-1)/2$, a necessary condition for the feasibility of $T=( T_1, \dots, T_p )$ in \eqref{ANDopt} is that
\begin{equation*}
|\widehat{V}^{-}_{i}(T_i)| 
\leq |E_{\max}| \cdot 2q/(c_a p)-a
= q(p-1)/c_a-a,
\quad \text{for all } i \in [p].
\end{equation*}
Let $m_{\max}$ be the largest integer that is smaller than or equal to $q(p-1)/c_a-a$. If $m_{\max}<0$, then there is no feasible point for this optimization problem. Now assume that $m_{\max} \geq 0$. If there is an element $T_i$ of $T$ such that $|\widehat{V}^{-}_{i}(T_i)| >  m_{\max}$, we can conclude that any $T^*=( T^*_1, \dots, T^*_p )$ with $T_i^* \leq T_i$ is infeasible, since $ |\widehat{V}^{-}_i(T_i^*)| \geq |\widehat{V}^{-}_{i}(T_i)| >  m_{\max}$, and stop searching in this direction. %This reduces a lo t of the searching space. %For example, when $p=101$, $q=0.3$, we have $m_{\max}=29$.

Second, note that the left hand side of the constraint
\begin{equation*}
\frac{ a + |\widehat{V}^{-}_{i}(T_i)|}{|\widehat{E}_{\AND}(T)| \vee 1} \leq \frac{2q}{c_ap} 
\end{equation*}
is decreasing in $|\widehat{E}_{\AND}(T)|$, and for different $i\in[p]$, the only difference lies in $|\widehat{V}^{-}_{i}(T_i)|$. Hence, if $T=( T_1, \dots, T_p )$ is feasible, that is,
\begin{equation*}
\frac{ a + |\widehat{V}^{-}(T)_{\max}|}{|\widehat{E}_{\AND}(T)| \vee 1} \leq \frac{2q}{c_a p},
\end{equation*}
where $|\widehat{V}^{-}(T)_{\max}| =  \max \left\{ |\widehat{V}^{-}_{i}(T_i)|, i\in [p] \right\}$, we can immediately find a new feasible vector $T^* = (T_1^*, \dots, T_p^*)$ with $T_i^* = \min \Big\{ t_i \in \left\{  |W^{(i)}_j|, j \in [p] \right\} \cup \{+\infty\}  \backslash \{0\}: |\widehat{V}^{-}_{i}(t_i)| \leq |\widehat{V}^{-}(T)_{\max}| \Big\} \leq T_i$, which is a choice that at least as good as $T$, since $|\widehat{E}_{\AND}(T^*)| \geq |\widehat{E}_{\AND}(T)| $. Therefore, it is sufficient to consider only threshold vectors of the form of $T^* $.

Combining these two observations, \eqref{ANDopt} can be solved by checking the feasibility of $T=( T_1, \dots, T_p )$ with $T_i = \min \left\{ t_i \in \left\{ |W^{(i)}_j|, j \in [p] \right\} \cup \{+\infty\}  \backslash \{0\}: |\widehat{V}^{-}_{i}(t_i)| \leq m \right\}$, for $m=m_{\max},m_{\max}-1,\dots,0$, which leads to Algorithm \ref{algo2}. 

%%%%%%%%%%%%%%%%%%%%%%%%%%%%%%%%%%%%%%%%%%%%%%%%%%%%%%%%%%%%%%%%%%%%%%%%%%%%%%%%%%%%%%%%%%%%%%
\section{Technical details} \label{appendix: Proof}
\subsection{Proof of Lemma \ref{SignFlipRandomDesign}} \label{appendix: lemmaRandomDesign}
\begin{restatable}{lemma}{RandomSignFlip}(\textbf{Sign-flip property on $S^c$ with random design matrix}) \label{SignFlipRandomDesign}
	~\\
	Let $\bX \in \bbR^{n \times p}$ be a random design matrix in a linear model, $S^c$ be the index set of null variables and $W$ be a feature statistic vector satisfying both the sufficiency and the antisymmetry properties. Then $W$ possesses the sign-flip property on $S^c$.
\end{restatable}

This lemma trivially follows from the proof of Lemma 1 in \cite{barber2015controlling} by first conditioning on the random design matrix then marginalizing over it, here we write it in detail for completeness. We first introduce two useful lemmas from \cite{barber2015controlling}.

\begin{lemma}(\textbf{Pairwise exchangeability for the features}, \cite{barber2015controlling}) \label{FirstLemmaFromBC}
	~\\
	Let $\bX = [ \bX^{(1)}, \dots, \bX^{(p)}] \in \bbR^{n \times p}$ be a fixed design matrix in a linear model and $\widetilde{\bX} \in \bbR^{n \times p}$ be its knockoff counterpart. Then for any subset $H \subseteq [p]$, the Gram matrix of $ [ \bX \ \widetilde{\bX} ] $ is unchanged when we swap $\bX^{(i)}$ and $\widetilde{\bX}^{(i)}$ for each $ i \in H $. That is,
	\begin{equation*}
	[ \bX \ \widetilde{\bX} ]^T_{\swap(H)} [ \bX \ \widetilde{\bX} ]_{\swap(H)} = [ \bX \ \widetilde{\bX} ]^T [ \bX \ \widetilde{\bX} ].
	\end{equation*}
	Here $[ \bX \ \widetilde{\bX} ]^T_{\swap(H)}$ denotes the matrix obtained by first swapping then transposing matrix $[ \bX \ \widetilde{\bX} ]$.
\end{lemma}

\begin{lemma}(\textbf{Pairwise exchangeability for the response}, \cite{barber2015controlling}) \label{SecondLemmaFromBC}
	~\\ 
	Let $ \by $ be a response vector, $\bX = [ \bX^{(1)}, \cdots, \bX^{(p)}] \in \bbR^{n \times p}$ be a fixed design matrix and $S^c$ be the index set of the null variables in a linear model, and let $\widetilde{\bX} \in \bbR^{n \times p}$ be the knockoff counterpart of $\bX$. Then, for any subset $ H \subseteq S^c$, the distribution of $[ \bX \ \widetilde{\bX} ]^T \by$ is unchanged when we swap $\bX^{(i)}$ and $\widetilde{\bX}^{(i)}$ for each $ i \in H $. That is,
	\begin{equation*}
	[ \bX \ \widetilde{\bX} ]^T_{\swap(H)} \by \stackrel{d}{=} [ \bX \ \widetilde{\bX} ]^T \by. 
	\end{equation*}
\end{lemma}

Now we prove Lemma \ref{SignFlipRandomDesign}. This proof follows the same idea as the proof of Lemma 1 in \cite{barber2015controlling}.
\begin{proof}[Proof of Lemma \ref{SignFlipRandomDesign}]
	Let $\epsilon = (\epsilon_1, \cdots, \epsilon_p)$ be a sign sequence independent of $W$ with $\epsilon_i = +1$ for $i \in S$ and $\epsilon_i$ i.i.d.\ from a Rademacher distribution for $i \in S^c$. Let $K = \{ i \in [p]: \epsilon_i=-1 \}$ be a set depending on $\epsilon$, note that we have $K \subseteq S^c$. Then, by the antisymmetry property of $W$ and the definition of $K$, we have
	\begin{equation*}
	W_{\swap(K)} = (W_1 \cdot \epsilon_1, \cdots, W_p\cdot \epsilon_p).
	\end{equation*}
	Thus, to prove Lemma \ref{SignFlipRandomDesign}, it suffices to show that
	\begin{equation*}
	W_{\swap(K)}  \stackrel{d}{=}  W.
	\end{equation*}
	
	Conditional on $\bX$, we have
	\begin{equation*}
	( [ \bX \ \widetilde{\bX} ]^T_{\swap(K)}  [ \bX \ \widetilde{\bX} ]_{\swap(K)}, \ [\bX \ \widetilde{\bX} ]^T_{\swap(K)}  \by ) \ | \ \bX
	\stackrel{d}{=} ( [ \bX \ \widetilde{\bX} ]^T [ \bX \ \widetilde{\bX} ], \ [\bX \ \widetilde{\bX} ]^T \by ) \ | \ \bX.
	\end{equation*}
	by Lemma \ref{FirstLemmaFromBC} and Lemma \ref{SecondLemmaFromBC}. Thus, for the joint distribution, we have
	\begin{equation*}
	( [ \bX \ \widetilde{\bX} ]^T_{\swap(K)}  [ \bX \ \widetilde{\bX} ]_{\swap(K)}, \ [\bX \ \widetilde{\bX} ]^T_{\swap(K)}  \by, \bX )
	\stackrel{d}{=} ( [ \bX \ \widetilde{\bX} ]^T [ \bX \ \widetilde{\bX} ], \ [\bX \ \widetilde{\bX} ]^T \by, \bX).
	\end{equation*}
	Then, by marginalizing $\bX$ out, we have
	\begin{equation*}
	( [ \bX \ \widetilde{\bX} ]^T_{\swap(K)}  [ \bX \ \widetilde{\bX} ]_{\swap(K)}, \ [\bX \ \widetilde{\bX} ]^T_{\swap(K)}  \by )
	\stackrel{d}{=} ( [ \bX \ \widetilde{\bX} ]^T [ \bX \ \widetilde{\bX} ], \ [\bX \ \widetilde{\bX} ]^T \by)
	\end{equation*}
    Finally, by the sufficiency property of $W$, we have
    \begin{align*}
    	W_{\swap(K)} &= f ( [ \bX \ \widetilde{\bX} ]^T_{\swap(K)}  [ \bX \ \widetilde{\bX} ]_{\swap(K)}, \ [\bX \ \widetilde{\bX} ]^T_{\swap(K)}  \by ) \\
    	&\stackrel{d}{=} f ( [ \bX \ \widetilde{\bX} ]^T [ \bX \ \widetilde{\bX} ], \ [\bX \ \widetilde{\bX} ]^T \by ) = W.
    \end{align*}
\end{proof}

%%%%%%%%%%%%%%%%%%%%%%%%%%%%%%%%%%%%%%%%%%%%%%%%%%%%%%%%%%%%%%%%%%%%%%%%
%%%%%%%%%%%%%%%%%%%%%%%%%%%%%%%%%%%%%%%%%%%%%%%%%%%%%%%%%%%%%%%%%%%%%%%%
\subsection{Derivation of the bound $c_a$} \label{appendix: Bound of c_a}

We first introduce a lemma which can be used to obtain the bound $c_a$ for any $a>0$. We do not claim any originality for the proof of this lemma. The proof is exactly as in \cite{katsevich2017multilayer} with replacing $1$ by $a$ in the denominator of the target term.
\begin{lemma} \label{Lemma: Bound of c_a}
	Let $\widehat{V}^{+}_{N_i}$ and $\widehat{V}^{-}_{N_i}$ be defined as in \eqref{NotationV} with $W^{(i)}$ possessing the sign-flip property on $NE_i^c$. 
	Let $S_k = \sum_{i=1}^k X_i$, where $X_i \stackrel{i.i.d.\ }{\sim} Ber(1/2)$ and $S_0=0$. For a fixed $k_0 \geq 1$ and $a > 0$, let 
	\begin{equation*}
		R_{k_0}(x,a) = \max_{k \leq k_0} \frac{\sum_{i \leq k} x_i}{a + k - \sum_{i \leq k} x_i}
		\quad \text{and} \quad
		P_{k_0}(x) = \sum_{i \leq k_0} x_i,
	\end{equation*}
	where $x=(x_1, \dots, x_{k_0})$ is the realization of the first $k_0$ steps of the random walk, which has $2^{k_0}$ possibilities in total. Let 
	\begin{equation*}
	\theta_t = \frac{1}{1+t} \Theta \left( \frac{t}{1+t} \right)
	\end{equation*}
	where $\Theta \left( \frac{t}{1+t} \right)$ is the unique positive root of the nonlinear equation 
	\begin{equation*}
	\exp (\theta/(1+t)) + \exp (-\theta t/(1+t)) = 2.
	\end{equation*}
	Then, for any $\widehat{T}_i$,
	\begin{align*}
	& \bbE \left[ \frac{ |\widehat{V}^{+}_{N_i}(\widehat{T}_i)| }{ a + |\widehat{V}^{-}_{N_i}(\widehat{T}_i)| } \right] \leq
	\bbE \Bigg[ \underset{ T_i > 0 } {\sup} \frac{ |\widehat{V}^{+}_{N_i}(T_i)| }{ a + |\widehat{V}^{-}_{N_i}(T_i)|} \Bigg]
	\leq 2^{-k_0} \sum_{x_1, \dots, x_{k_0}}
	\Bigg\{
	\max(R_{k_0}(x,a),1) + \\
	&\int_{0}^{2 - \exp \big(\theta_{\max(R_{k_0}(x,a),1)} \big) } u^{a+k_0-P_{k_0}(x)-1} (2-u)^{P_{k_0}(x)} \frac{1}{\log^2 (2-u)} \left[ (2-u) \log (2-u) + u \log u \right] du
	\Bigg\}.
	\end{align*}
\end{lemma}

\begin{proof}
	Let $m = |NE_i^c|$ be the total number of nulls. Because $W^{(i)}$ possesses the sign-flip property on $NE_i^c$, we have
	\begin{equation*}
	\underset{ T_i > 0 } {\sup} \frac{ |\widehat{V}^{+}_{N_i}(T_i)| }{ a+ |\widehat{V}^{-}_{N_i}(T_i)| }
	\stackrel{d}{=}
	\underset{ 0 \leq k \leq m } {\sup} \frac{ S_k }{ a+ k - S_k}.
	\end{equation*}
	Hence, 
	\begin{equation*}
	\bbE \left[  \frac{ |\widehat{V}^{+}_{N_i}(\widehat{T}_i)| }{ a + |\widehat{V}^{-}_{N_i}(\widehat{T}_i)| } \right]  \leq
	\bbE \Bigg[ \underset{ T_i > 0 } {\sup} \frac{ |\widehat{V}^{+}_{N_i}(T_i)| }{ a+|\widehat{V}^{-}_{N_i}(T_i)|} \Bigg]
	= \bbE \Bigg[ \underset{ 0 \leq k \leq m } {\sup} \frac{ S_k }{ a + k - S_k} \Bigg] 
	\leq \bbE \Bigg[ \underset{ k \geq 0 } {\sup} \frac{ S_k }{ a+k-S_k} \Bigg].
	\end{equation*}
	Using the same arguments as the proof of Lemma 3 in \cite{katsevich2017multilayer}, we obtain that 
	\begin{align*}
	&  \bbE \Bigg[ \underset{ k \geq 0 } {\sup} \frac{ S_k }{ a+k-S_k} \Bigg]
	\leq
	2^{-k_0} \sum_{x_1, \dots, x_{k_0}}
	\Bigg\{
	\max(R_{k_0}(x,a),1) + \\
	&\int_{\theta_{\max(R_{k_0}(x,a),1)}}^{\log 2} \exp \left( (a + k_0 - P_{k_0}(x) - 1) \log (2-e^{\theta_t}) + \theta_t P_{k_0}(x) \right) 
	\frac{ \theta_t e^{\theta_t} + (2-e^{\theta_t}) \log (2-e^{\theta_t})} {\theta_t^2} d\theta_t
	\Bigg\}.
	\end{align*}
	The final result then follows from taking the transformation $u=2-e^{\theta_t}$ in the integration. The reason for taking this transformation is that it can help with calculating the high precision numerical integration when $a$ is small.
\end{proof} 

%%%%%%%%%%%%
By using Lemma \ref{Lemma: Bound of c_a}, one can obtain pairs $(a, c_a)$ that satisfy the following inequality:
\begin{align}\label{SupBoundTerm}
\bbE \left[ \frac{ |\widehat{V}^{+}_{N_i}(\widehat{T}_i)| } {a + |\widehat{V}^{-}_{N_i}(\widehat{T}_i)| } \right] \leq c_a.
\end{align}

\begin{proposition} \label{Proposition: get c_a}
	Let $\widehat{V}^{+}_{N_i}$ and $\widehat{V}^{-}_{N_i}$ be defined as in \eqref{NotationV} with $W^{(i)}$ possessing the sign-flip property on $NE_i^c$. Then for any $\widehat{T}_i$,
	\begin{equation*}
	\bbE \left[ \frac{ |\widehat{V}^{+}_{N_i}(\widehat{T}_i)| }{ a + |\widehat{V}^{-}_{N_i}(\widehat{T}_i)| } \right]  \leq
	\bbE \Bigg[ \underset{ T_i > 0 } {\sup} \frac{ |\widehat{V}^{+}_{N_i}(T_i)| }{ a+|\widehat{V}^{-}_{N_i}(T_i)|} \Bigg] \leq c_a,
	\end{equation*}
	where $c_a=1.93$ for $a=1$ and $c_a=102$ for $a=0.01$.
\end{proposition}

\begin{proof}
	The result for $a=1$ is given by the Lemma 3 in \cite{katsevich2017multilayer}. The result for $a=0.01$ can be derived as follows.
	
	By Lemma \ref{Lemma: Bound of c_a}, for any fixed $k_0 \geq 1$, we have
	\begin{align*}
	& \bbE \left[  \frac{ |\widehat{V}^{+}_{N_i}(\widehat{T}_i)| }{ 0.01 + |\widehat{V}^{-}_{N_i}(\widehat{T}_i)| } \right] 
	\leq \bbE \Bigg[ \underset{ T_i > 0 } {\sup} \frac{ |\widehat{V}^{+}_{N_i}(T_i)| }{ 0.01+|\widehat{V}^{-}_{N_i}(T_i)|} \Bigg]
	\leq 2^{-k_0} \sum_{x_1, \dots, x_{k_0}}
	\Bigg\{
	\max(R_{k_0}(x, 0.01),1) + \\
	&\int_{0}^{2 - \exp \big(\theta_{\max(R_{k_0}(x, 0.01),1)} \big) } u^{0.01 + k_0 - P_{k_0}(x) - 1} (2-u)^{P_{k_0}(x)} \frac{1}{\log (2-u)^2} \left[  (2-u) \log (2-u) + u \log u \right]  du
	\Bigg\} (\ast).
	\end{align*}
	Figure \ref{Plot: bound ca} shows the values of $(\ast)$ for $k_0=1,\dots,10$. The value is already below $102$ when $k_0 \geq 6$, which proves the proposition.
	\begin{figure}[t!]
		\centering
		\includegraphics[width=10cm,height=6cm]{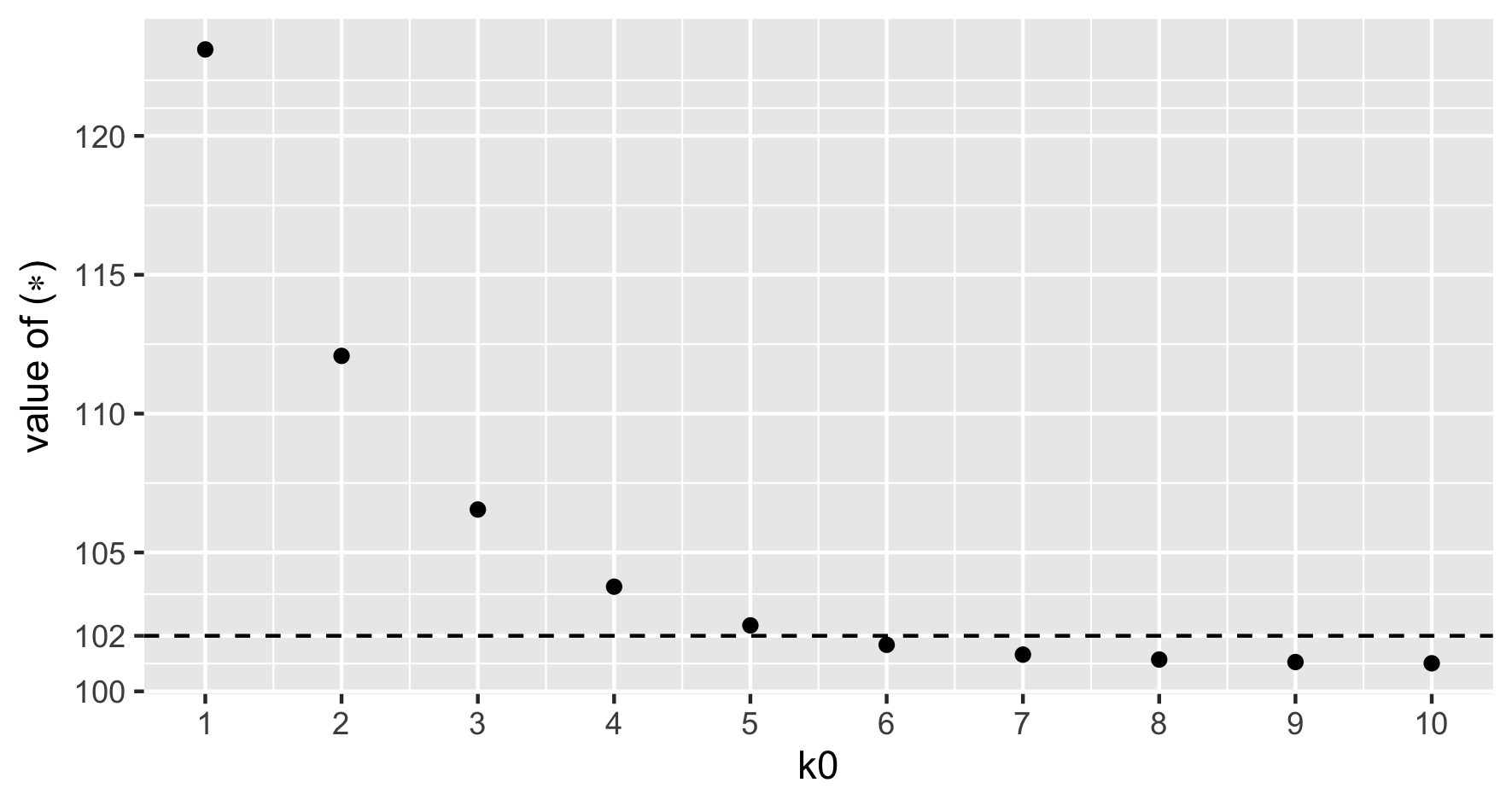}
		\caption{Values of $(\ast)$ for different $k_0$}
		\label{Plot: bound ca}
	\end{figure}
\end{proof}

%%%%%%%%%%%%%%%%%%%%%%%%%%%%%%%%%%%%%%%%%%%%%%%%%%%%%%%%%%%%%%%%%%%%%%%%
%%%%%%%%%%%%%%%%%%%%%%%%%%%%%%%%%%%%%%%%%%%%%%%%%%%%%%%%%%%%%%%%%%%%%%%%
\subsection{Proof of Theorem \ref{Thm: GGMfilter FDR}} \label{appendix: proofGGMfilter FDR}
\FDRcontrol*

\begin{proof}[Proof of Theorem \ref{Thm: GGMfilter FDR}]
	 Let $W=(W^{(1)}, \dots, W^{(p)})$ be the feature statistic matrix calculated using Algorithm \ref{Algo: GGMknock}. Then, by Lemma \ref{SignFlipRandomDesign}, $W^{(i)}$ possesses the sign-flip property on $NE_i^c$ for each $i\in[p]$, which implies that Lemma \ref{Lemma: Bound of c_a} can be used to obtain pair $(a,c_a)$ that satisfies inequality \eqref{SupBoundTerm}. We show that for both cases $\mR=\mR_{\AND}$ and $\mR=\mR_{\OR}$, the FDR of Algorithm \ref{Algo: GGMknock} with $\delta=1$ is controlled, and thus prove the theorem.
	
	%%%%%%%%%%%
    In the case of $\mR=\mR_{\AND}$, we have 
	\begin{equation} \label{ProofFDRandRule}
	\begin{aligned}
	\FDR 
	&= \bbE \Bigg[ \frac{ |\widehat{F}_{\AND}| }{ |\widehat{E}_{\AND}| \vee 1} \Bigg]
	= \bbE \Bigg[ \frac{ |\widehat{F}_{\AND}| }{ |\widehat{E}_{\AND}| \vee 1} \bbone_{ \{ |\widehat{E}_{\AND}| \geq 1 \} } \Bigg] \\
	&=  \bbE \Bigg[ \frac{1}{2} \sum_{i=1}^p \frac{  |\widehat{V}^{+}_{N_i}| }{ |\widehat{E}_{\AND}| \vee 1} \bbone_{ \{ |\widehat{E}_{\AND}| \geq 1 \} } 
	- \frac{1}{2} \cdot \frac{ |\widehat{F}_{\OR} \backslash \widehat{F}_{\AND}| }{ |\widehat{E}_{\AND}| \vee 1} \bbone_{ \{ |\widehat{E}_{\AND}| \geq 1 \} }  \Bigg] \\
	&\leq \frac{1}{2} \bbE \Bigg[ \sum_{i=1}^p \frac{ |\widehat{V}^{+}_{N_i}| }{ |\widehat{E}_{\AND}| \vee 1} \bbone_{ \{ |\widehat{E}_{\AND}| \geq 1 \} } \Bigg] 
	= \frac{1}{2} \sum_{i=1}^p \bbE \Bigg[ \frac{ |\widehat{V}^{+}_{N_i}| }{ a+|\widehat{V}^{-}_{N_i}|} \frac{ a+|\widehat{V}^{-}_{N_i}|}{|\widehat{E}_{\AND}| \vee 1} \bbone_{ \{ |\widehat{E}_{\AND}| \geq 1 \} } \Bigg] \\
	&\leq \frac{1}{2} \sum_{i=1}^p \bbE \Bigg[ \frac{ |\widehat{V}^{+}_{N_i}| }{ a+|\widehat{V}^{-}_{N_i}|} \frac{ a+|\widehat{V}^{-}_{i}|}{|\widehat{E}_{\AND}| \vee 1} \bbone_{ \{ |\widehat{E}_{\AND}| \geq 1 \} } \Bigg] \\
	&\stackrel{\text{see } \eqref{ANDopt}}{\leq} \frac{1}{2} \cdot \frac{2q}{c_ap} \sum_{i=1}^p \bbE \Bigg[ \frac{ |\widehat{V}^{+}_{N_i}(\widehat{T}_i)| }{ a+|\widehat{V}^{-}_{N_i}(\widehat{T}_i)|} \Bigg]  \\
	& \leq \frac{q}{c_ap} \sum_{i=1}^p \bbE \Bigg[ \underset{ T_i > 0 } {\sup} \frac{ |\widehat{V}^{+}_{N_i}(T_i)| }{ a+|\widehat{V}^{-}_{N_i}(T_i)|} \Bigg]
	\stackrel{ \text{see } \eqref{SupBoundTerm}} {\leq} q,
	\end{aligned}
	\end{equation}
	where the indication function $\bbone_{ \{ |\widehat{E}_{\AND}| \geq 1 \} } $ serves as a restriction to the case where feasible thresholds exist, so that \eqref{ANDopt} can be used in the second to last line.
	
	%%%%%%%%%%%
	In the case of $\mR=\mR_{\OR}$, we have
	\begin{equation*}
	\begin{aligned}
	\FDR
	&= \bbE \Bigg[ \frac{ |\widehat{F}_{\OR}| }{ |\widehat{E}_{\OR}| \vee 1} \Bigg]
	= \bbE \Bigg[ \sum_{i=1}^p \frac{ |\widehat{V}^{+}_{N_i}| }{ |\widehat{E}_{\OR}| \vee 1} 
	- \frac{ | \widehat{F}_{\AND} | }{ |\widehat{E}_{\OR}| \vee 1} \Bigg] \\
	&\leq \bbE \Bigg[ \sum_{i=1}^p \frac{ |\widehat{V}^{+}_{N_i}| }{ |\widehat{E}_{\OR}| \vee 1} \Bigg] 
	= \sum_{i=1}^p \bbE \Bigg[ \frac{ |\widehat{V}^{+}_{N_i}| }{ a+|\widehat{V}^{-}_{N_i}|} \cdot \frac{ a+|\widehat{V}^{-}_{N_i}|}{|\widehat{E}_{\OR}| \vee 1} \Bigg] \\
	&\leq \sum_{i=1}^p \bbE \Bigg[ \frac{ |\widehat{V}^{+}_{N_i}| }{ a+|\widehat{V}^{-}_{N_i}|} \cdot \frac{ a+|\widehat{V}^{-}_{i}|}{|\widehat{E}_{\OR}| \vee 1} \bbone_{ \{ |\widehat{E}_{\OR}| \geq 1 \} } \Bigg]
	\stackrel{\text{see } \eqref{ORopt}}{\leq} \frac{q}{c_ap} \sum_{i=1}^p \bbE \Bigg[ \frac{ |\widehat{V}^{+}_{N_i}(\widehat{T}_i)| }{ a+|\widehat{V}^{-}_{N_i}(\widehat{T}_i)|} \Bigg] \\
	& \leq \frac{q}{c_ap} \sum_{i=1}^p \bbE \Bigg[ \underset{ T_i > 0 } {\sup} \frac{ |\widehat{V}^{+}_{N_i}(T_i)| }{ a+|\widehat{V}^{-}_{N_i}(T_i)|} \Bigg]
	\stackrel{ \text{see } \eqref{SupBoundTerm}} {\leq} q.
	\end{aligned}
	\end{equation*}
\end{proof}

\begin{remark}\label{ReasonForNotUsingSummationInOpt}
	From the fourth to the fifth line in inequalities \eqref{ProofFDRandRule}, one can see that there are two terms indexed by $i$ inside of the summation and the expectation is taken with respect to them.
	This leads us to use separate constraints for each $i$ in \eqref{ANDopt} rather than one summed constraint (i.e., $\sum_{i=1}^{p} \frac{ a +|\widehat{V}^{-}_{i}(T_i)|}{|\widehat{E}_{\AND}(T)| \vee 1} \leq 2q/c_a$ ).
\end{remark}

%%%%%%%%%%%%%%%%%%%%%%%%%%%%%%%%%%%%%%%%%%%%%%%%%%%%%%%%%
\subsection{Proof of Theorem \ref{Thm: GGMfilter mFDR}}
\mFDRcontrol*

\begin{proof}[Proof of Theorem \ref{Thm: GGMfilter mFDR}]
	
	The proof of mFDR control is similar to the proof of FDR control.
	
	 Let $W=(W^{(1)}, \dots, W^{(p)})$ be the feature statistic matrix calculated using Algorithm \ref{Algo: GGMknock}. Then, by Lemma \ref{SignFlipRandomDesign}, $W^{(i)}$ possesses the sign-flip property on $NE_i^c$ for each $i\in[p]$, which implies that Lemma \ref{Lemma: Bound of c_a} can be used to obtain pair $(a,c_a)$ that satisfies inequality \eqref{SupBoundTerm}. We show that for  and for both cases $\mR=\mR_{\AND}$ and $\mR=\mR_{\OR}$, the mFDR of Algorithm \ref{Algo: GGMknock} with $\delta=0$ is controlled, and thus prove the theorem.
	
	%%%%%%%%%%%
	In the case of $\mR=\mR_{\AND}$, we have
	\begin{equation*}
	\begin{aligned}
	\mFDR_{\AND} 
	&= \bbE \Bigg[ \frac{ |\widehat{F}_{\AND}| }{ |\widehat{E}_{\AND}| + a c_a p/(2q) } \Bigg]
	= \bbE \Bigg[ \frac{ |\widehat{F}_{\AND}| }{ |\widehat{E}_{\AND}| + a c_a p/(2q)} \bbone_{ \{ |\widehat{E}_{\AND}| \geq 1 \} } \Bigg] \\
	&\leq \frac{1}{2} \bbE \Bigg[ \sum_{i=1}^p \frac{ |\widehat{V}^{+}_{N_i}| }{ |\widehat{E}_{\AND}| + a c_a p/(2q)} \bbone_{ \{ |\widehat{E}_{\AND}| \geq 1 \} } \Bigg]  \\
	&\leq \frac{1}{2} \sum_{i=1}^p \bbE \Bigg[ \frac{ |\widehat{V}^{+}_{N_i}| }{ a+|\widehat{V}^{-}_{N_i}|} \cdot \frac{ a+|\widehat{V}^{-}_{i}|}{|\widehat{E}_{\AND}| + a c_a p/(2q) } \bbone_{ \{ |\widehat{E}_{\AND}| \geq 1 \} } \Bigg] \\
	& \stackrel{\text{see } (\ref{mFDRoptAND})}{\leq} \frac{1}{2} \sum_{i=1}^p \bbE \Bigg[ \frac{ |\widehat{V}^{+}_{N_i}| }{ a+|\widehat{V}^{-}_{N_i}|} \cdot \frac{ a+ 2q/(c_a p) \cdot |\widehat{E}_{\AND}|}{|\widehat{E}_{\AND}|  + a c_a p/(2q) } \bbone_{ \{ |\widehat{E}_{\AND}| \geq 1 \} } \Bigg] \\
	& = \frac{q}{c_a p}  \sum_{i=1}^p \bbE \Bigg[ \frac{ |\widehat{V}^{+}_{N_i}| }{ a +|\widehat{V}^{-}_{N_i}|} \cdot \frac{ a c_a p/(2q) + |\widehat{E}_{\AND}|}{|\widehat{E}_{\AND}|  + a c_a p/(2q) } \bbone_{ \{ |\widehat{E}_{\AND}| \geq 1 \} } \Bigg] \\
	&\leq \frac{q}{c_a p}  \sum_{i=1}^p  \bbE \Bigg[ \frac{ |\widehat{V}^{+}_{N_i}(\widehat{T}_i)| }{ a+|\widehat{V}^{-}_{N_i}(\widehat{T}_i)|} \Bigg]
	\leq \frac{q}{c_a p}  \sum_{i=1}^p \bbE \Bigg[ \underset{ T_i > 0 } {\sup} \frac{ |\widehat{V}^{+}_{N_i}(T_i)| }{ a+|\widehat{V}^{-}_{N_i}(T_i)|} \Bigg]
	\stackrel{ \text{see } \eqref{SupBoundTerm}} {\leq} q.
	\end{aligned}
	\end{equation*}
	
	%%%%%%%%%%%
	In the case of $\mR=\mR_{\OR}$, we have
	\begin{equation*}
	\begin{aligned}
	\mFDR_{\OR} 
	&= \bbE \Bigg[ \frac{ |\widehat{F}_{\OR}| }{ |\widehat{E}_{\OR}| + a c_a p/q } \Bigg]
	\leq\sum_{i=1}^p \bbE \Bigg[ \frac{ |\widehat{V}^{+}_{N_i}| }{ |\widehat{E}_{\OR}| + a c_a p/q} \Bigg] \\
	&\leq \sum_{i=1}^p \bbE \Bigg[ \frac{ |\widehat{V}^{+}_{N_i}| }{ a +|\widehat{V}^{-}_{N_i}|} \cdot \frac{ a +|\widehat{V}^{-}_{i}|}{|\widehat{E}_{\OR}| + a c_a p/q } \bbone_{ \{ |\widehat{E}_{\OR}| \geq 1 \} } \Bigg] \\
	& \stackrel{\text{see } (\ref{mFDRoptOR})}{\leq} \sum_{i=1}^p \bbE \Bigg[ \frac{ |\widehat{V}^{+}_{N_i}| }{ a +|\widehat{V}^{-}_{N_i}|} \cdot \frac{ a + q/(c_a p) \cdot |\widehat{E}_{\OR}|}{|\widehat{E}_{\OR}|  + a c_a p/q } \bbone_{ \{ |\widehat{E}_{\OR}| \geq 1 \} } \Bigg] \\
	& = \frac{q}{c_a p} \sum_{i=1}^p \bbE \Bigg[ \frac{ |\widehat{V}^{+}_{N_i}| }{ a+|\widehat{V}^{-}_{N_i}|} \cdot \frac{ a c_a p/q + |\widehat{E}_{\OR}|}{|\widehat{E}_{\OR}|  + a c_a p/q } \bbone_{ \{ |\widehat{E}_{\OR}| \geq 1 \} } \Bigg] \\
	&= \frac{q}{c_ap} \sum_{i=1}^p \bbE \Bigg[ \frac{ |\widehat{V}^{+}_{N_i}(\widehat{T}_i)| }{ a+|\widehat{V}^{-}_{N_i}(\widehat{T}_i)|} \Bigg]
	\leq \frac{q}{c_ap}  \sum_{i=1}^p \bbE \Bigg[ \underset{ T_i > 0 } {\sup} \frac{ |\widehat{V}^{+}_{N_i}(T_i)| }{ a+|\widehat{V}^{-}_{N_i}(T_i)|} \Bigg]
	\stackrel{ \text{see } \eqref{SupBoundTerm}} {\leq} q.
	\end{aligned}
	\end{equation*}
\end{proof}

%%%%%%%%%%%%%%%%%%%%%%%%%%%%%%%%%%%%%%%%%%%%%%%%%%%%%%%%%
\subsection{Proof of Theorem \ref{Thm: FDRmFDRSplittingRecycling}}

\SampleSplittingRecyclinfFDRmFDRresuklts*

Compared to Algorithm \ref{Algo: GGMknock}, we first need to show that each column of the feature statistic matrix $W$ possesses the sign-flip property on $NE_i^c$ with the sample-splitting-recycling step, so that Lemma \ref{Lemma: Bound of c_a} can be applied and inequality \eqref{SupBoundTerm} holds. The remaining proofs are then similar.

\begin{lemma} (\textbf{Sign-flip property of each column of $W$ with sample-splitting-recycling} ) \label{Lemma: SplitRecyclingColumnSignFlip}
	~\\
	Let $\bX \in \bbR^{n \times p}$ be the original sample matrix, $\bX_{1} \in \bbR^{n_1 \times p}$ and $\bX_{2} \in \bbR^{n_2 \times p}$ be two subsamples obtained by randomly splitting $\bX$ with $n_1+n_2=n$, and $ \bX^{\re} = \begin{pmatrix} \bX_{1 } \\ \bX_{2 } \end{pmatrix} \in \bbR^{n \times p}$ be the collection of these two subsamples. Let $W=(W^{(1)}, \dots, W^{(p)})$ be the feature statistic matrix calculated using Algorithm \ref{Algo: GGMknockWithSplittingRecycling}. Then for each $i\in[p]$, conditional on $\bX_{1}$, $W^{(i)} $ possesses the sign-flip property on $NE_i^c$.
\end{lemma}

\begin{proof}[Proof of Lemma \ref{Lemma: SplitRecyclingColumnSignFlip}]
	
	For any $i\in [p]$, we first show that conditional on $\bX_{1}$ and $\bX^{\re(-i)} $, $W^{(i)}$ possesses the sign-flip property on $NE_i^c$ using the same idea as the proof of Lemma 1 in \cite{barber2015controlling}. Note that in this case, the only random term is $\bX^{(i)}_{2}$. Let $H \subseteq NE_i^c$ be any subset. Because $\widetilde{\bX}^{\re(-i)} $ is a knockoff matrix of $\bX^{\re(-i)}$, we have 
	\begin{equation*}
	[ \bX^{\re(-i)} \ \widetilde{\bX}^{\re(-i)} ]_{\swap(H)}^T  [ \bX^{\re(-i)} \ \widetilde{\bX}^{\re(-i)} ]_{\swap(H)}
	= [ \bX^{\re(-i)} \ \widetilde{\bX}^{\re(-i)} ]^T [ \bX^{\re(-i)} \ \widetilde{\bX}^{\re(-i)} ].
	\end{equation*}
	Now we show that
	\begin{equation*}
	[ \bX^{\re(-i)} \ \widetilde{\bX}^{\re(-i)} ]_{\swap(H)}^T  \bX^{\re(i)} \stackrel{d}{=} [ \bX^{\re(-i)} \ \widetilde{\bX}^{\re(-i)} ]^T  \bX^{\re(i)}.
	\end{equation*}
	This is true because
	\begin{equation*}
	\begin{aligned}
	& \quad \ [ \bX^{\re(-i)} \ \widetilde{\bX}^{\re(-i)} ]_{\swap(H)}^T  \bX^{\re(i)}  \\ 
	&= \begin{pmatrix} [\bX^{(-i)}_{1} \ \bX^{(-i)}_{1}]_{\swap(H)} \\ [\bX^{(-i)}_{2} \ \widetilde{\bX}^{(-i)}_{2}]_{\swap(H)} \end{pmatrix}^T 
	\begin{pmatrix} \bX^{(i)}_{1} \\ \bX^{(i)}_{2} \end{pmatrix} \\
	&=[\bX^{(-i)}_{1} \ \bX^{(-i)}_{1}]_{\swap(H)}^T \bX^{(i)}_{1} + [\bX^{(-i)}_{2} \ \widetilde{\bX}^{(-i)}_{2}]_{\swap(H)}^T \bX^{(i)}_{2} \\
	&\stackrel{d}{=} [\bX^{(-i)}_{1} \ \bX^{(-i)}_{1}]^T \bX^{(i)}_{1} +  [\bX^{(-i)}_{2} \ \widetilde{\bX}^{(-i)}_{2}]^T \bX^{(i)}_{2} \\
	&= [ \bX^{\re(-i)} \ \widetilde{\bX}^{\re(-i)} ]^T  \bX^{\re(i)}.
	\end{aligned}
	\end{equation*}
	The third equation holds because $ [\bX^{(-i)}_{1} \ \bX^{(-i)}_{1}]_{sw\textbf{}ap(H)}^T \bX^{(i)}_{1} = [\bX^{(-i)}_{1} \ \bX^{(-i)}_{1}]^T \bX^{(i)}_{1} $ and \\ $ [\bX^{(-i)}_{2} \ \widetilde{\bX}^{(-i)}_{2}]_{\swap(H)}^T \bX^{(i)}_{2} \stackrel{d}{=} [\bX^{(-i)}_{2} \ \widetilde{\bX}^{(-i)}_{2}]^T \bX^{(i)}_{2}$ by Lemma \ref{SecondLemmaFromBC}. Combining the above results we have
	% the third equation tells us why we need a particular way to construct XnkRe. o.w. (construct knockoffs for \bX^{\re} using (\ref{ConstructKnockoffs}) ) this step doesnot hold.
	\begin{equation}\label{EqualInDistri}
	\begin{aligned}
	&\quad \left( [ \bX^{\re(-i)} \ \widetilde{\bX}^{\re(-i)} ]_{\swap(H)}^T  [ \bX^{\re(-i)} \ \widetilde{\bX}^{\re(-i)} ]_{\swap(H)}, 
		[ \bX^{\re(-i)} \ \widetilde{\bX}^{\re(-i)} ]_{\swap(H)}^T  \bX^{\re(i)} \right) \\
	&\stackrel{d}{=} \left( [ \bX^{\re(-i)} \ \widetilde{\bX}^{\re(-i)} ]^T [ \bX^{\re(-i)} \ \widetilde{\bX}^{\re(-i)} ], 
		[ \bX^{\re(-i)} \ \widetilde{\bX}^{\re(-i)} ]^T  \bX^{\re(i)} \right) 
	\end{aligned}
	\end{equation}
	
	Let $\epsilon=(\epsilon_1, \dots, \epsilon_p)$ be a sign sequence independent of $W^{(i)}$ with $\epsilon_j = +1$ for $j \not\in NE_i^c$ and $\epsilon_j$ i.i.d.\ from a Rademacher distribution for $j \in NE_i^c$. Let $K = \{ j \in [p]: \epsilon_j=-1 \}$ be a set depending on $\epsilon$. We have that $K \subseteq NE_i^c$. Then, conditional on $\bX_{1}$ (so $\mathcal{P}(\bX_{1}) $ is fixed) and $\bX^{\re(-i)}$, we have
	\begin{equation*}
	\begin{aligned}
	&\quad (W^{(i)}_1 \cdot \epsilon_1, \cdots, W^{(i)}_p\cdot \epsilon_p) \\
	&= {W^{(i)}}_{\swap(K)} \\
	&= f \left(  [ \bX^{\re(-i)} \ \widetilde{\bX}^{\re(-i)} ]_{\swap(K)}^T  [ \bX^{\re(-i)} \ \widetilde{\bX}^{\re(-i)} ]_{\swap(K)}, 
		[ \bX^{\re(-i)} \ \widetilde{\bX}^{\re(-i)} ]_{\swap(K)}^T  \bX^{\re(i)},  \mathcal{P}(\bX_{1}) \right)  \\
	&\stackrel{d}{=} f \left( [ \bX^{\re(-i)} \ \widetilde{\bX}^{\re(-i)} ]^T [ \bX^{\re(-i)} \ \widetilde{\bX}^{\re(-i)} ], 
	   [ \bX^{\re(-i)} \ \widetilde{\bX}^{\re(-i)} ]^T  \bX^{\re(i)},  \mathcal{P}(\bX_{1}) \right)  \\
	&= W^{(i)},
	 \end{aligned}
	\end{equation*}
	where the first equation follows from the antisymmetry property of $W^{(i)}$ and the definition of $K$, and the third equation follows from \eqref{EqualInDistri}.
	
	Therefore, $W^{(i)}$ possesses the sign-flip property on $NE_i^c$ conditional on $\bX_{1}$ and $\bX^{\re(-i)} $. Then, similar to the proof of Lemma \ref{SignFlipRandomDesign}, we can show that $ {W^{(i)}}_{\swap(K)} \stackrel{d}{=} W^{(i)}$ when $\bX^{\re(-i)} $ is treated as random, and finally conclude that $W^{(i)}$ possesses the sign-flip property on $NE_i^c$ conditional on $\bX_{1}$.
\end{proof}

%%%%%%%%%%%%%%%%%%%%%%%%%%%%%%%%%%%%%%%%%%%%%%%%%%%%%%%%%
We now prove Theorem \ref{Thm: FDRmFDRSplittingRecycling}.
\begin{proof}[Proof of Theorem \ref{Thm: FDRmFDRSplittingRecycling}]
	
	By Lemma \ref{Lemma: SplitRecyclingColumnSignFlip} and Lemma \ref{Lemma: Bound of c_a}, for a fixed $a>0$, we can obtain $c_a$ such that
	\begin{equation}\label{BoundForProofRecycle}
	\bbE \left[ \underset{ T_i > 0 } {\sup} \left. \frac{ |\widehat{V}^{+}_{N_i}(T_i)| }{ a+|\widehat{V}^{-}_{N_i}(T_i)|} \rgiven \bX_{1} \right]  \leq c_a.
	\end{equation}
	where $\widehat{V}^{+}_{N_i}$ and $\widehat{V}^{-}_{N_i}$ are based on the feature statistic matrix $W$ obtained from Algorithm \ref{Algo: GGMknockWithSplittingRecycling}. In particular, pairs $(a=1,c_a=1.93)$ and $(a=0.01,c_a=102)$ satisfy inequality \eqref{BoundForProofRecycle}.
	
	Conditional on $\bX_{1}$, $(a,c_a)$ is fixed as it is chosen based on $\bX_{1}$, hence inequality \eqref{BoundForProofRecycle} can be used. By using the same arguments as the proof of Theorem \ref{Thm: GGMfilter FDR} with the only modification that invoking inequality \eqref{BoundForProofRecycle}) instead of \eqref{SupBoundTerm} in the last step, we have
	\begin{equation} \label{FDRcontrolCondiX1}
	\bbE \left[ \left. \frac{ |\widehat{F}^{\re}_{\AND}| }{ |\widehat{E}^{\re}_{\AND}| \vee 1} \rgiven \bX_{1} \right] \leq q
	\quad \text{and} \quad
	\bbE \left[ \left. \frac{ |\widehat{F}^{\re}_{\OR}| }{ |\widehat{E}^{\re}_{\OR}| \vee 1} \rgiven \bX_{1} \right] \leq q,
	\end{equation}
	where $\widehat{E}^{\re}_{\AND}$ is the estimated edge set using $\mR_{\AND}$ and $W$ obtained from Algorithm \ref{Algo: GGMknockWithSplittingRecycling}, $\widehat{F}^{\re}_{\AND}$ is the set of falsely discovered edges in $\widehat{E}^{\re}_{\AND}$, and similarly for $\widehat{E}^{\re}_{\OR}$ and $\widehat{F}^{\re}_{\OR}$.
	
	Therefore, for Algorithm \ref{Algo: GGMknockWithSplittingRecycling} with $\delta=1$, we have
	\begin{equation*}
	\begin{aligned}
	\FDR 
	&= \bbE \left[ \frac{ |\widehat{F}^{\re}| }{ |\widehat{E}^{\re}| \vee 1} \right]  \\
	&= \bbE  \left[ \bbE \left[  \left. \frac{ |\widehat{F}^{\re}| }{ |\widehat{E}^{\re}| \vee 1} \rgiven \bX_{1} \right] \right] \\
	&= \bbE \Bigg[ \bbE \left[ \left. \frac{ |\widehat{F}^{\re}_{\AND}| }{ |\widehat{E}^{\re}_{\AND}| \vee 1} \bbone_{ \{ \mR_{\AND}\text{ is chosen based on } \bX_{1} \} }  \rgiven \bX_{1} \right] \\
	& \quad	+ \bbE \left[ \left. \frac{ |\widehat{F}^{\re}_{\OR}| }{ |\widehat{E}^{\re}_{\OR}| \vee 1} \bbone_{ \{ \mR_{\OR} \text{ is chosen based on } \bX_{1} \} }  \rgiven \bX_{1} \right]  \Bigg] \\
	&= \bbE \Bigg[ \bbone_{ \{ \mR_{\AND}\text{ is chosen based on } \bX_{1} \} } \bbE \left[ \left. \frac{ |\widehat{F}^{\re}_{\AND}| }{ |\widehat{E}^{\re}_{\AND}| \vee 1} \rgiven \bX_{1} \right] \\
	& \quad	+ \bbone_{ \{ \mR_{\OR} \text{ is chosen based on } \bX_{1} \} } \bbE \left[  \left. \frac{ |\widehat{F}^{\re}_{\OR}| }{ |\widehat{E}^{\re}_{\OR}| \vee 1} \rgiven \bX_{1} \right]  \Bigg] \\
	& \stackrel{\text{see } (\ref{FDRcontrolCondiX1})}{\leq} q \cdot \bbE [ \bbone_{ \{ \mR_{\AND}\text{ is chosen based on } \bX_{1} \} } + \bbone_{ \{ \mR_{\OR} \text{ is chosen based on } \bX_{1} \} } ] \\
	&= q.
	\end{aligned}
	\end{equation*}
	
	The proof about the $\mFDR^{\re}$ control (corresponding to Algorithm \ref{Algo: GGMknockWithSplittingRecycling} with $\delta=0$) is similar:

	Conditional on $\bX_{1}$, $(a,c_a)$ is fixed. By using the same arguments as the proof of Theorem \ref{Thm: GGMfilter mFDR} with the only modification that invoking inequality \eqref{BoundForProofRecycle}) instead of \eqref{SupBoundTerm} in the last step, we have
	\begin{equation} \label{mFDRcontrolCondiX1}
	\bbE \left[ \left. \frac{ |\widehat{F}^{\re}_{\AND}| }{ |\widehat{E}^{\re}_{\AND}| + a c_a p/(2q)} \rgiven \bX_{1} \right] \leq q
	\quad \text{and} \quad
	\bbE \left[ \left. \frac{ |\widehat{F}^{\re}_{\OR}| }{ |\widehat{E}^{\re}_{\OR}| + a c_a p/q} \rgiven \bX_{1} \right] \leq q.
	\end{equation}
	Therefore, we have
	\begin{equation*}
	\begin{aligned}
	&\quad \mFDR^{\re} 
	= \bbE  \left[ \frac{ |\widehat{F}^{\re}| }{ |\widehat{E}^{\re}| + a c_a p/(q\bbone_{ \{ \mR_{\AND}\text{ is chosen based on } \bX_{1} \} } + q)} \right] \\
	&= \bbE  \Bigg[ \frac{ |\widehat{F}^{\re}_{\AND}| }{ |\widehat{E}^{\re}_{\AND}| + a c_a p/(2q)} \bbone_{ \{ \mR_{\AND}\text{ is chosen based on } \bX_{1} \} } \\
	& \quad	+  \frac{ |\widehat{F}^{\re}_{\OR}| }{ |\widehat{E}^{\re}_{\OR}| + a c_a p/q} \bbone_{ \{ \mR_{\OR} \text{ is chosen based on } \bX_{1} \} } \Bigg]  \\
	&= \bbE \Bigg[ \bbE \Bigg[ \left. \frac{ |\widehat{F}^{\re}_{\AND}| }{ |\widehat{E}^{\re}_{\AND}| + a c_a p/(2q)} \bbone_{ \{ \mR_{\AND}\text{ is chosen based on } \bX_{1} \} }  \rgiven \bX_{1} \Bigg] \\
	&\quad	+ \bbE \Bigg[ \left. \frac{ |\widehat{F}^{\re}_{\OR}| }{ |\widehat{E}^{\re}_{\OR}| + a c_a p/q} \bbone_{ \{ \mR_{\OR} \text{ is chosen based on } \bX_{1} \} }  \rgiven \bX_{1} \Bigg] \Bigg ] \\
	&= \bbE \Bigg [ \bbone_{ \{ \mR_{\AND}\text{ is chosen based on } \bX_{1} \} } \bbE \Bigg[ \left. \frac{ |\widehat{F}^{\re}_{\AND}| }{ |\widehat{E}^{\re}_{\AND}| + a c_a p/(2q)} \rgiven \bX_{1} \Bigg]  \\
	&\quad	+ \bbone_{ \{ \mR_{\OR} \text{ is chosen based on } \bX_{1} \} } \bbE \Bigg[ \left. \frac{ |\widehat{F}^{\re}_{\OR}| }{ |\widehat{E}^{\re}_{\OR}| + a c_a p/q} \rgiven \bX_{1}  \Bigg]  \Bigg] \\
	& \stackrel{\text{see } (\ref{mFDRcontrolCondiX1})}{\leq} q \cdot \bbE [ \bbone_{ \{ \mR_{\AND}\text{ is chosen based on } \bX_{1} \} } + \bbone_{ \{ \mR_{\OR} \text{ is chosen based on } \bX_{1} \} } ]\\
	&= q.
	\end{aligned}
	\end{equation*}	
	\end{proof}

\section{Details of the simulations in Section \ref{sec: sample-splitting-recycling}} \label{appendix: SimulationSettingSection4}

Here, we give the details of the two settings related to Figure \ref{OppoPower} in Section \ref{sec: sample-splitting-recycling}. 
Both settings use the band graph. Specifically, we first generate $\Omega^{o}$ via letting $\Omega^{o}_{i,i}=1$ for $i=1,\dots, p$ and $\Omega^{o}_{i,j} = sign(b)\cdot |b|^{|i-j|/10} \cdot \bbone_{|i-j|\leq 10}$ for all $i \neq j$. We use $b=-0.1$ in setting 1 and $b=0.1$ in setting 2. Then, we randomly permute $\Omega^{o}$ by rows and columns to break the pattern of the matrix. The finally used precision matrix $\Omega$ is obtained by $\Omega = \Omega^{o} +(|\lambda_{\min}(\Omega^{o} )| + 0.5)I $, which ensures that $\Omega$ is positive definite. Here, $\lambda_{\min} (\Omega^{o} )$ denotes the minimal eigenvalue of $\Omega^{o} $. The samples are generated independently from $\mN_p(0,\Omega^{-1})$. For both setting, we use sample size $n=3000$, dimension $p=200$ and the nominal FDR level $q=0.2$. 
The results are based on $100$ replications. 
For each replication, we store the resulting TPP of the estimated graph $\widehat E$. We then take average of the $100$ TPP realizations to obtain the empirical power.

The gray and red points in Figure \ref{OppoPower} represent the empirical power of $880$ different FDR control procedures in these two settings. They are obtained by running Algorithm \ref{Algo: GGMknock} with $\delta=1$  and $880$ different combinations of $ ((a,c_a), \mO, \mP, \mR) $ as described in Section \ref{sec:simu5.1}. The blue point in Figure \ref{OppoPower} represents the empirical power of the GGM knockoff filter with sample-splitting-recycling in these two settings. It is obtained by running Algorithm \ref{Algo: GGMknockWithSplittingRecycling} with $\delta=1$ and $(\bA, \bmO, \bmP, \bmR)$ as described in Section \ref{sec:simu5.1}.

%%%%%%%%%%%%%%%%%%%%%%%%%%%%%%%%%%%%%%%%%%%%%%%%%%%%%%%%%%%%%%%%%%%%%%%%%%%%%%%%%%%%%%%%%%%%%%
\section{Additional simulation results}
\subsection{Power gain of the sample-splitting-recycling approach} \label{appendix: RecyclingPowerGain}
%Describe, discuss and comment simulation results.

In this section, we show by simulations that the sample-splitting-recycling approach indeed gains more power compared to the vanilla sample-spitting. We consider the same four graph types as in Section \ref{sec:simu} with the number of variable $p=200$, the edge parameter $b=-0.6$ for band and block graphs, the nominal FDR level $q=0.2$, and the sample size $n$ varied between $1500$ and $4000$. 
We use $100$ replications for each setting.
For each replication, we store the resulting FDP and TPP of the estimated graph $\widehat E$. 
We then take average of the $100$ FDP and TPP realizations to obtain the empirical FDR and power.
The GGM knockoff filter with sample-splitting-recycling and with sample-splitting are denoted by GKF-Re+ and GKF-NoRe+, respectively. 
The simulation results are shown in Figure \ref{RecyclepowerGain}. As expected, GKF-Re+ is always more powerful than GKF-NoRe+.
\begin{figure}[t!]
	\centering
	\includegraphics[width=13cm,height=13cm]{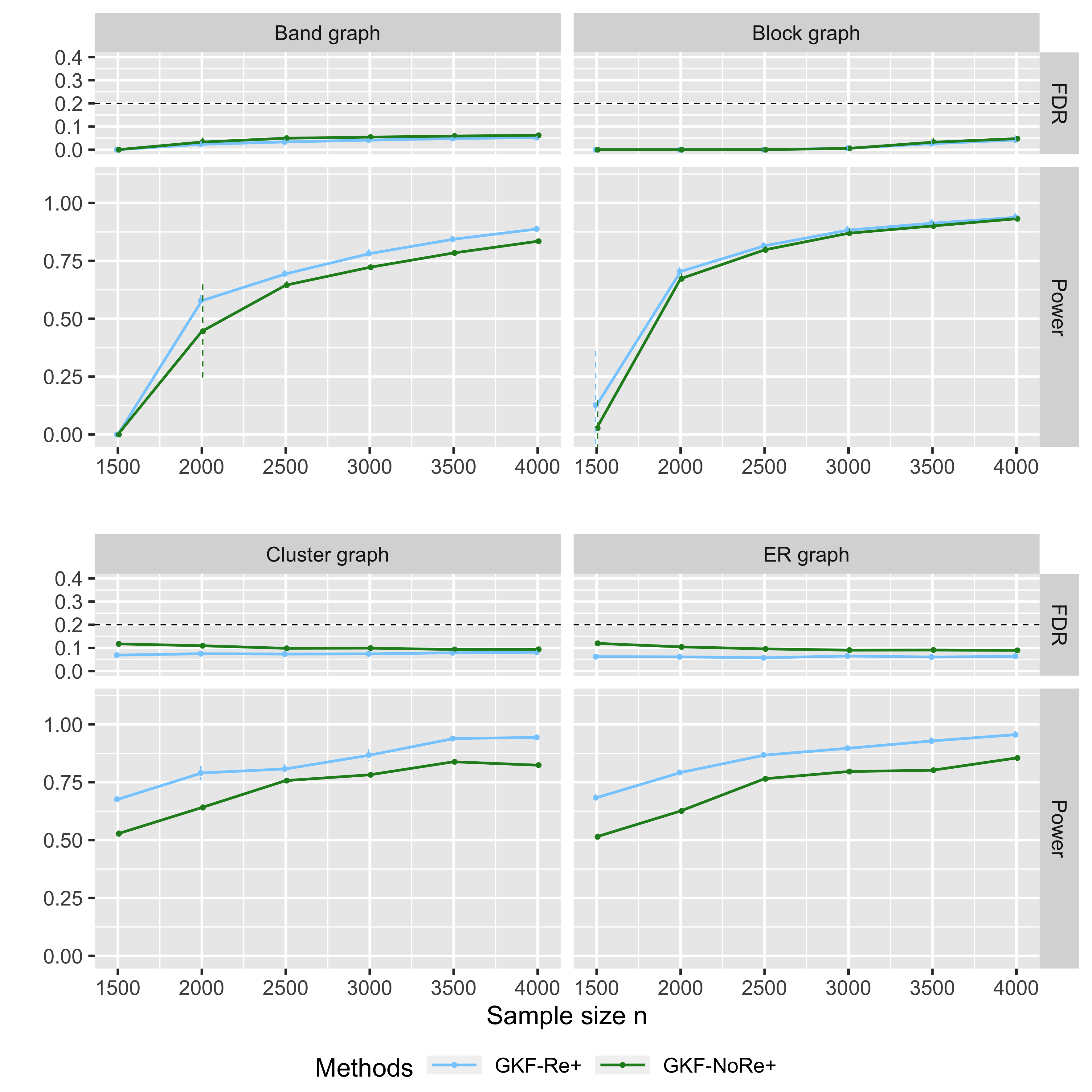}
	\caption{Simulation results: the empirical FDR and power of GKF-Re+ and GKF-NoRe+ on the four graph types when varying the sample size $n$, while the nominal FDR level $q=0.2$ and the edge parameter $b=-0.6$ for band and block graphs. The dashed vertical bars indicate plus/minus one empirical standard deviation of the FDP and TPP.}
	\label{RecyclepowerGain}
\end{figure}

%%%%%%%%%%%%%%%%%%%%%%%%%%%%%%%%%%%%%%%%%%%%%%%%%%%%%%%%%%%%%%%%%%%%%%%%%%%%
\subsection{Oracle hyperparameters v.s. sample-splitting-recycling} \label{appendix: Oracle}

In this section, we compares the performance of GKF+ (Algorithm \ref{Algo: GGMknock} with $\delta=1$) with oracle hyperparameters (denoted by GKF-Oracle+) to that of GKF-Re+. 
We consider the same four graph types as in Section \ref{sec:simu} with the number of variable $p=200$, the edge parameter $b=-0.6$ for band and block graphs, the nominal FDR level $q=0.2$, and the sample size $n$ varied between $1500$ and $4000$. 
We use $100$ replications for each setting.
For each replication, we store the resulting FDP and TPP of the estimated graph $\widehat E$. 
We then take average of the $100$ FDP and TPP realizations to obtain the empirical FDR and power.
GKF-Oracle+ is GKF+ with the hyperparameters (among the $880$ choices, see Section \ref{sec:simu5.1} for details) that lead to the largest empirical powers calculated using all samples.  
The simulation results are shown in Figure \ref{OracleFig:VaryN}. 

As expected, GKF-Oracle+ is always better than GKF-Re+, because it uses the best hyperparameters for the given setting and it uses the full sample. In particular, when the sample size $n$ is not large enough, the empirical power of GKF-Re+ dramatically decreases while GKF-Oracle+ is still good. 
%When the sample size $n$ is large enough, the powers of these two approaches are comparable. This is because that in the large $n$ case, half of the sample is enough for GKF-Re+ to choose good hyperparameters, and using full sample will not improve too much power compared to using half of the sample.
Knowing the oracle hyperparameters is generally impossible in practice, but from Figure \ref{OracleFig:VaryN}, we can see that the empirical power of GKF+ can be largely improved, especially in small $n$ csae, if a better solution than sample-splitting-recycling can be proposed to choose the hyperparameters.
% This is due to the fact that only half of the sample is used in choosing the hyperparameters, and although a sample recycling approach is used when we apply the chosen procedure, there is still a power loss compared to using all sample. 
\begin{figure}[t!]
	\centering
	\includegraphics[width=13cm, height=13cm]{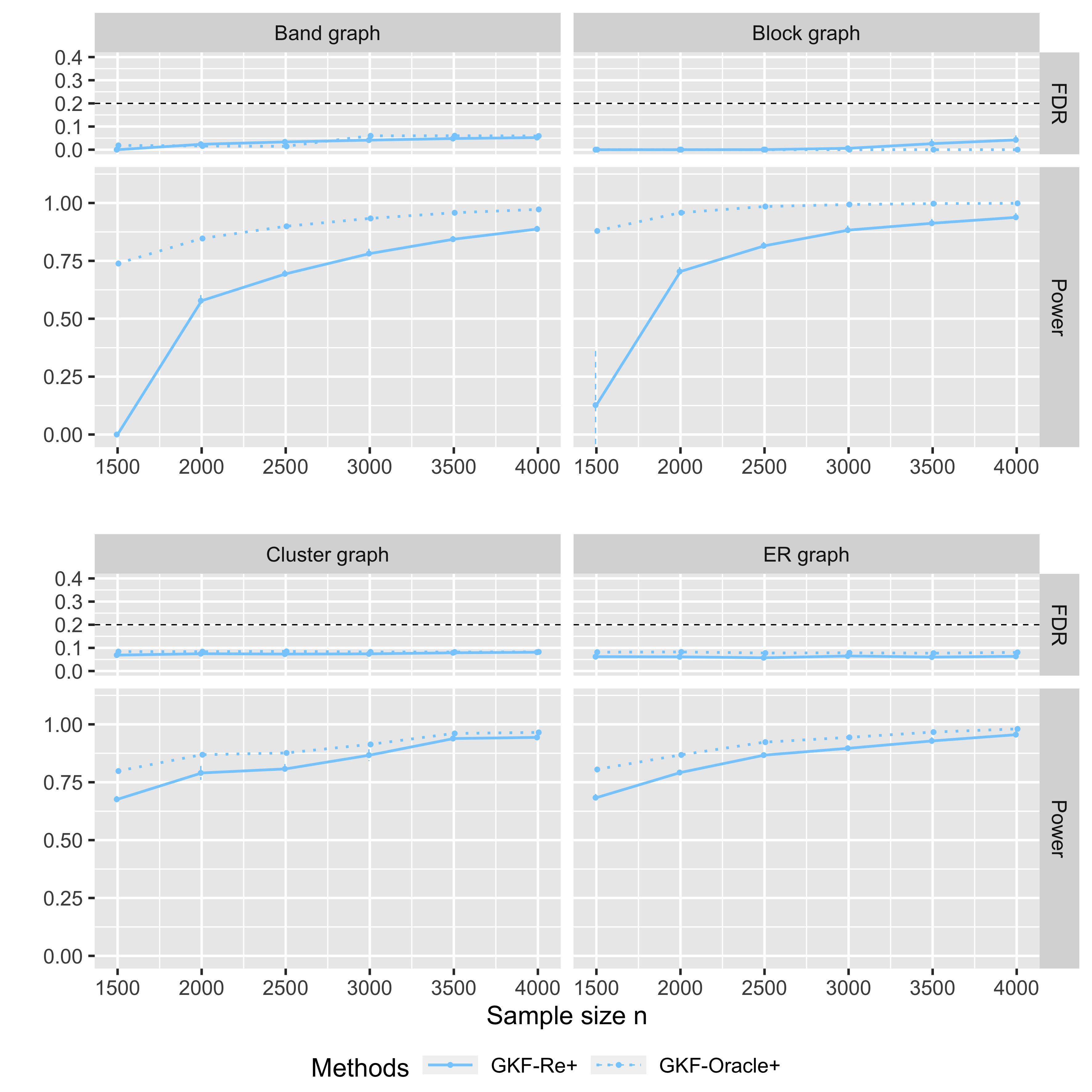}
	\caption{Simulation results: the empirical FDR and power of GKF-Re+ and GKF-Oracle+ on the four graph types when varying the sample size $n$, while the nominal FDR level $q=0.2$ and the edge parameter $b=-0.6$ for band and block graphs. The dashed vertical bars indicate plus/minus one empirical standard deviation of the FDP and TPP.}
	\label{OracleFig:VaryN}
\end{figure}

%%%%%%%%%%%%%%%%%%%%%%%%%%%%%%%%%%%%%%%%%%%%%%%%%%%%%%%%%%%%%%%%%%%%%%%%%%%%
\subsection{Performance of the mFDR control methods} \label{appendix: mFDR}

In this section, we investigate the performance of the mFDR controlled counterpart of GKF-Re+ (denoted by GKF-Re), which is the Algorithm \ref{Algo: GGMknockWithSplittingRecycling} with $\delta=0$ and the hyperparameter space as described in Section \ref{sec:simu5.1}. We consider the same four graph types as in Section \ref{sec:simu} with the number of variable $p=200$, the edge parameter $b=-0.6$ for band and block graphs, the nominal FDR level $q=0.2$, and the sample size $n$ varied between $1500$ and $4000$. 
We use $100$ replications for each setting.
For each replication, we store the resulting FDP and TPP of the estimated graph $\widehat E$. 
We then take average of the $100$ FDP and TPP realizations to obtain the empirical FDR and power.
The simulation results are shown in Figure \ref{mFDRFig:VaryN}.
\begin{figure}[t!]
	\centering
	\includegraphics[width=13cm, height=13cm]{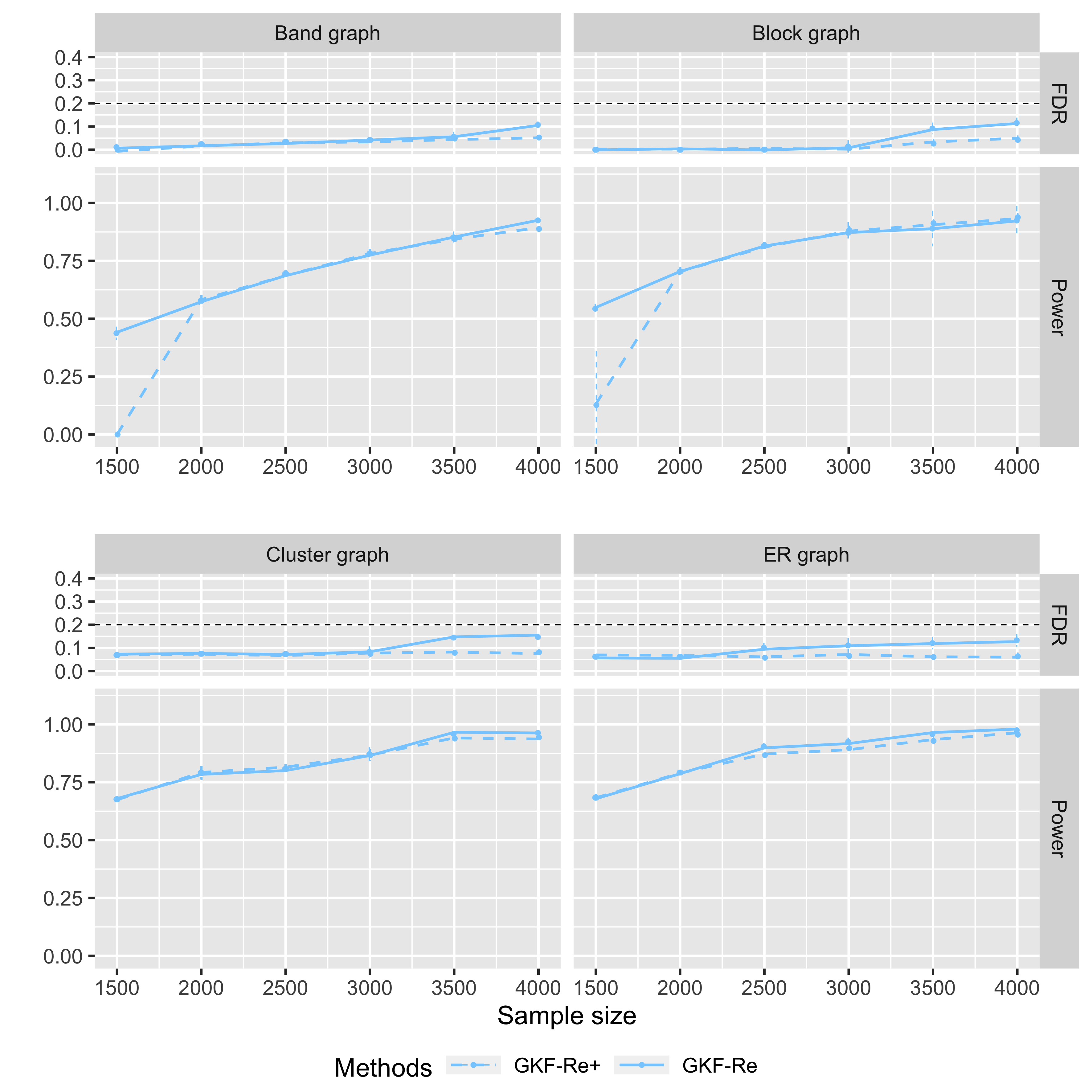}
	\caption{Simulation results: the empirical FDR and power of GKF-Re+ and GKF-Re on the four graph types when varying the sample size $n$, while the nominal FDR level $q=0.2$ and the edge parameter $b=-0.6$ for band and block graphs. The dashed vertical bars indicate plus/minus one empirical standard deviation of the FDP and TPP.}
	\label{mFDRFig:VaryN}
\end{figure}

We can see from Figure \ref{mFDRFig:VaryN} that GKF-Re and GKF-Re+ have similar empirical FDR and power when the sample size $n$ is large enough. This is because that  in this case, the extra $a$ in the optimization problems \eqref{mFDRoptAND} and \eqref{mFDRoptOR} does not matter a lot. However, when the the sample size $n$ is not large enough, the empirical power of GKF-Re+ dramatically decreases whereas GKF-Re is still good. In such case, the extra $a$ in the optimization problems \eqref{mFDRoptAND} and \eqref{mFDRoptOR} really matters, because a vector $\widehat{T}$ with $\widehat{V}_i^{-}(\widehat{T}_i)=0$ for all $i\in[p]$ is always feasible without $a$ in the left side of the optimization problem, but it may be infeasible when $a$ is in the optimization problem.

GKF-Re successfully controls the FDR in the above simulations. However, we emphasize that it only possesses the mFDR control guarantee, and it can lose FDR control in some settings. We now present an example showing that GKF-Re loses FDR control. The graph used in this simulation is an empty graph. Specifically, we let $\Omega$ be the identity matrix with dimension $p=20$. We set the nominal FDR level $q=0.2$ and the sample size $n = 200 $. The empirical FDR are computed based on different number of replications in order to show the convergence. The simulation results are presented in Figure \ref{mFDREmptyGraph}. One can see that GKF-Re loses FDR control whereas GKF-Re+ still controls the FDR.
\begin{figure}[t!]
	\centering
	\includegraphics[width=8cm, height=8cm]{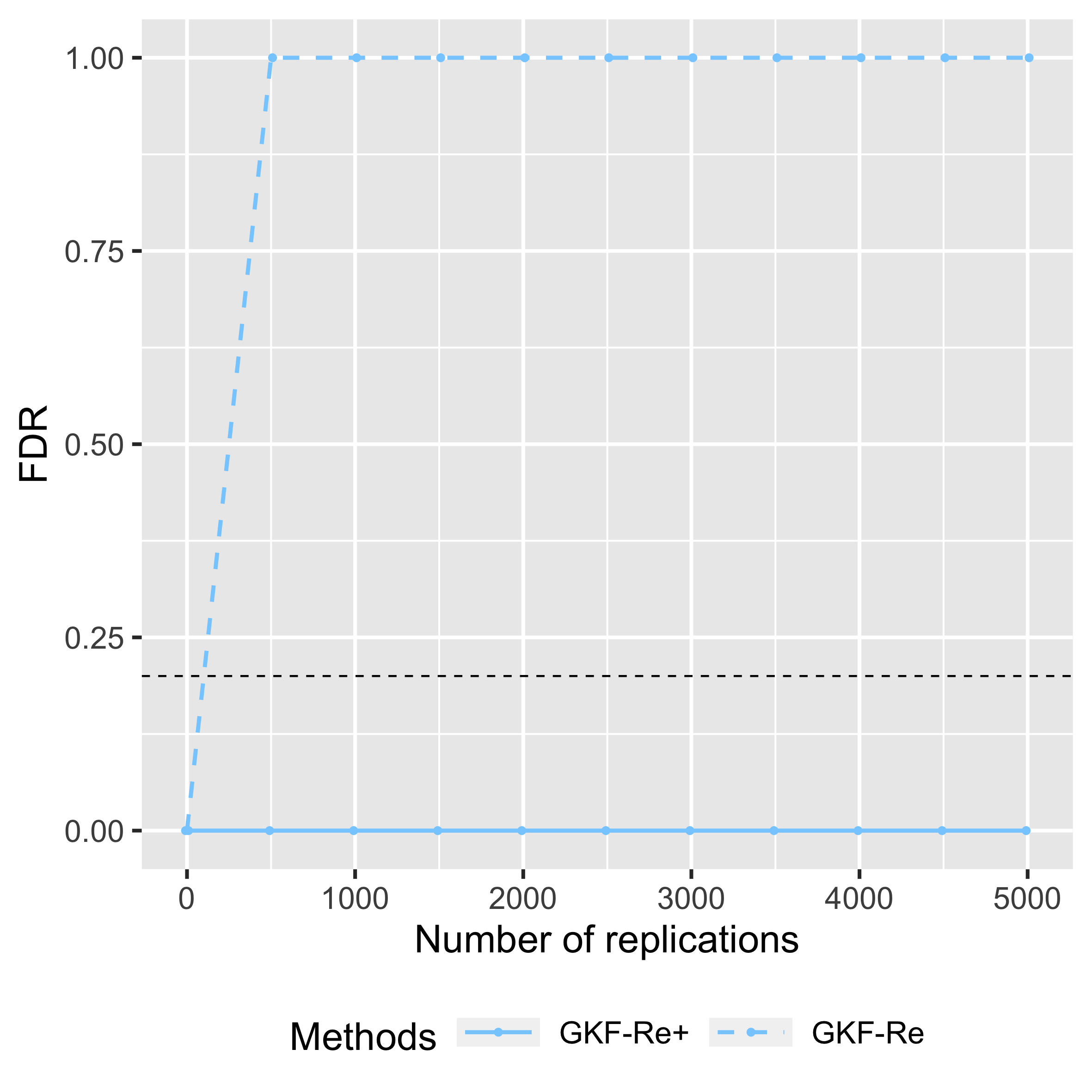}
	\caption{Simulation results: the empirical FDR of GKF-Re+ and GKF-Re on the empty graph in setting where the dimension $p=20$, the sample size $n=200$ and the nominal FDR level $q=0.2$. }
	\label{mFDREmptyGraph}
\end{figure}

%%%%%%%%%%%%%%%%%%%%%%%%%%%%%%%%%%%%%%%%%%%%%%%%%%%%%%%%%%%%%%%%%%%%%%%%%%%%%%%%%%%%%%%%%%%%%%
\subsection{Normal QQ-plot for real data} \label{appendix: QQplot}

In this section, we show the corresponding normal QQ-plot for the real data used in Section \ref{SecRealData}. The original real data can be downloaded from \citep[S2 Appendix: \url{https://journals.plos.org/ploscompbiol/article/file?type=supplementary&id=info:doi/10.1371/journal.pcbi.1006369.s002}]{zheng2017massively}. As described in Section \ref{SecRealData}, we first preprocess the gene data by performing a $\log_2$(counts+1) transformation and a nonparanormal transformation \citep{liu2009nonparanormal}, then we take the $50$ genes with the largest sample variances. As it is hard to verify multivariate Gaussianity, we look at the normal QQ-plots of the $50$ genes to verify marginal Gaussianity. We emphasize that this is only an implication of the multivariate Gaussian assumption, and it does not fully verify multivariate Gaussianity. To avoid too many plots, we only present the QQ-plots of genes $1,10,20,30,40$ and $50$ (the indices are determined by the variance) in Figure \ref{QQ-plot}. One can see that the marginal Gaussianity seems reasonable here.
\begin{figure}[t!]
	\centering
	\includegraphics[width=8cm, height=8cm]{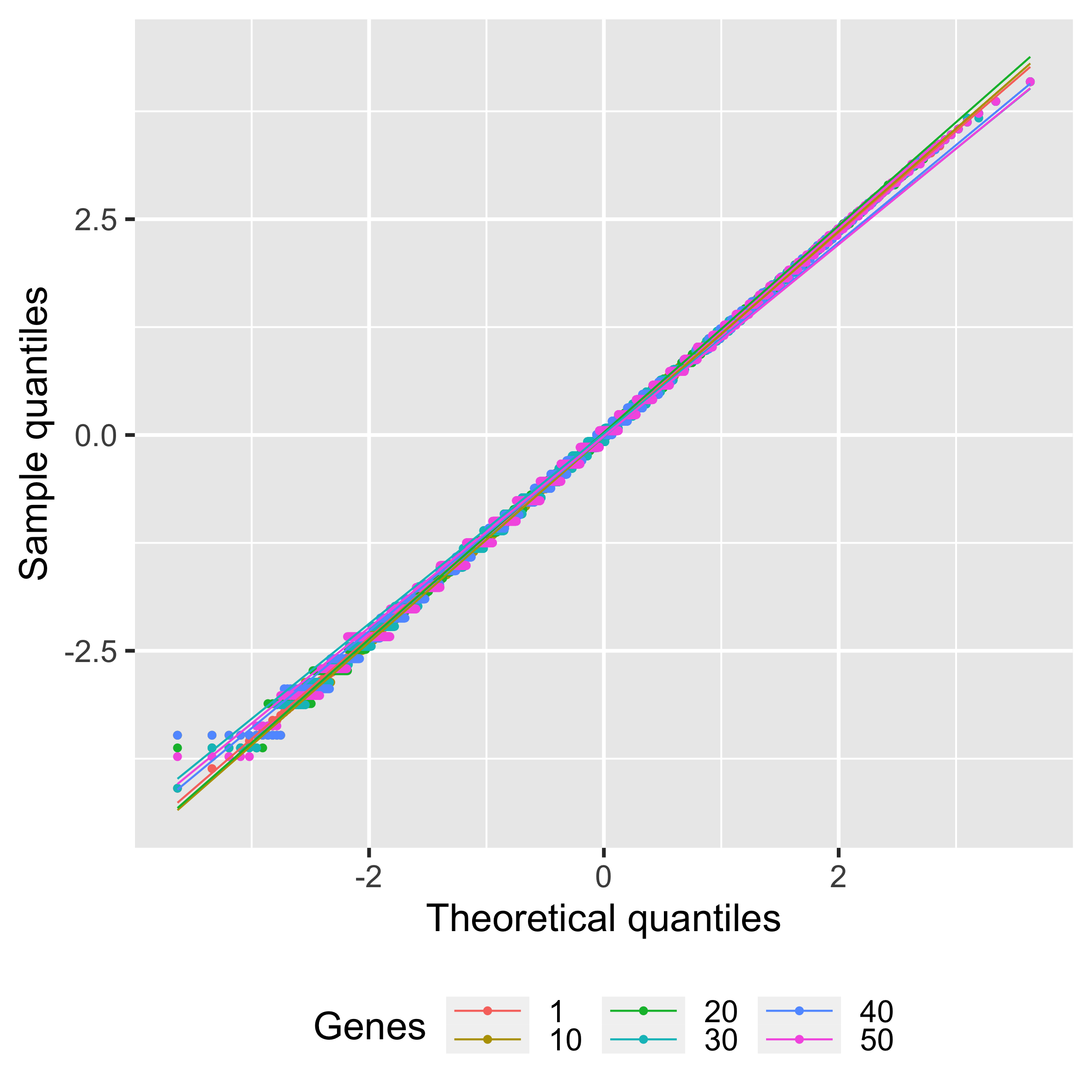}
	\caption{Real data: normal QQ-plots of gene $1,10,20,30,40$ and $50$.}
	\label{QQ-plot}
\end{figure}

%% file: Main.bbl
\begin{thebibliography}{}

\bibitem[\protect\citeauthoryear{Ahmed and Xing}{Ahmed and
  Xing}{2009}]{ahmed2009recovering}
Ahmed, A. and E.~P. Xing (2009).
\newblock Recovering time-varying networks of dependencies in social and
  biological studies.
\newblock {\em Proceedings of the National Academy of Sciences\/}~{\em
  106\/}(29), 11878--11883.

\bibitem[\protect\citeauthoryear{Baker}{Baker}{2016}]{baker20161}
Baker, M. (2016).
\newblock 1,500 scientists lift the lid on reproducibility.
\newblock {\em Nature News\/}~{\em 533\/}(7604), 452.

\bibitem[\protect\citeauthoryear{Barber and Cand{\`e}s}{Barber and
  Cand{\`e}s}{2015}]{barber2015controlling}
Barber, R.~F. and E.~J. Cand{\`e}s (2015).
\newblock Controlling the false discovery rate via knockoffs.
\newblock {\em The Annals of Statistics\/}~{\em 43\/}(5), 2055--2085.

\bibitem[\protect\citeauthoryear{Barber and Cand{\`e}s}{Barber and
  Cand{\`e}s}{2019}]{barber2016knockoff}
Barber, R.~F. and E.~J. Cand{\`e}s (2019).
\newblock A knockoff filter for high-dimensional selective inference.
\newblock {\em The Annals of Statistics\/}~{\em 47\/}(5), 2504--2537.

\bibitem[\protect\citeauthoryear{Begley and Ellis}{Begley and
  Ellis}{2012}]{begley2012drug}
Begley, C.~G. and L.~M. Ellis (2012).
\newblock Drug development: Raise standards for preclinical cancer research.
\newblock {\em Nature\/}~{\em 483\/}(7391), 531.

\bibitem[\protect\citeauthoryear{Benjamini and Hochberg}{Benjamini and
  Hochberg}{1995}]{benjamini1995controlling}
Benjamini, Y. and Y.~Hochberg (1995).
\newblock Controlling the false discovery rate: a practical and powerful
  approach to multiple testing.
\newblock {\em Journal of the Royal Statistical Society. Series B
  (Methodological)\/}~{\em 57\/}(1), 289--300.

\bibitem[\protect\citeauthoryear{Benjamini and Yekutieli}{Benjamini and
  Yekutieli}{2001}]{benjamini2001control}
Benjamini, Y. and D.~Yekutieli (2001).
\newblock The control of the false discovery rate in multiple testing under
  dependency.
\newblock {\em The Annals of Statistics\/}~{\em 29\/}(4), 1165--1188.

\bibitem[\protect\citeauthoryear{Cai, Liu, and Luo}{Cai
  et~al.}{2011}]{cai2011constrained}
Cai, T., W.~Liu, and X.~Luo (2011).
\newblock A constrained $\ell_1$ minimization approach to sparse precision
  matrix estimation.
\newblock {\em Journal of the American Statistical Association\/}~{\em
  106\/}(494), 594--607.

\bibitem[\protect\citeauthoryear{Cand{\`e}s, Fan, Janson, and Lv}{Cand{\`e}s
  et~al.}{2018}]{candes2016panning}
Cand{\`e}s, E., Y.~Fan, L.~Janson, and J.~Lv (2018).
\newblock Panning for gold: {`model-X'} knockoffs for high dimensional
  controlled variable selection.
\newblock {\em Journal of the Royal Statistical Society. Series B
  (Methodological)\/}~{\em 80\/}(3), 551--577.

\bibitem[\protect\citeauthoryear{Clarke and Hall}{Clarke and
  Hall}{2009}]{clarke2009robustness}
Clarke, S. and P.~Hall (2009).
\newblock Robustness of multiple testing procedures against dependence.
\newblock {\em The Annals of Statistics\/}~{\em 37\/}(1), 332--358.

\bibitem[\protect\citeauthoryear{Dai and Barber}{Dai and
  Barber}{2016}]{dai2016knockoff}
Dai, R. and R.~Barber (2016).
\newblock The knockoff filter for {FDR} control in group-sparse and multitask
  regression.
\newblock In {\em International Conference on Machine Learning}, pp.\
  1851--1859.

\bibitem[\protect\citeauthoryear{d'Aspremont, Banerjee, and
  El~Ghaoui}{d'Aspremont et~al.}{2008}]{d2008first}
d'Aspremont, A., O.~Banerjee, and L.~El~Ghaoui (2008).
\newblock First-order methods for sparse covariance selection.
\newblock {\em SIAM Journal on Matrix Analysis and Applications\/}~{\em
  30\/}(1), 56--66.

\bibitem[\protect\citeauthoryear{Drton and Maathuis}{Drton and
  Maathuis}{2017}]{drton2017structure}
Drton, M. and M.~H. Maathuis (2017).
\newblock Structure learning in graphical modeling.
\newblock {\em Annual Review of Statistics and Its Application\/}~{\em 4},
  365--393.

\bibitem[\protect\citeauthoryear{Drton and Perlman}{Drton and
  Perlman}{2007}]{drton2007multiple}
Drton, M. and M.~D. Perlman (2007).
\newblock Multiple testing and error control in {Gaussian} graphical model
  selection.
\newblock {\em Statistical Science\/}~{\em 22\/}(3), 430--449.

\bibitem[\protect\citeauthoryear{Fallat, Lauritzen, Sadeghi, Uhler, Wermuth,
  and Zwiernik}{Fallat et~al.}{2017}]{fallat2017total}
Fallat, S., S.~Lauritzen, K.~Sadeghi, C.~Uhler, N.~Wermuth, and P.~Zwiernik
  (2017).
\newblock Total positivity in {Markov} structures.
\newblock {\em The Annals of Statistics\/}~{\em 45\/}(3), 1152--1184.

\bibitem[\protect\citeauthoryear{Fan, Feng, and Wu}{Fan
  et~al.}{2009}]{fan2009network}
Fan, J., Y.~Feng, and Y.~Wu (2009).
\newblock Network exploration via the adaptive {LASSO} and {SCAD} penalties.
\newblock {\em The Annals of Applied Statistics\/}~{\em 3\/}(2), 521.

\bibitem[\protect\citeauthoryear{Fan, Demirkaya, Li, and Lv}{Fan
  et~al.}{2020}]{fan2020rank}
Fan, Y., E.~Demirkaya, G.~Li, and J.~Lv (2020).
\newblock Rank: large-scale inference with graphical nonlinear knockoffs.
\newblock {\em Journal of the American Statistical Association\/}~{\em
  115\/}(529), 362--379.

\bibitem[\protect\citeauthoryear{Friedman, Hastie, and Tibshirani}{Friedman
  et~al.}{2008}]{friedman2008sparse}
Friedman, J., T.~Hastie, and R.~Tibshirani (2008).
\newblock Sparse inverse covariance estimation with the graphical {Lasso}.
\newblock {\em Biostatistics\/}~{\em 9\/}(3), 432--441.

\bibitem[\protect\citeauthoryear{Friedman, Hastie, and Tibshirani}{Friedman
  et~al.}{2010}]{friedman2010regularization}
Friedman, J., T.~Hastie, and R.~Tibshirani (2010).
\newblock Regularization paths for generalized linear models via coordinate
  descent.
\newblock {\em Journal of Statistical Software\/}~{\em 33\/}(1), 1.

\bibitem[\protect\citeauthoryear{Gimenez and Zou}{Gimenez and
  Zou}{2019}]{gimenez2018improving}
Gimenez, J.~R. and J.~Zou (2019).
\newblock Improving the stability of the knockoff procedure: Multiple
  simultaneous knockoffs and entropy maximization.
\newblock In {\em The 22nd International Conference on Artificial Intelligence
  and Statistics}, pp.\  2184--2192.

\bibitem[\protect\citeauthoryear{Giudici and Spelta}{Giudici and
  Spelta}{2016}]{giudici2016graphical}
Giudici, P. and A.~Spelta (2016).
\newblock Graphical network models for international financial flows.
\newblock {\em Journal of Business \& Economic Statistics\/}~{\em 34\/}(1),
  128--138.

\bibitem[\protect\citeauthoryear{Huang and Janson}{Huang and
  Janson}{2020}]{huang2020relaxing}
Huang, D. and L.~Janson (2020).
\newblock Relaxing the assumptions of knockoffs by conditioning.
\newblock {\em The Annals of Statistics\/}~{\em 48\/}(5), 3021--3042.

\bibitem[\protect\citeauthoryear{Ioannidis}{Ioannidis}{2005}]{ioannidis2005most}
Ioannidis, J.~P. (2005).
\newblock Why most published research findings are false.
\newblock {\em PLoS Medicine\/}~{\em 2\/}(8), e124.

\bibitem[\protect\citeauthoryear{Janson and Su}{Janson and
  Su}{2016}]{janson2016familywise}
Janson, L. and W.~Su (2016).
\newblock Familywise error rate control via knockoffs.
\newblock {\em Electronic Journal of Statistics\/}~{\em 10\/}(1), 960--975.

\bibitem[\protect\citeauthoryear{Kalisch, Fellinghauer, Grill, Maathuis,
  Mansmann, B{\"u}hlmann, and Stucki}{Kalisch
  et~al.}{2010}]{kalisch2010understanding}
Kalisch, M., B.~A. Fellinghauer, E.~Grill, M.~H. Maathuis, U.~Mansmann,
  P.~B{\"u}hlmann, and G.~Stucki (2010).
\newblock Understanding human functioning using graphical models.
\newblock {\em BMC Medical Research Methodology\/}~{\em 10\/}(1), 14.

\bibitem[\protect\citeauthoryear{Karlin and Rinott}{Karlin and
  Rinott}{1983}]{karlin1983m}
Karlin, S. and Y.~Rinott (1983).
\newblock M-matrices as covariance matrices of multinormal distributions.
\newblock {\em Linear algebra and its applications\/}~{\em 52}, 419--438.

\bibitem[\protect\citeauthoryear{Katsevich and Ramdas}{Katsevich and
  Ramdas}{2020}]{katsevich2020simultaneous}
Katsevich, E. and A.~Ramdas (2020).
\newblock Simultaneous high-probability bounds on the false discovery
  proportion in structured, regression and online settings.
\newblock {\em The Annals of Statistics\/}~{\em 48\/}(6), 3465--3487.

\bibitem[\protect\citeauthoryear{Katsevich and Sabatti}{Katsevich and
  Sabatti}{2019}]{katsevich2017multilayer}
Katsevich, E. and C.~Sabatti (2019).
\newblock Multilayer knockoff filter: Controlled variable selection at multiple
  resolutions.
\newblock {\em The Annals of Applied Statistics\/}~{\em 13\/}(1), 1--33.

\bibitem[\protect\citeauthoryear{Lafit, Tuerlinckx, Myin-Germeys, and
  Ceulemans}{Lafit et~al.}{2019}]{lafit2019partial}
Lafit, G., F.~Tuerlinckx, I.~Myin-Germeys, and E.~Ceulemans (2019).
\newblock A partial correlation screening approach for controlling the false
  positive rate in sparse {Gaussian} graphical models.
\newblock {\em Scientific reports\/}~{\em 9\/}(1), 1--24.

\bibitem[\protect\citeauthoryear{Lauritzen}{Lauritzen}{1996}]{lauritzen1996graphical}
Lauritzen, S.~L. (1996).
\newblock {\em Graphical Models}.
\newblock Clarendon Press.

\bibitem[\protect\citeauthoryear{Lee, Sobczyk, and Bogdan}{Lee
  et~al.}{2019}]{lee2019structure}
Lee, S., P.~Sobczyk, and M.~Bogdan (2019).
\newblock Structure learning of {Gaussian} {Markov} random fields with false
  discovery rate control.
\newblock {\em Symmetry\/}~{\em 11\/}(10), 1311.

\bibitem[\protect\citeauthoryear{Liu, Lafferty, and Wasserman}{Liu
  et~al.}{2009}]{liu2009nonparanormal}
Liu, H., J.~Lafferty, and L.~Wasserman (2009).
\newblock The nonparanormal: Semiparametric estimation of high dimensional
  undirected graphs.
\newblock {\em Journal of Machine Learning Research\/}~{\em 10\/}(10).

\bibitem[\protect\citeauthoryear{Liu and Wang}{Liu and
  Wang}{2017}]{liu2017tiger}
Liu, H. and L.~Wang (2017).
\newblock Tiger: A tuning-insensitive approach for optimally estimating
  gaussian graphical models.
\newblock {\em Electronic Journal of Statistics\/}~{\em 11\/}(1), 241--294.

\bibitem[\protect\citeauthoryear{Liu and Rigollet}{Liu and
  Rigollet}{2019}]{liu2019power}
Liu, J. and P.~Rigollet (2019).
\newblock Power analysis of knockoff filters for correlated designs.
\newblock In {\em Advances in Neural Information Processing Systems}, pp.\
  15446--15455.

\bibitem[\protect\citeauthoryear{Liu}{Liu}{2013}]{liu2013gaussian}
Liu, W. (2013).
\newblock {Gaussian} graphical model estimation with false discovery rate
  control.
\newblock {\em The Annals of Statistics\/}~{\em 41\/}(6), 2948--2978.

\bibitem[\protect\citeauthoryear{Maathuis, Drton, Lauritzen, and
  Wainwright}{Maathuis et~al.}{2019}]{HandbookGraphicalModels19}
Maathuis, M., M.~Drton, S.~Lauritzen, and M.~Wainwright (2019).
\newblock {\em Handbook of Graphical Models}.
\newblock Handbooks of Modern Statistical Methods. Chapman\&Hall/CRC.

\bibitem[\protect\citeauthoryear{Meinshausen and B{\"u}hlmann}{Meinshausen and
  B{\"u}hlmann}{2006}]{meinshausen2006high}
Meinshausen, N. and P.~B{\"u}hlmann (2006).
\newblock High-dimensional graphs and variable selection with the {Lasso}.
\newblock {\em The Annals of Statistics\/}~{\em 34\/}(3), 1436--1462.

\bibitem[\protect\citeauthoryear{Raskutti, Yu, Wainwright, and
  Ravikumar}{Raskutti et~al.}{2009}]{raskutti2009model}
Raskutti, G., B.~Yu, M.~J. Wainwright, and P.~K. Ravikumar (2009).
\newblock Model selection in {Gaussian} graphical models: High-dimensional
  consistency of $\ell_1$-regularized {MLE}.
\newblock In {\em Advances in Neural Information Processing Systems}, pp.\
  1329--1336.

\bibitem[\protect\citeauthoryear{Rothman, Bickel, Levina, and Zhu}{Rothman
  et~al.}{2008}]{rothman2008sparse}
Rothman, A.~J., P.~J. Bickel, E.~Levina, and J.~Zhu (2008).
\newblock Sparse permutation invariant covariance estimation.
\newblock {\em Electronic Journal of Statistics\/}~{\em 2}, 494--515.

\bibitem[\protect\citeauthoryear{Shin, Fauman, Petersen, Krumsiek, Santos,
  Huang, Arnold, Erte, Forgetta, Yang, et~al.}{Shin
  et~al.}{2014}]{shin2014atlas}
Shin, S.-Y., E.~B. Fauman, A.-K. Petersen, J.~Krumsiek, R.~Santos, J.~Huang,
  M.~Arnold, I.~Erte, V.~Forgetta, T.-P. Yang, et~al. (2014).
\newblock An atlas of genetic influences on human blood metabolites.
\newblock {\em Nature genetics\/}~{\em 46\/}(6), 543.

\bibitem[\protect\citeauthoryear{Tibshirani}{Tibshirani}{1996}]{tibshirani1996regression}
Tibshirani, R. (1996).
\newblock Regression shrinkage and selection via the {Lasso}.
\newblock {\em Journal of the Royal Statistical Society: Series B
  (Methodological)\/}~{\em 58\/}(1), 267--288.

\bibitem[\protect\citeauthoryear{Wang and Janson}{Wang and
  Janson}{2020}]{wang2020power}
Wang, W. and L.~Janson (2020).
\newblock A power analysis of the conditional randomization test and knockoffs.
\newblock arXiv:2010.02304.

\bibitem[\protect\citeauthoryear{Weinstein, Barber, and Cand{\`e}s}{Weinstein
  et~al.}{2017}]{weinstein2017power}
Weinstein, A., R.~Barber, and E.~Cand{\`e}s (2017).
\newblock A power and prediction analysis for knockoffs with {Lasso}
  statistics.
\newblock arXiv:1712.06465.

\bibitem[\protect\citeauthoryear{Weinstein, Su, Bogdan, Barber, and
  Cand{\`e}s}{Weinstein et~al.}{2020}]{weinstein2020power}
Weinstein, A., W.~J. Su, M.~Bogdan, R.~F. Barber, and E.~J. Cand{\`e}s (2020).
\newblock A power analysis for knockoffs with the {Lasso}
  coefficient-difference statistic.
\newblock arXiv:2007.15346.

\bibitem[\protect\citeauthoryear{Yu, Kaufmann, and Lederer}{Yu
  et~al.}{2019}]{YuKaufmannLederer19}
Yu, L., T.~Kaufmann, and J.~Lederer (2019).
\newblock False discovery rates in biological networks.
\newblock arXiv:1907.03808.

\bibitem[\protect\citeauthoryear{Yuan and Lin}{Yuan and
  Lin}{2007}]{yuan2007model}
Yuan, M. and Y.~Lin (2007).
\newblock Model selection and estimation in the {Gaussian} graphical model.
\newblock {\em Biometrika\/}~{\em 94\/}(1), 19--35.

\bibitem[\protect\citeauthoryear{Zhang, Ren, and Chen}{Zhang
  et~al.}{2018}]{zhang2018silggm}
Zhang, R., Z.~Ren, and W.~Chen (2018).
\newblock {SILGGM}: An extensive {R} package for efficient statistical
  inference in large-scale gene networks.
\newblock {\em PLoS Computational Biology\/}~{\em 14\/}(8), e1006369.

\bibitem[\protect\citeauthoryear{Zheng, Terry, Belgrader, Ryvkin, Bent, Wilson,
  Ziraldo, Wheeler, McDermott, Zhu, et~al.}{Zheng
  et~al.}{2017}]{zheng2017massively}
Zheng, G.~X., J.~M. Terry, P.~Belgrader, P.~Ryvkin, Z.~W. Bent, R.~Wilson,
  S.~B. Ziraldo, T.~D. Wheeler, G.~P. McDermott, J.~Zhu, et~al. (2017).
\newblock Massively parallel digital transcriptional profiling of single cells.
\newblock {\em Nature Communications\/}~{\em 8}, 14049.

\bibitem[\protect\citeauthoryear{Zheng, Zhou, Guo, and Li}{Zheng
  et~al.}{2018}]{zheng2018recovering}
Zheng, Z., J.~Zhou, X.~Guo, and D.~Li (2018).
\newblock Recovering the graphical structures via knockoffs.
\newblock {\em Procedia Computer Science\/}~{\em 129}, 201--207.

\bibitem[\protect\citeauthoryear{Zou and Hastie}{Zou and
  Hastie}{2005}]{zou2005regularization}
Zou, H. and T.~Hastie (2005).
\newblock Regularization and variable selection via the elastic net.
\newblock {\em Journal of the Royal Statistical Society: Series B (Statistical
  Methodology)\/}~{\em 67\/}(2), 301--320.

\end{thebibliography}
